\documentclass[11pt]{article}

\usepackage{amsmath, amsthm}
\usepackage{amsmath, amsfonts}
\usepackage{amsmath, amssymb}
\usepackage{amsmath}
\usepackage{graphics}
\usepackage{graphicx}
\usepackage{color}
\usepackage{hyperref}
\usepackage[all]{xy}

\newtheorem{theorem}{Theorem}[subsection]
\newtheorem{definition}[theorem]{Definition}
\newtheorem{definition-lemma}[theorem]{Definition/Lemma}
\newtheorem{definition-explanation}[theorem]{Definition/Explanation}
\newtheorem{explanation-definition}[theorem]{Explanation/Definition}
\newtheorem{definition-fact}[theorem]{Definition/Fact}
\newtheorem{definition-notation}[theorem]{Definition/Notation}
\newtheorem{definition-conjecture}[theorem]{Definition/Conjecture}
\newtheorem{lemma}[theorem]{Lemma}
\newtheorem{lemma-definition}[theorem]{Lemma/Definition}

\newtheorem{remark}[theorem]{\it Remark}
\newtheorem{remark-notation}[theorem]{\it Remark/Notation}
\newtheorem{remark-convention}[theorem]{\it Remark/Convention}

\newtheorem{application-lemma}[theorem]{Application/Lemma}

\newtheorem{example}[theorem]{Example}
\newtheorem{example-definition}[theorem]{Example/Definition}

\newtheorem{definition-prototype}[theorem]{Definition-Prototype}

\newtheorem{terminology}[theorem]{\it Terminology}

\numberwithin{equation}{subsection}

\newtheorem{stheorem}{Theorem}[section]
\newtheorem{sdefinition}[stheorem]{Definition}
\newtheorem{sdefinition-lemma}[stheorem]{Definition/Lemma}
\newtheorem{sdefinition-explanation}[stheorem]{Definition/Explanation}
\newtheorem{sexplanation-definition}[stheorem]{Explanation/Definition}
\newtheorem{sdefinition-fact}[stheorem]{Definition/Fact}
\newtheorem{sdefinition-notation}[stheorem]{Definition/Notation}
\newtheorem{sdefinition-conjecture}[stheorem]{Definition/Conjecture}

\newtheorem{slemma-definition}[stheorem]{Lemma/Definition}

\newtheorem{sremark}[stheorem]{\it Remark}
\newtheorem{sremark-notation}[stheorem]{\it Remark/Notation}
\newtheorem{sremark-convention}[stheorem]{\it Remark/Convention}

\newtheorem{sapplication-lemma}[stheorem]{Application/Lemma}

\newtheorem{sexample-definition}[stheorem]{Example/Definition}

\newtheorem{sdefinition-prototype}[stheorem]{Definition-Prototype}

\newtheorem{sstheorem}{Theorem}[subsubsection]
\newtheorem{ssdefinition}[sstheorem]{Definition}
\newtheorem{ssdefinition-lemma}[sstheorem]{Definition/Lemma}
\newtheorem{ssdefinition-explanation}[sstheorem]{Definition/Explanation}
\newtheorem{ssexplanation-definition}[sstheorem]{Explanation/Definition}
\newtheorem{ssdefinition-fact}[sstheorem]{Definition/Fact}
\newtheorem{ssdefinition-notation}[sstheorem]{Definition/Notation}
\newtheorem{ssdefinition-conjecture}[sstheorem]{Definition/Conjecture}
\newtheorem{sslemma}[sstheorem]{Lemma}
\newtheorem{sslemma-definition}[sstheorem]{Lemma/Definition}
\newtheorem{ssproposition}[sstheorem]{Proposition}
\newtheorem{sscorollary}[sstheorem]{Corollary}
\newtheorem{ssremark}[sstheorem]{\it Remark}
\newtheorem{ssremark-notation}[sstheorem]{\it Remark/Notation}
\newtheorem{ssremark-convention}[stheorem]{\it Remark/Convention}

\newtheorem{ssapplication-lemma}[sstheorem]{Application/Lemma}

\newtheorem{ssexample}[sstheorem]{Example}
\newtheorem{ssexample-definition}[sstheorem]{Example/Definition}

\newtheorem{ssdefinition-prototype}[sstheorem]{Definition-Prototype}

\newcommand{\Ann}{\mbox{\it Ann}\,}

\newcommand{\Aut}{\mbox{\it Aut}\,}

\newcommand{\Der}{\mbox{\it Der}\,}
  \newcommand{\sDer}{\mbox{\it sDer}\,}

\newcommand{\End}{\mbox{\it End}\,}

\newcommand{\Endsheaf}{\mbox{\it ${\cal E}\!$nd}\,}

\newcommand{\GL}{\mbox{\it GL}}

\newcommand{\Higgsscriptsize}{\mbox{\scriptsize\it  Higgs}}

\newcommand{\Hom}{\mbox{\it Hom}\,}
\newcommand{\Homsheaf}{\mbox{\it ${\cal H}$om}\,}

\newcommand{\Id}{\mbox{\it Id}\,}

\newcommand{\Image}{\mbox{\it Im}\,}

\newcommand{\Ker}{\mbox{\it Ker}\,}

\newcommand{\Lorentzscriptsize}{\mbox{\scriptsize\it Lorentz}}

\newcommand{\ModCategory}{\mbox{\it ${\cal M}$\!od}\,}

\newcommand{\QCohCategory}{\mbox{\it ${\cal Q}$${\cal C}$\!oh}\,}

\newcommand{\Quot}{\mbox{\it Quot}\,}

\newcommand{\SL}{\mbox{\it SL}}

\newcommand{\SO}{\mbox{\it SO}\,}

\newcommand{\Span}{\mbox{\it Span}\,}

\newcommand{\Spin}{\mbox{\it Spin}\,}

\newcommand{\Stab}{\mbox{\it Stab}\,}

\newcommand{\SU}{\mbox{\it SU}}
\newcommand{\Supp}{\mbox{\it Supp}\,}
 \newcommand{\scriptsizeSupp}{\mbox{\scriptsize\it Supp}\,}

\newcommand{\Tr}{\mbox{\it Tr}\,}

\newcommand{\Yukawascriptsize}{\mbox{\scriptsize\it  Yukawa}}

 \newcommand{\scriptsizeeven}{\mbox{\scriptsize\it even}\,}
\newcommand{\fermionscriptsize}{\mbox{\scriptsize\it  fermion}}

\newcommand{\gaugescriptsize}{\mbox{\scriptsize\it  gauge}}
\newcommand{\gl}{\mbox{\it gl}\,}

\newcommand{\nc}{\mbox{\scriptsize\it nc}}

 \newcommand{\scriptsizeodd}{\mbox{\scriptsize\it odd}\,}

\newcommand{\pr}{\mbox{\it pr}}

\newcommand{\redscriptsize}{\mbox{\scriptsize\rm red}\,}

\newcommand{\su}{\mbox{\it su}\,}

\begin{document}

\enlargethispage{24cm}

\begin{titlepage}

$ $

\vspace{-1.5cm} 

\noindent\hspace{-1cm}
\parbox{6cm}{\small October 2014}\
   \hspace{6cm}\
   \parbox[t]{6cm}{yymm.nnnn [hep-th] \\
                D(11.2): fermionic D-brane
				}

\vspace{2cm}

\centerline{\large\bf
 D-branes and Azumaya/matrix noncommutative differential geometry,}
\vspace{1ex}
\centerline{\large\bf
 II: Azumaya/matrix supermanifolds and differentiable maps therefrom}
 \vspace{1ex}
 \centerline{\large\bf
 -- with a view toward dynamical fermionic D-branes in string theory}

\bigskip

\vspace{3em}

\centerline{\large
  Chien-Hao Liu    
            \hspace{1ex} and \hspace{1ex}
  Shing-Tung Yau
}

\vspace{4em}

\begin{quotation}
\centerline{\bf Abstract}

\vspace{0.3cm}

\baselineskip 12pt  
{\small  
 In this Part II of D(11), 
 we introduce new objects: super-$C^k$-schemes and Azumaya super-$C^k$-manifolds with a fundamental module
  (or synonymously matrix super-$C^k$-manifolds with a fundamental module),
  and extend the study in D(11.1) ([L-Y3], arXiv:1406.0929 [math.DG]) to define the notion of
  `differentiable maps from an Azumaya/matrix supermanifold with a fundamental module
   to a real manifold or supermanifold'.
 This allows us to introduce the notion of `fermionic D-branes'  in two different styles,
  one parallels Ramond-Neveu-Schwarz fermionic string and the other Green-Schwarz fermionic string.
 A more detailed discussion on
   the Higgs mechanism on dynamical D-branes in our setting,
      taking maps from the D-brane world-volume to the space-time in question
	             and/or sections of the Chan-Paton bundle on the D-brane world-volume
	  as Higgs fields,
  is also given for the first time in the D-project.
 Finally note that mathematically  string theory begins with the notion of a differentiable map
  from a string world-sheet (a $2$-manifold) to a target space-time (a real manifold).
 In comparison to this,
  D(11.1) and the current D(11.2) together bring us to the same starting point
  for studying D-branes in string theory as dynamical objects.
 } 
\end{quotation}

\vspace{9em}

\baselineskip 12pt
{\footnotesize
} 

 \bigskip

\noindent {\small MSC number 2010: 
  81T30, 58A40,14A22; 
  58A50, 16S50, 51K10, 46L87, 81T60, 81T75, 81V22.
} 

\bigskip

\baselineskip 10pt
{\scriptsize
\noindent{\bf Acknowledgements.}
We thank
 Cumrun Vafa
       for lectures and discussions that influence our understanding of string theory.
C.-H.L.\ thanks in addition
 Harald Dorn for illuminations on nonabelian Dirac-Born-Infeld action for coincident D-branes;
 Gregory Moore, Cumrun Vafa
   for illuminations on Higgs fields on D-branes;
 Dennis Westra
   for thesis that influences his understanding of super-algebraic geometry
   in line with Grothendieck's Algebraic Geometry;
 Murad Alim, Ga\"{e}tan Borot, Daniel Freed, Siu-Cheong Lau, Si Li, Baosen Wu
   for discussions on issues beyond;
 Alison Miller, Freed, Vafa
   for topic and basic courses, fall 2014;
 Gimnazija Kranj Symphony Orchestra for work of Nikolay Rimsky-Korsakov
   that accompanies the typing of the notes;
 Ling-Miao Chou
   for discussions on electrodynamics, comments on illustrations, and moral support.
 D(11.1) and D(11.2) together bring this D-project to another phase;
 for that, special thanks also to
 Si Li, Ruifang Song
   for the bi-weekly Saturday D-brane Working Seminar, spring 2008,
   that gave him another best time at Harvard and tremendous momentum to the project.
The project is supported by NSF grants DMS-9803347 and DMS-0074329.
} 
1
\end{titlepage}

\newpage

\begin{titlepage}

$ $

\vspace{12em}

\centerline{\small\it
 Chien-Hao Liu dedicates this note}
\centerline{\small\it
 to his another advisor Prof.\ Orlando Alvarez}
\centerline{\small\it
  during his Berkeley and Miami years,}
\centerline{\small\it
 who gave him the first lecture on D-branes}
\centerline{\small\it
  and brought him to the amazing world of stringy dualities;}
\centerline{\small\it
 and to Prof.\ Rafael Nepomechie, }
\centerline{\small\it
 a pioneer on higher-dimensional extended objects beyond strings,}
\centerline{\small\it
 who gave him the first course on supersymmetry.}
    
\end{titlepage}


\newpage
$ $

\vspace{-3em}

\centerline{\sc
 DB \& NCDG II: Maps from Matrix Supermanifolds, and Fermionic D-Branes
 } %

\vspace{2em}


\begin{flushleft}
{\Large\bf 0. Introduction and outline}
\end{flushleft}
As a preparation to study D-branes in string theory as dynamical objects,
in [L-Y3] (D(11.1))
 we developed
  the notion of `differentiable maps from an Azumaya/matrix manifold with a fundamental module to a real manifold'
   along the line of Algebraic Geometry of Grothendieck
      and synthetic/$C^k$-algebraic differential geometry of Dubuc, Joyce, Kock, Moerdijk, and Reyes;
 and gave examples to illustrate
  how deformations of differentiable maps in our setting capture various behaviors of D-branes.
	
In this continuation of [L-Y3] (D(11.1)),
 we extend the study to the notion of
  `differentiable maps from an Azumaya/matrix supermanifold with a fundamental module to a real supermanifold'.
This allows us to introduce the notion of `fermionic D-branes'  in two different styles,
  one parallels Ramond-Neveu-Schwarz fermionic string and the other Green-Schwarz fermionic string.
[L-Y1] (D(1)), [L-L-S-Y] (D(2)), [L-Y3] (D(11.1)) and the current note (D(11.2))
 together bring
   \begin{itemize}
    \item[{\Large $\cdot$}]
    the study of D-branes in string theory as dynamical objects
     in the context/realm/language of algebraic geometry or differential/symplectic/calibrated geometry,
	 without supersymmetry or with supersymmetry
   \end{itemize}
 all in the equal footing.
This brings us to the door of a new world on dynamical D-branes,
  whose mathematical and stringy-theoretical details have yet to be understood.

\bigskip

The organization of the current note is as follows.
In Sec.~1,
 we brings out the notion of differential maps from an Azumaya/matrix brane with fermions
  in a most primitive setting based on [L-Y3] (D(11.1)).
In Sec.~2 -- Sec.~4,
 we first pave our way toward uniting the new fermionic degrees of freedom into the Azumaya/matrix geometry
   involved,  as is done in the study of supersymmetric quantum field theory to the ordinary geometry,  and
 then
  define the notion of
    `differentiable maps from an Azumaya/matrix supermanifold with a fundamental module to a real manifold'
    in Sec.~4.2  and
  further extend it to the notion of
    `differentiable maps from an Azumaya/matrix supermanifold with a fundamental module to a real supermanifold'
    in Sec.~4.3.
{To} give string-theory-oriented readers a taste of how such notions are put to work for D-branes,
 in Sec.~5.1, we introduce the two notions of fermionic D-branes,
   one following the style of Ramond-Neveu-Schwarz fermionic string and
   the other the style of Green-Schwarz fermionic string;   and
 in Sec.~5.2 we give a more precise discussion on
    the Higgs mechanism on dynamical D-branes in our setting for the first time in this D-project.
Seven years have passed 	since the first note [L-Y1] (D(1)) in this project in progress.
In Sec.~6, we reflect where we are in this journey on D-branes, with a view toward the future.

\bigskip

\bigskip

\noindent
{\bf Convention.}
 Standard notations, terminology, operations, facts in
  (1) superring theory toward superalgebraic geometry;$\,$
  (2) supersymmetry, supersymmetric quantum field theory;$\,$			
  (3) Higgs mechanism, gauge symmetry breaking; grand unification theory$\,$
 can be found respectively in$\,$
  (1) [Wes];$\,$
  (2) [Arg1], [Arg2], [Arg3], [Arg4], [D-E-F-J-K-M-M-W], [Freed], [Freund], [G-G-R-S], [St], [Wei], [W-B];$\,$
  (3) [I-Z], [P-S], [Ry]; [B-H], [Mo], [Ros].$\,$
There are several inequivalent notions of$\,$
 (4) `supermanifold';$\,$
 all intend to (and each does) capture (some part of) the geometry behind supersymmetry in physics.
The setting in$\,$
 (4) [Man], [S-W]$\,$
   is particularly in line with Grothendieck's Algebraic Geometry and hence relevant to us.
 \begin{itemize}
  \item[$\cdot$]
   `{\it field}' in the sense of quantum {\it field} theory (e.g.\ {\it fermionic field}) vs.\
  `{\it field}' as an algebraic structure in ring theory (e.g.\ the {\it field ${\Bbb R}$ of real numbers}).
 
  \item[$\cdot$]
   For clarity, the {\it real line} as a real $1$-dimensional manifold is denoted by ${\Bbb R}^1$,
    while the {\it field of real numbers} is denoted by ${\Bbb R}$.
   Similarly, the {\it complex line} as a complex $1$-dimensional manifold is denoted by ${\Bbb C}^1$,
    while the {\it field of complex numbers} is denoted by ${\Bbb C}$.
	
  \item[$\cdot$]	
  The inclusion `${\Bbb R}\hookrightarrow{\Bbb C}$' is referred to the {\it field extension
   of ${\Bbb R}$ to ${\Bbb C}$} by adding $\sqrt{-1}$, unless otherwise noted.

 \item[$\cdot$]	
  The {\it real $n$-dimensional vector spaces} ${\Bbb R}^{\oplus n}$
      vs.\ the {\it real $n$-manifold} $\,{\Bbb R}^n$; \\
  similarly, the {\it complex $r$-dimensional vector space ${\Bbb C}^{\oplus r}$}
     vs.\ the {\it complex $r$-fold} $\,{\Bbb C}^r$.

 \item[$\cdot$]
  All manifolds are paracompact, Hausdorff, and admitting a (locally finite) partition of unity.
  We adopt the {\it index convention for tensors} from differential geometry.
   In particular, the tuple coordinate functions on an $n$-manifold is denoted by, for example,
   $(y^1,\,\cdots\,y^n)$.
  The up-low index summation convention is always spelled out explicitly when used.


  \item[$\cdot$]
   `{\it differentiable}' = $k$-times differentiable (i.e.\ $C^k$)
         for some $k\in{\Bbb Z}_{\ge 1}\cup{\infty}$;
   `{\it smooth}' $=C^{\infty}$;
   $C^0$ = {\it continuous} by standard convention.

 %
 
  \item[$\cdot$]
   All the $C^k$-rings in this note can be assumed to be
     {\it finitely generated and finitely-near-point determined}.
   In particular, they are {\it finitely generated and germ-determined} ([Du], [M-R: Sec.\ I.4])
    i.e.\ {\it fair} in the sense of [Joy: Sec.~2.4]. (Cf.\ [L-Y4].)
  
  \item[$\cdot$]
   For a $C^k$-subscheme $Z$ of a $C^k$-scheme $Y$,
   $Z_{\redscriptsize}$ denotes its associated reduced subscheme of $Y$
   by modding out all the nilpotent elements in ${\cal O}_Z$.

  \item[$\cdot$]
   The `{\it support}' $\Supp({\cal F})$
    of a quasi-coherent sheaf ${\cal F}$ on a scheme $Y$ in algebraic geometry
     	or on a $C^k$-scheme in $C^k$-algebraic geometry
    means the {\it scheme-theoretical support} of ${\cal F}$
   unless otherwise noted;
    ${\cal I}_Z$ denotes the {\it ideal sheaf} of
    a (resp.\ $C^k$-)subscheme of $Z$ of a (resp.\ $C^k$-)scheme $Y$;
    $l({\cal F})$ denotes the {\it length} of a coherent sheaf ${\cal F}$ of dimension $0$.

  \item[$\cdot$]
   {\it coordinate-function index}, e.g.\ $(y^1,\,\cdots\,,\, y^n)$ for a real manifold
      vs.\  the {\it exponent of a power},
	  e.g.\  $a_0y^r+a_1y^{r-1}+\,\cdots\,+a_{r-1}y+a_r\in {\Bbb R}[y]$.

 
  \item[$\cdot$]
    {\it global section functor} $\varGamma (\,\cdot\,)$ on sheaves
	   vs.\  {\it graph} $\varGamma_f$ of a function $f$.

 
  \item[$\cdot$]
   `{\it d-manifold}$\,$' in the sense of `{\it d}erived manifold'
    vs.\  `{\it D-manifold}$\,$' in the sense of `{\it D}(irichlet)-brane that is supported on a manifold'
	vs.\  `{\it D-manifold}$\,$' in the sense of works [B-V-S1] and [B-V-S2]
	   of Michael Bershadsky, Cumrun Vafa and Vladimir Sadov.

  \item[$\cdot$]
   The current Note D(11.2) continues the study in
	  \begin{itemize}
       \item[] $\mbox{\hspace{2.3em}}$
	   \parbox[t]{34em}{\hspace{-4.86em} [L-Y1]\hspace{1em} {\it	
	     Azumaya-type noncommutative spaces and morphism therefrom: Polchinski's
         D-branes in string theory from Grothendieck's viewpoint},
		 arXiv:0709.1515 [math.AG] (D(1)).
		 } 

	   \medskip	
	   \item[] $\mbox{\hspace{2.3em}}$
	    \parbox[t]{34em}{\hspace{-4.86em} [L-L-S-Y]\hspace{1em}
	     (with Si Li and Ruifang Song),$\;$ {\it	
         Morphisms from Azumaya prestable curves with a fundamental module to a projective variety:
		 Topological D-strings as a master object for curves},
		 arXiv:0809.2121 [math.AG](D(2)).
		 } 
	
	    \medskip
	    \item[]  \hspace{-2em} [L-Y3]\hspace{1em} \parbox[t]{34em}{{\it
    	   D-branes and Azumaya/matrix noncommutative differential geometry,
        I: D-branes as fundamental objects in string theory  and differentiable maps
         from Azumaya/matrix manifolds with a fundamental module to real manifolds},
         arXiv:1406.0929 [math.DG](D(11.1)).
		 }
	  \end{itemize}  	
   Notations and conventions follow these earlier works when applicable.
 \end{itemize}

\bigskip

\newpage

\begin{flushleft}
{\bf Outline}
\end{flushleft}
\nopagebreak
{\small
\baselineskip 12pt  
\begin{itemize}
 \item[0.]
 Introduction.

 \item[1.]
 Differentiable maps from fermionic Azumaya/matrix branes: A primitive setting
   \vspace{-.6ex}
   \begin{itemize}
    \item[{\Large $\cdot$}]
	 Bosonic fields and fermionic fields on the world-volume of coincident D-branes
	
	\item[{\Large $\cdot$}]
     Differentiable maps from matrix branes with fermions
   \end{itemize}

 \item[2.]
 Algebraic geometry over super-$C^k$-rings
 \vspace{-.6ex}
  \begin{itemize}
	\item[2.1]
	 Superrings, modules, and differential calculus on superrings
	
	\item[2.2]
	 Super-$C^k$-rings, modules, and differential calculus on super-$C^k$-rings

    \item[2.3]
	 Super-$C^k$-manifolds, super-$C^k$-ringed spaces, and super-$C^k$-schemes	
	     	
	\item[2.4]
	 Sheaves of modules and differential calculus on a super-$C^k$-scheme
  \end{itemize}

 \item[3.]
 Azumaya/matrix super-$C^k$-manifolds with a fundamental module
 %
 %
  
 \item[4.]
 Differentiable maps from an Azumaya/matrix supermanifold with a fundamental module to a real manifold
  \vspace{-.6ex}
  \begin{itemize}
	\item[4.1]
	A local study: The affine case	
	
    \item[4.2]
    Differentiable maps from an Azumaya/matrix supermanifold with a fundamental module
    to a real manifold	
    \begin{itemize}
     \item[4.2.1]
	  Aspect I [fundamental]: Maps as gluing systems of ring-homomorphisms
	
	 \item[4.2.2]
	  Aspect II: The graph of a differentiable map
	
	 \item[4.2.3]
      Aspect III: From maps to the stack of D0-branes
	
	 \item[4.2.4]
	  Aspect IV: From associated $\GL_r({\Bbb C})$-equivariant maps
    \end{itemize}
			
    \item[4.3]
    Remarks on differentiable maps from a general endomorphism-ringed super-$C^k$-scheme
	with a fundamental module to a real supermanifold
  \end{itemize}

 \item[5.]
 A glimpse of super D-branes, as dynamical objects, and the Higgs mechanism in the current setting
  \begin{itemize}
   \item[5.1]
   Fermionic D-branes as fundamental/dynamical objects in string theory
    \begin{itemize}
     \item[{\Large $\cdot$}]
	  Ramond-Neveu-Schwarz fermionic string and Green-Schwarz fermionic string
          from the viewpoint Grothendieck's Algebraic Geometry
		
	 \item[{\Large $\cdot$}]
	 Fermionic D-branes as dynamical objects \`{a} la RNS or GS fermionic strings
    \end{itemize}
   
   \item[5.2]
   The Higgs mechanism on D-branes vs.\ deformations of maps from a matrix brane
    \begin{itemize}
     \item[5.2.1]
	  The Higgs mechanism in the Glashow-Weinberg-Salam model
	
	 \item[5.2.2]
	  The Higgs mechanism on the matrix brane world-volume
    \end{itemize}
  \end{itemize}

 \item[6.]
  Where we are, and some more new directions
 %
 %
 %

\end{itemize}
} 

\newpage

\section{Differentiable maps from fermionic Azumaya/matrix branes: A primitive setting}

In this section, we give a very primitive view of D-branes with fermions
 directly from the viewpoint of [L-Y3] (D(11.1)) that
 a bottommost ingredient to describe a D-brane in string theory as a dynamical object is
 the notion of a differentiable map from a matrix manifold (cf.\ the world-volume of coincident D-branes)
  to a real manifold (cf.\ the target space-time).

\bigskip
  
\begin{flushleft}
{\bf Bosonic fields and fermionic fields on the world-volume of coincident D-branes\footnote{
                                          The contents of this theme are by now standard textbook materials
										     under the theme on {\it quantization of closed or open strings and
											 the spectrum closed or open strings create on the target}.
										  The concise conceptual highlight here is only meant to make a passage
  										     to relate fields on the world-volume of coincident D-branes to sheaves of modules
											 on a matrix manifold in the next theme.
                                          Unfamiliar mathematicians are referred to
										    [B-B-Sc], [G-S-W], [Po2] 
											for details.
                     }.  
}\end{flushleft}
Fields on the world-volume of a D-brane are created by excitations of oriented open strings
 through their end-points that stick to the D-brane.
Cf.~{\sc Figure}~1-1.
%
%
\begin{figure}[htbp]
 \bigskip
  \centering
  \includegraphics[width=0.80\textwidth]{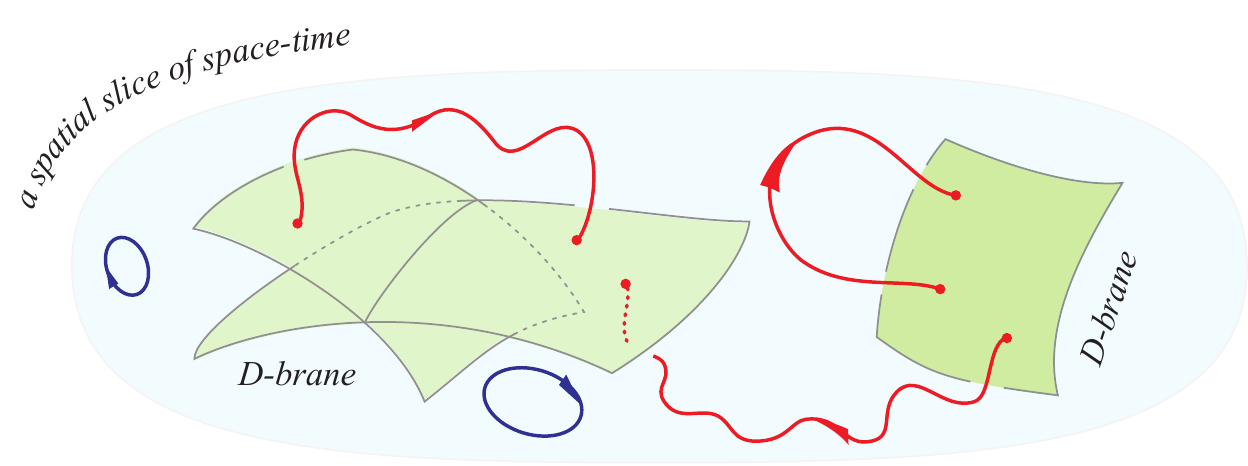}
 
  \bigskip
  \bigskip
 \centerline{\parbox{13cm}{\small\baselineskip 12pt
  {\sc Figure}~1-1.
  Fields on the world-volume of a D-brane are created by excitations of oriented open strings
    through their end-points that stick to the D-branes.
  The dynamics of these fields are dictated by the anomaly-free requirement of the conformal field theory
    on the open-string world-sheet ([Le]).
  }}
\end{figure}
When the open string carries in addition fermionic degrees of freedom,
 the fields it creates on the D-brane world-volume include not only bosonic ones but also fermionic ones.
Each of these fields is associated to an open string state $|\Lambda \rangle$
 from a representation of  the $2$-dimensional superconformal algebra
                                                         (associated to the open-string world-sheet theory)
  that is repackaged to a representation of Lorentz group under the requirement
  that the quantum field theory of these fields on the D-brane world-volume be Lorentz invariant.
  
When $r$-many simple D-branes coincide, the open string spectrum $\{|\Lambda\rangle\,|\, \Lambda\}$
 on the common D-brane world-volume gets enhanced.
There are three possible sectors of the newly re-organized spectrum of open string states:
\begin{itemize}
 \item[(1)]
 {\it From oriented open strings with both end-points stuck to the coincident D-brane} $\,$: \\
  In this case, one has an enhancement
   $$
      |\Lambda\rangle\;  \Rightarrow \;   |\Lambda; i,\bar{j}  \rangle\,,        \hspace{1em}
	   1\le i,j\le r\,.
   $$
  The field $\psi_{\Lambda}$ on the D-brane world-volume that is associated
     to $\{|\Lambda;i,\bar{j} \rangle\,|\, 1\le i, j\le r \}$ as a collection
   takes now $r\times r$-matrix-values.
     
 \item[(2)]
  {\it When, for example, the whole target space-time itself is taken as a background simple D-brane
    or equivalently one of the two end-point of the oriented open string takes the  Neumann boundary condition
   instead of the Dirichlet boundary condition}$\,$:
   
 \item[]
 There are two sectors in this case:
  $$
	|\Lambda\rangle\; \Rightarrow\;  |\Lambda; i \rangle\,, \hspace{1em}1\le i\le r\,,
  	  \hspace{3em}\mbox{and}\hspace{3em}
    |\Lambda\rangle\; \Rightarrow\;  |\Lambda; \bar{j} \rangle\,,  \hspace{1em}1\le j\le r\,.
  $$
 The former are created by oriented open strings with only the beginning end-point stuck to the D-brane world-volume;
     and the latter are created by oriented open strings with only the ending end-point stuck to the D-brane world-volume.	
 The field $\psi_{\Lambda}$ on the D-brane world-volume that is associated
     to $\{|\Lambda;i \rangle\,|\, 1\le i \le r \}$ as a collection
   takes now $r\times 1$-matrix-values, i.e., column-vector-values.
 And the field $\psi_{\Lambda}^{\prime}$ on the D-brane world-volume that is associated
     to $\{|\Lambda;\bar{j} \rangle\,|\, 1\le \bar{j} \le r \}$ as a collection
   takes now $1\times r$-matrix-values, i.e., row-vector-values.
\end{itemize}
Cf.~{\sc Figure}~1-2.
%
%
\begin{figure}[htbp]
 \bigskip
  \centering
  \includegraphics[width=0.80\textwidth]{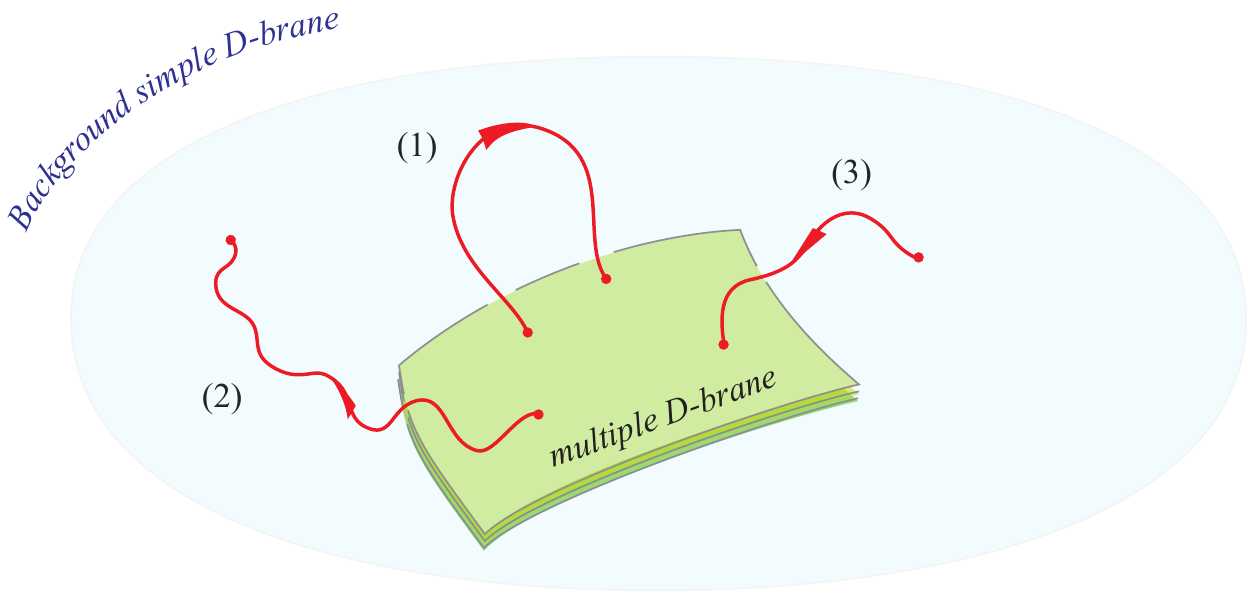}
 
  \bigskip
  \bigskip
 \centerline{\parbox{13cm}{\small\baselineskip 12pt
  {\sc Figure}~1-2.
   Three possible sectors of fields on the world-volume of coincident $D$-branes.
   They are created respectively by
   (1) oriented open strings with both end-points stuck to the D-brane world-volume, or
   (2) oriented open strings with only the beginning end-point stuck to the D-brane world-volume;   or
   (3) oriented open strings with only the ending end-point stuck to the D-brane world-volume.
   Sector (1) is always there on the D-brane world-volume
   while Sectors (2) and (3) can arise only when there is a background D-brane world-volume in the space-time
      to which the other end of oriented open strings can stick.
   Fields in Sector (1) (resp.\ Sector (2), Sector (3)) are matrix-valued
    (resp.\ column-vector-valued,  row-vector-valued).	
  }}  
\end{figure}

\bigskip

\begin{flushleft}	
{\bf Differentiable maps from matrix branes with fermions}
\end{flushleft}
As explained in [L-Y1] (D(1)),
  while originally the pair $(i,\bar{j})$ should be thought of  as labeling elements in the Lie algebra $u(r)$,
to bring geometry to  the enhanced scalar field on the D-brane world-volume
     that describes the collective deformations of the coincident D-branes
	 along the line of Grothendieck's Algebraic Geometry,
it is more natural to embed $u(n)$ into the Lie algebra $\gl(r,{\Bbb C})$
  which now has the underlying unital associative algebra structure, namely
  the matrix ring $M_{r\times r}({\Bbb C})$.
Following this line,
 coincident D-branes are now collectively described
 by a differentiable map
   $$
     \varphi\; :\;
     (X^{\!A\!z},{\cal E})\,  :=\,
		 (X,{\cal O}_X^{A\!z}:=\Endsheaf_{{\cal O}_X^{\,\Bbb C}}({\cal E}),
		{\cal E})\;
		 \longrightarrow\; Y
   $$		
 from a matrix manifold $X^{\!A\!z}$ with a fundamental module ${\cal E}$ to the target space-time $Y$;
cf.\ [L-Y3] (D(11.1)).
The vector bundle associated to ${\cal E}$ plays the role of the Chan-Paton bundle on the D-brane world-volume.
 
The enhancement of  fields on the D-brane world-volume due to coincidence of simple D-branes
  now takes the following form:
\begin{itemize}
 \item[(1)]
 {\it Field in Sector $(1)$ corresponding to
         $|  \Lambda\rangle \Rightarrow \{|\Lambda; i,\bar{j}  \rangle \,|\, 1\le i,j\le r\}$}$\,$:\\[.6ex]
  ${\cal O}_X$-module ${\cal F}_{\Lambda}\;\;\;
     \Rightarrow\;\;\;$
	 bi-${\cal O}_X^{A\!z}$-module
    $\; {\cal G}_{\Lambda}\;
	    :=\; {\cal E}\otimes_{{\cal O}_X^{\,\Bbb C}} {\cal F}_{\Lambda}
                     \otimes_{{\cal O}_{X}}{\cal E}^{\vee}
	  	\simeq {\cal O}_X^{A\!z}\otimes_{{\cal O}_X^{A\!z}}{\cal F}_{\Lambda}$.
 
 \item[(2)]
{\it Field in Sector $(2)$ corresponding to
         $|  \Lambda\rangle \Rightarrow \{|\Lambda; i \rangle \,|\, 1\le i \le r\}$}$\,$:\\[.6ex]
  ${\cal O}_X$-module ${\cal F}_{\Lambda}\;\;\;
    \Rightarrow \;\;\;$
  	 left ${\cal O}_X^{A\!z}$-module
    $\; {\cal G}_{\Lambda}\; :=\; {\cal E}\otimes_{{\cal O}_X}{\cal F}_{\Lambda}$.

 \item[(3)]
{\it Field in Sector $(3)$ corresponding to
         $|  \Lambda\rangle \Rightarrow \{|\Lambda; \bar{j} \rangle \,|\, 1\le j \le r\}$}$\,$:\\[.6ex]
 ${\cal O}_X$-module ${\cal F}_{\Lambda}\;\;\;
      \Rightarrow\;\;\;$
  right ${\cal O}_X^{A\!z}$-module
	 $\; {\cal G}_{\Lambda}\;
	 :=\; {\cal F}_{\Lambda}\otimes_{{\cal O}_X^{\,\Bbb C}}{\cal E}^{\vee}$.
\end{itemize}
Here
 ${\cal E}^{\vee}
     :=\Homsheaf_{{\cal O}_X^{\,\Bbb C}}({\cal E},{\cal O}_X^{\,\Bbb C}) $
 is the dual of ${\cal E}$.	
Note that  the functor in Item (2) (resp.\ Item (3)) is the functor that appears in the Morita equivalence
 of the category of ${\cal O}_X^{\,\Bbb C}$-modules and the category of left (resp.\ right)
 ${\cal O}_X^{A\!z}$-modules.
 
Sections of ${\cal G}_{\Lambda}$ correspond to the field  $\psi_{\Lambda}$
  on the world-volume of coincident D-branes in the previous theme.
Here, $\psi_{\Lambda}$  can be either bosonic or fermionic.
The dynamics of differentiable map $\varphi$  and that of sections of the various ${\cal G}_{\Lambda}$'s
 in general will influence each other through their equations of motion, which is a topic in its own right.
  
\bigskip

With this primitive setting in mind and as a motivation,
we now proceed to study
 how the fermionic degrees of freedom on a D-brane world-volume can be united into the geometry
  of the D-brane world-volume --- rendering it a matrix supermanifold with a fundamental module ---   and
 how the notion of differentiable maps from a matrix manifold can be promoted to the notion of
  differentiable maps from a matrix supermanifold.
  
\bigskip

\begin{sremark}$[\,$reduction from $M_{r\times r}({\Bbb C})$ to $u(r)$$\,]$. {\rm
 Gauge theoretically, a reduction from the underlying Lie algebra
   $\gl(r,{\Bbb C})$ of $M_{r\times r}({\Bbb C})$ to the original $u(r)$
  can be realized by introducing a Hermitian metric on the fundamental module ${\cal E}$.
 However, how this influences or constrains the notion of differential maps in our setting in [L-Y3] (D(11.1))
   should be studied in more detail.
}\end{sremark}

\bigskip

\section{Algebraic geometry over super-$C^k$-rings}

Basic notions and terminology from super-$C^k$-algebraic geometry required for the current note
 are introduced in this section.
The setting given is guided by
 a formal ${\Bbb Z}/2$-graded extension of $C^k$-algebraic geometry  and
 the goal to study fermionic D-branes later.

%
%
%

\bigskip

\subsection{Superrings, modules, and differential calculus on superrings}

We collect in this subsection the most basic notions in superrings, supermodules, and superdifferential calculus
 needed for the current note.
Readers are referred to the thesis `{\sl Superrings and supergroups}' [Wes] of Dennis Westra
 for further details and the foundation toward super-algebraic geometry in line with Grothendieck's Algebraic Geometry.

\bigskip

\begin{flushleft}
{\bf Superrings and modules over superrings}
\end{flushleft}
\begin{definition}{\bf [superring].} {\rm
 A {\it superring} $A$  is a ${\Bbb Z}/2$-graded
    ${\Bbb Z}/2$-commutative (unital associative) ring
  $A=A_0\oplus A_1$
  such that the multiplication $A\times A \rightarrow A$ satisfies
  $$
   \begin{array}{llll}
    \mbox{({\it ${\Bbb Z}/2$-graded})}            &&&
      A_0A_0\; \subset\;  A_0\,, \hspace{1em}
	  A_0A_1\; =\; A_1A_0\; \subset\;  A_1\,, \hspace{1em}\mbox{and}\hspace{1em}
	  A_1A_1\; \subset\;  A_0\,, \\[.6ex]
	\mbox{({\it ${\Bbb Z}/2$-commutative})}  &&&
	  aa^{\prime}\; =\; (-1)^{ii^{\prime}}a^{\prime}a
	  \hspace{2em}\mbox{for $a\in A_i$ and $a^{\prime}\in A_{i^{\prime}}$,
	                                                  $i, i^{\prime}=0,\,1\,$.}	
   \end{array}	
  $$
 A {\it morphism} between superrings (i.e.\ {\it superring-homomorphism})
  is a ${\Bbb Z}/2$-grading-preserving ring-homomorphism of the underlying unital associative rings.
 The elements of $A_0$ are called {\it even}, the elements of $A_1$ are called {\it odd},
 and an element that is either even or odd is said to be {\it homogeneous}.
 For a homogeneous element $a\in A$,
  denote by $|a|$ the {\it ${\Bbb Z}/2$-degree} or {\it parity} of $a$;
  $|a|=i$ if $a\in A_i$, for $i=0,\,1$.
 
 An {\it ideal} of $I$ of $A$ is said to be {\it ${\Bbb Z}/2$-graded}
  if $I=(I\cap A_0)+(I\cap A_1)$.
 In this case, $A$ induces a superring structure on the quotient ring $A/I$,
   with the ${\Bbb Z}/2$-grading given by $A/I=(A_0/(I\cap A_0)) \oplus(A_1/(I\cap A_1))  $.
 The converse is also true; cf.\ Definition/Lemma~2.1.2.
}\end{definition}

\medskip

\begin{definition-lemma}{\bf [${\Bbb Z}/2$-graded ideal = supernormal ideal].}
{\rm
   An ideal $I$ of a superring $A$ is called {\it supernormal}
    if $A$ induces a superring structure on the quotient ring $A/I$.}
 In terms of this, $I$ is ${\Bbb Z}/2$-graded if and only if $I$ is supernormal.	
\end{definition-lemma}
 
\begin{proof}
 The only-if part is immediate.
 For the if part,
   let $\kappa:A\rightarrow A/I$ is the quotient-superring map
      and $a=a_0+a_1\in I =\Ker(\kappa)$.
 If, say, $a_0\not\in I$,
 then both $\kappa(a_0)$ 	 and $\kappa(a_1)$ are non-zero in $A/I$
    and hence have parity even and odd respectively
   since $\kappa$ is a superring-homomorphism by the assumption.
 On the other hand,
  $\kappa(a)=\kappa(a_0)+\kappa(a_1)=0$; thus, $\kappa(a_0)=-\kappa(a_1)$.
 Since $(A/I)_0\cap (A/I)_1=0$,
  this implies that $\kappa(a_0)=\kappa(a_1)=0$, which is a contradiction.
 This proves the lemma.
 
\end{proof}

\bigskip

\begin{definition}{\bf [module over superring].} {\rm
 Let $A$ be a superring.
 An {\it $A$-module} $M$ is a {\it left module} over the unital associative ring underlying $A$
   that is endowed with a {\it ${\Bbb Z}/2$-grading} $M=M_0\oplus M_1$
    such that
     $$
	   A_0M_0\; \subset\; M_0\,, \hspace{1em}
	   A_1M_0\; \subset\; M_1\,, \hspace{1em}
	   A_0M_1\; \subset\; M_1\,, \hspace{1em}\mbox{and}\hspace{1em}
	   A_1M_1\;\subset\; M_0\,.
	 $$
 The elements of $M_0$ are called {\it even}, the elements of $M_1$ are called {\it odd},
   and an element that is either even or odd is said to be {\it homogeneous}.
 For a homogeneous element $m\in M$,
  denote by $|m|$ the {\it ${\Bbb Z}/2$-degree} or {\it parity} of $m$;
  $|m|=i$ if $m\in M_i$, for $i=0,\,1$.	
		
 For a superring $A$,
  \begin{itemize}
   \item[{\Large $\cdot$}]
 {\it a left $A$-module is canonically a right $A$-module}
    by setting
      $ma := (-1)^{|m||a|}am  $ for homogeneous elements $a\in A$ and $m\in M$
       and then extending ${\Bbb Z}$-linearly to all elements.
  \end{itemize}	
 For that reason, as in the case of commutative rings and modules,
  we don't  distinguish a left-, right-, or bi-module for a module over a superring.
  
 A {\it morphism} (or {\it module-homomorphism})
     $h:M\rightarrow M^{\prime}$ between $A$-modules
   is a right-module-homomorphism between the right-module over the unital associative ring underlying $A$;
 or equivalently
   a left-module-homomorphism between the left-module over the unital associative ring underlying $A$
   but with the sign rule applied to homogeneous components of $h$ and homogeneous elements of $A$.
 Explicitly,
 $h$ is said to be
  {\it even} if it preserves the ${\Bbb Z}/2$-grading   or
  {\it odd}  if it switches the ${\Bbb Z}/2$-grading;
 decompose $h$ to $h=h_0+h_1$  a summation of even and odd components,
 then $h_i(am)=(-1)^{i|a|}a h_i(m)$, $i=0,\,1$, for $a\in A$ homogeneous and $m\in M$.
 
 Subject to the above sign rules when applicable, the notion of
 
      \medskip
	
	{\Large $\cdot$}
       {\it submodule} $M^{\prime} \hookrightarrow M$, (cf.\ {\it monomorphism}),
	
	{\Large $\cdot$}
       {\it quotient module}  $M \twoheadrightarrow M^{\prime}$, (cf.\ {\it epimorphism}), 	
	
	{\Large $\cdot$}
	   {\it direct sum}  $M\oplus M^{\prime}$ of $A$-modules,
	
    {\Large $\cdot$}
      {\it tensor product} $M\otimes_A M^{\prime}$ of $A$-modules,
	
    {\Large $\cdot$}
	    {\it finitely generated}:
    		if $A^{\oplus l} \twoheadrightarrow M$
		     exists for some $l$,
			
    {\Large $\cdot$}
        {\it finitely presented}:
           if $A^{\oplus l^{\prime}}
		          \rightarrow   A^{\oplus l} \rightarrow M \rightarrow 0$
              is exact for some $l$, $l^{\prime}$				
	
       \medskip
	
	\noindent
    are all defined in the ordinary way as in commutative algebra.
}\end{definition}

\bigskip

\begin{flushleft}
{\bf Differential calculus on a superring}
\end{flushleft}
\begin{definition} {\bf [superderivation on superring].} {\rm
 Let $A$ be a superring over another superring $B$
   with $B\rightarrow A$ the underlying superring-homomorphism.
 Then,
  a ({\it left}) {\it super-$B$-derivation} $\zeta$ on $A$ {\it of  fixed parity} $|\zeta|=0$ or $1$
   is a map
    $$
      \zeta \;:\; A\;\longrightarrow\; A
    $$
   that satisfies
    $$
    \begin{array}{llll}	
     \mbox{({\it left-$B$-superlinearity})}
	   &&&  \zeta(ba+b^{\prime}a^{\prime})\;
	             =\;  (-1)^{|\zeta||b|}b\,\zeta(a)\,
				          +\,  (-1)^{|\zeta||b^{\prime}|} \,b^{\prime}\,\zeta(a^{\prime})\,,  \\[.6ex]
	 \mbox{({\it super Leibniz rule})}
	  &&& \zeta(aa^{\prime})\;
	           =\; \zeta(a)\, a^{\prime}\,
			           +\, (-1)^{|\zeta||a|}\,a\, \zeta(a^{\prime})
    \end{array}
   $$	
   for all $b,\,b^{\prime}\in B$ and $a,\,a^{\prime} \in A$ homogeneous,
   and all the ${\Bbb Z}$-linear extensions of these relations.
 A ({\it left}) {\it super-$B$-derivation} $\zeta$ on $A$
   is a formal sum $\zeta=\zeta_0+\zeta_1$
   of  a super-$B$-derivation $\zeta_0$ on $A$ of even parity  and
        a super-$B$-derivation $\zeta_1$ on $A$ of odd parity.
 
 Denote by $\sDer_B(A)$ the set of all super-$B$-derivations on $A$.
 Then, $\sDer_B(A)$ is ${\Bbb Z}/2$-graded by construction.
 Furthermore,
   if $\zeta\in \sDer_B(A)$,
      then so does $a \zeta = (-1)^{|a||\zeta|}\zeta a$,
     with
	 $$
	   (a\zeta)(\,\cdot\,)\; :=\;   a(\zeta(\,\cdot\,))
  	        \hspace{2em}\mbox{and}\hspace{2em}
	   (\zeta a)(\,\cdot\,)\; :=\;  (-1)^{|a||\,\cdot\,|}\,(\zeta(\,\cdot\,))a
	 $$
	 for $a\in A$.
 Thus,
  {\it $\sDer_B(A)$ is naturally a (bi-)$A$-module},
    with $a\cdot \zeta :=  a\zeta$ and $\zeta\cdot a := \zeta a$
	        for $a\in A$ and $\zeta\in \sDer_B(A)$.
			
 Note that
  if $\zeta,\, \zeta^{\prime} \in \sDer_B(A)$ (homogeneous),
   then so does their {\it super Lie bracket}  (synonymously, {\it supercommutator})
   $$
     [\zeta, \zeta^{\prime}]\;
	    :=\;  \zeta\zeta^{\prime}\,
		         -\,(-1)^{|\zeta||\zeta^{\prime}|}\,\zeta^{\prime}\zeta\,.
   $$
 Thus, {\it $\sDer_B(A)$ is naturally a super-Lie algebra}.
 Homogeneous elements in which satisfy the super-anti-commutativity identity and the super-Jacobi identity:
  $$
    \begin{array}{c}
	  [\zeta, \zeta^{\prime}]\;
	    =\;  -\,(-1)^{|\zeta||\zeta^{\prime}|}\,[\zeta^{\prime},\zeta]\,, \\[.6ex]
	  [\zeta,[\zeta^{\prime},\zeta^{\prime\prime}]]\;
	    =\;  [[\zeta,\zeta^{\prime}], \zeta^{\prime\prime}] \,
		        +\, (-1)^{|\zeta||\zeta^{\prime}|}\,[\zeta^{\prime},[\zeta,\zeta^{\prime\prime}]]\,.	
   \end{array}
  $$
 
 When $A$ is a $k$-algebra,
  we will denote $\sDer_k(A)$ also by $\sDer(A)$, with the ground field $k$ understood.
}\end{definition}

\medskip

\begin{remark}$[\,$equivalent condition$\,]$. {\rm
 The {\it left-$B$-superlinearity condition} and the {\it super Leibniz rule condition} in Definition~2.1.4,
  are equivalent to
   $$
    \begin{array}{llll}	
     \mbox{({\it right-$B$-linearity})}
	  &&&  \zeta(ab+a^{\prime}b^{\prime})\;
	             =\;  \zeta(a)\,b\, +\,   \zeta(a^{\prime})\,b^{\prime}   \,,  \\[.6ex]
	 \mbox{({\it super Leibniz rule}$^{\prime}$)}
	  &&&  \zeta(aa^{\prime})\;
	             =\; \zeta(a)\, a^{\prime}\,
			                   +\, (-1)^{|a||a^{\prime}|}\,\zeta(a^{\prime})\,a
    \end{array}
   $$	
  respectively.
 Note that the parity $|\zeta|$ of $\zeta$ is removed in this equivalent form.
 The format in Definition~2.1.4 is more natural-looking
   while the above equivalent form is more convenient to use occasionally.
 Cf.\ Definition~2.1.7.
   
}\end{remark}

\medskip

\begin{remark} $[\,$inner superderivation$\,]$. {\rm
 Similarly to the commutative case,
   all the inner superderivations of  a superring are zero.
}\end{remark}

\medskip

\begin{definition} {\bf [superderivation with value in module].} {\rm
  Let $A$ be a superring over another superring $B$ and $M$ be an $A$-module.
  A $\Bbb Z$-linear map
    $$
	   d\; :\;  A\;  \rightarrow\;  M
	$$
    is called a {\it super-$B$-derivation with values in $M$}
  if $d$ satisfies
    $$
    \begin{array}{llll}	
     \mbox{({\it right-$B$-linearity})}
	  &&&  d(ab+a^{\prime}b^{\prime})\;
	             =\;  d(a)\,b\, +\,   d(a^{\prime})\,b^{\prime}   \,,  \\[.6ex]
	 \mbox{({\it super Leibniz rule}$^{\prime}$)}
	  &&&  d(aa^{\prime})\;
	             =\; d(a)\, a^{\prime}\,
			                   +\, (-1)^{|a||a^{\prime}|}\,d(a^{\prime})\,a
    \end{array}
   $$	
   for $a$, $a^{\prime}\in A$ homogeneous, and all ${\Bbb Z}$-linear extension of such identities.
 In particular, $d$ is a $B$-module-homomorphism.
 $d$ is said to be {\it even} if $d(A_0)\subset M_0$ and $d(A_1)\subset M_1$;
   and {\it odd} if $d(A_0)\subset M_1$ and $d(A_1)\subset M_0$.
 The set $\sDer_B(A,M)$ of all super-$B$-derivations with values in $M$
  is naturally a (${\Bbb Z}/2$-graded) bi-$A$-module,
   with the multiplication defined by
    $$
	   a \cdot d\;:\; a^{\prime}\; \longmapsto\;   a\,(d(a^{\prime}))
	    \hspace{2em}\mbox{and}\hspace{2em}
	  d \cdot a\;:\; a^{\prime}\; \longmapsto\;   (-1)^{|a||a^{\prime}|}\,(d(a^{\prime}))\,a
	$$
 for $a\in A$ and $d\in \sDer_B(A,M)$ homogeneous, plus a ${\Bbb Z}$-linear extension.
}\end{definition}

\medskip

\begin{definition} {\bf [module of differentials of superring].} {\rm
 Continuing Definition~2.1.7.
 An $A$-module $M$ with a super-$B$-derivation $d: A\rightarrow M$
   is called the {\it cotangent module} of $A$
  if it satisfies the following universal property:
  \begin{itemize}
   \item[{\Large $\cdot$}]
    For any $A$-module $M^{\prime}$  and
    super-$B$-derivation $d^{\prime}: A  \rightarrow M^{\prime}$,
    there exists a unique homomorphism of $A$-modules $\psi:M\rightarrow M^{\prime}$
    such that $\;d^{\prime}= \psi \circ d\,$.	
 	$$
     \xymatrix{
	 & A \ar[rr]^-{d} \ar[dr]_-{d^{\prime}}  && M   \ar @{.>} [dl]^-{\psi}\\
	 &&    M^{\prime}&&.		  		
	  }
    $$			
  \end{itemize}
  (Thus, $M$ is unique up to a unique $A$-module isomorphism.)
 We denote this $M$ with $d: A\rightarrow M$ by $\Omega_{A/B}$,
  with the built-in super-$B$-derivation $d: A \rightarrow \Omega_{A/B}$ understood.
}\end{definition}

\medskip

\begin{remark}$[\,$explicit construction of $\Omega_{A/B}$$\,]$. {\rm
 The cotangent module $\Omega_{A/B}$ of a superring $A$ over $B$
   can be constructed explicitly from the $A$-module generated by the set
   $$
     \{ d(a)\,|\,  a \in A   \}\,,
   $$
   subject to the relations
   $$
    \begin{array}{llll}	
	 \mbox{({\it ${\Bbb Z}/2$-grading})}
	  &&& |d(a)|\;=\; |a|\,,   \\[.6ex]
     \mbox{({\it right-$B$-linearity})}
	  &&& d(ab+a^{\prime}b^{\prime})\;
	            =\; d(a)\,b\, +\,d(a^{\prime})\, b^{\prime}\,,   \\[.6ex]
	 \mbox{({\it super Leibniz rule$^{\prime}$})}
	  &&&  d(aa^{\prime})\;
	               =\;  d(a)\,a^{\prime}\,
				          +\, (-1)^{|a||a^{\prime}|}\, d(a^{\prime})\, a \,,   \\[.6ex]
	 \mbox{({\it bi-$A$-module structure})}
	  &&&  d(a)\,a^{\prime}\;=\;   (-1)^{|a||a^{\prime}|}\,  a^{\prime}\,d(a)
    \end{array}
   $$	
   for all $b,\, b^{\prime}\in B$, $a,\,a^{\prime}\in A$ homogeneous,
      plus a ${\Bbb Z}$-linear extension of these relations.
 Denote the image of $d(a)$ under the quotient by $da$.
 Then, by definition, the built-in map
  $$
    \begin{array}{ccccc}
	 d & : &  A      & \longrightarrow   & \Omega_{A/B}   \\[.6ex]
	    &   &  a       & \longmapsto        & da
	\end{array}
  $$
  is a super-$B$-derivation from $A$ to $\Omega_{A/B}$.
}\end{remark}

\medskip

\begin{remark}
$[\,$relation between $\sDer_B(A)$ and $\Omega_{A/B}$$\,]$.
{\rm
 The universal property of $\Omega_{A/B}$ implies that
  there is an $A$-module-isomorphism
   $\sDer_B(A,M)\simeq  \Hom_A(\Omega_{A/B},M)$ for any $A$-module $M$.
 In particular, for $M=A$, one has
   $\sDer_B(A)\simeq  \Hom_A(\Omega_{A/B}, A)$.
}\end{remark}

\bigskip

\subsection{Super-$C^k$-rings, modules, and differential calculus on super-$C^k$-rings}

Basic notions and terminology from super-$C^k$-algebra are introduced in this subsection.
They serve the basis to construct super-$C^k$-schemes and quasi-coherent sheaves thereupon

\bigskip

\begin{remark}$[\,$on the setting in this subsection, alternative, and issue beyond$\,]$. {\rm
 The setting in the current subsection allows the transcendental/nonalgebraic notion of $C^k$-rings
   to merge with the algebraic notion of superrings immediately.
 It leads to the notion of super-$C^k$-manifolds and super-$C^k$-schemes
   that are nothing but a sheaf-type super-thickening of ordinary $C^k$-manifolds and ordinary $C^k$-schemes;
   cf.\ Remark~2.3.16 and {\sc Figure}~3-1.
 While mathematically these are not the most general kind of superspaces,
   physically they are broad enough to cover the superspaces and supermanifolds that appear
   in supersymmetric quantum field theory and superstring theory in most situations.

 There are alternatives to our setting.
 A most fundamental one would be re-do the algebraic geometry over $C^k$-rings,
  using now the $C^k$-function rings
  $\coprod_{(p,q)}C^k({\Bbb R}^{p|q})$,
  where ${\Bbb R}^{p|q}$ is the super-$C^k$-manifold
   whose coordinates $(x^1,\,\cdots\,, x^p;\theta^1,\,\cdots\,, \theta^q)$ have
    both commuting and anti-commuting variables.
 Different interpretations/settings
   to the evaluation of ordinary $C^k$-functions of ${\Bbb R}^p$
   on $p$-tuples of supernumbers $\in {\Bbb R}^{1|q^{\prime}}$
  may lead to different classes of algebraic geometry over super-$C^k$-rings.
 In this sense, our setting is the simplest one.
 Other more general settings should be studied in their own right.
}\end{remark}

\clearpage

\begin{flushleft}
{\bf Super-$C^k$-rings}
\end{flushleft}
With Sec.~2.1 as background, we now build into the superring in question
 an additional $C^k$-ring structure on an appropriate subring of its even subring.

\bigskip

\begin{definition} {\bf [superpolynomial ring over $C^k$-ring].} {\rm
 Let $R$ be a $C^k$-ring.
 A {\it superpolynomial ring over $R$} is a (unital) associative ring over $R$ of the following form
  $$
     R[\theta^1, \,\cdots\,,\theta^s]\;
       :=\;  \frac{  R\langle \theta^1, \,\cdots\,,\theta^s\rangle}
			   {( r\theta^{\alpha}-\theta^{\alpha}r\,,\,
   			            \theta^{\beta}\theta^{\gamma}+\theta^{\gamma}\theta^{\beta}
						 \rule{0em}{1em}
			             \:|\: r\in R\,,\, 1\le \alpha,\beta,\gamma \le s)}\,,
  $$
 where
 \begin{itemize}
  \item[{\Large $\cdot$}]
   $R\langle \theta^1, \,\cdots\,,\theta^s\rangle$
   is the unital associative ring over $R$ generated by the variables $\theta^1, \,\cdots\,,\theta^s$,
		
  \item[{\Large $\cdot$}]		
   $( r\theta^{\alpha}-\theta^{\alpha}r\,,\,
   			            \theta^{\beta}\theta^{\gamma}+\theta^{\gamma}\theta^{\beta}
						 \rule{0em}{1em}
			             \:|\: r\in R\,,\, 1\le \alpha,\beta,\gamma \le s)$   
   is the bi-ideal in $R\langle \theta^1,\,\cdots\,,\theta^s \rangle $  generated by the elements indicated.

  \item[{\Large $\cdot$}]
   The ${\Bbb Z}/2$-grading is determined by specifying
    $|r|=0$, $|\theta^{\alpha}|=1$ for all $r\in R$ and $1\le \alpha\le s$  and
 	the extension by the product rule
	$|\widehat{r}\widehat{r}^{\prime}|=|\widehat{r}||\widehat{r}^{\prime}|$
	whenever applicable.
 \end{itemize}
 
 Underlying $R[\theta^1,\,\cdots\,\theta^s]$ as an algebra-extension of $R$
  is a {\it built-in split short exact sequence} of $R$-modules
  $$
      \xymatrix{
        0  \ar[r]    &   (\theta^1,\,\cdots\,,\theta^s)  \ar[r]      & R[\theta^1,\,\cdots\,,\theta^s]    \ar[r]
	                     & R \ar[r]   \ar@{.>}@/^/[l]     & 0\;,
       }
  $$
  where $R\rightarrow R[\theta^1,\,\cdots\,, \theta^s]$ is the built-in $R$-algebra inclusion map.
 It is also an exact sequence of $R[\theta^1,\,\cdots\,,\theta^s]$-modules in the sense of Definition~2.1.3.

 Let
   $V$ be a vector space over ${\Bbb R}$ of dimension $s$,
       spanned by $\{\theta^1,\,\cdots\,,\theta^s\}$ , and
   $\bigwedge^{\bullet}V$ be the (${\Bbb Z}/2$-graded ${\Bbb Z}/2$-commutative)
   {\it Grassmann/exterior algebra} associated to $V$.
 Then
  $$\mbox{$
     R[\theta^1,\,\cdots\,,\theta^s]\;  \simeq\;   R\otimes_{\Bbb R}\bigwedge^{\bullet}V
	 $}
  $$
   as ${\Bbb Z}/2$-graded ${\Bbb Z}/2$-commutative rings over $R$,
   with
     $$\mbox{$
	    R[\theta^1,\,\cdots\,,\theta^s]_0\;
	      \simeq\;  R\otimes_{\Bbb R}\bigwedge^{\scriptsizeeven}V
   		\hspace{2em}\mbox{and}\hspace{2em}
       R[\theta^1,\,\cdots\,,\theta^s]_1\;
	      \simeq\;  R\otimes_{\Bbb R}\bigwedge^{\scriptsizeodd}V\,.
		$}
     $$.
}\end{definition}

\smallskip

\begin{definition} {\bf [super-$C^k$-ring: split super-extension of $C^k$-ring].} {\rm
 Let $R$ be a $C^k$-ring.
 A {\it split super-extension $\widehat{R}$ of $R$} is a (unital) associative ring $\widehat{R}$ over $R$
   that is equipped with a split short exact sequence of $R$-modules
   $$
     \xymatrix{
      0  \ar[r]    &   M  \ar[r]    & \widehat{R}\ar[r]   & R \ar[r]   \ar@{.>}@/^/[l]   &  0
      }
   $$
  such that
  \begin{itemize}
   \item[{\Large $\cdot$}]
    There exists a superpolynomial ring $R[\theta^1,\,\cdots\,,\theta^s]$ over $R$, for some $s$,
	  that can realize $\widehat{R}$ as its superring-quotient $R$-algebra
     $$
       \xymatrix{
	     R[\theta^1,\,\cdots\,,\theta^s]    \ar@{->>}[r]       & \widehat{R} \\
        }
     $$
   in such a way that  the following induced diagram commutes
   $$
     \xymatrix{
      &  0  \ar[r]    &   (\theta^1,\,\cdots\,,\theta^s)  \ar[r] \ar@{->>}[d]
	                     & R[\theta^1,\,\cdots\,,\theta^s]    \ar[r] \ar@{->>}[d]
	                     & R \ar[r]   \ar@{.>}@/^/[l]   \ar@{=}[d]    &  0   \\
	  &  0  \ar[r]    &   M  \ar[r]    & \widehat{R}\ar[r]   & R \ar[r]   \ar@{.>}@/^/[l]   &  0   &.
     }
   $$
   \end{itemize}
 We will call $\widehat{R}$ synonymously a {\it super-$C^k$-ring} over $R$.
 In particular, $R$ itself is trivially a super-$C^k$-ring, with $M=0$.
}\end{definition}

\smallskip

\begin{definition}{\bf [homomorphism between super-$C^k$-rings].} {\rm
 Let
   $R$ and $S$ be $C^k$-rings,
   $\widehat{R}$ be a super-$C^k$-ring over $R$,  and
   $\widehat{S}$ be a super-$C^k$-ring over $S$.
 A {\it super-$C^k$-ring-homomorphism} from $\widehat{R}$ to $\widehat{S}$
 is a pair of superring-homomorphisms (cf.\ Definition~2.1.1)
  $$
   \widehat{f}\;:\; \widehat{R}\;\longrightarrow\; \widehat{S}
     \hspace{2em}\mbox{and}\hspace{2em}
    f \;:\;   R \; \longrightarrow\;  S
  $$
  such that
  \begin{itemize}
   \item[(1)]
    $f$ is a $C^k$-ring-homomorphism,
	
  \item[(2)]
   $(\widehat{f}, f)$ is compatible with the underlying super-$C^k$-ring structure
    of $\widehat{R}$ and $\widehat{S}$; namely, the following diagram commutes
   $$
    \xymatrix{
	 &  \widehat{R}\ar[rr]^-{\widehat{f}} \ar@{->>}[d]     && \widehat{S}\ar@{->>}[d]\\
     &  R \ar@/^/@{.>}[u]  \ar[rr]^-{f}                     && S \ar@/^/@{.>}[u]    &.
	}
   $$
 \end{itemize}
 For the simplicity of notations, we may denote the pair
   $(\widehat{f},f):(\widehat{R},R)\rightarrow (\widehat{S},S)$
  also as $\widehat{f}:\widehat{R}\rightarrow \widehat{S}$.
 We say that
  a super-$C^k$-ring-homomorphism $\widehat{f}:\widehat{R}\rightarrow \widehat{S}$ is {\it injective}
     (resp.\ {\it surjective})
    if both $\widehat{f}$ and $f$ are injective (resp.\ surjective).
 In this case, $\widehat{f}$ is called a {\it super-$C^k$-ring-monomorphism}
  (resp.\ {\it super-$C^k$-ring-epimorphism}).	
}\end{definition}

\medskip

\begin{definition} {\bf [ideal, super-$C^k$-normal ideal, super-$C^k$-quotient].} {\rm
 Let $\widehat{R}$ be a super-$C^k$-ring.
 An {\it ideal} $\widehat{I}$ of $\widehat{R}$ is an ideal $\widehat{I}$ of $\widehat{R}$
  as a ${\Bbb Z}/2$-graded ${\Bbb R}$-algebra.
 $\widehat{I}$ is called {\it super-$C^k$-normal}
  if the super-$C^k$-ring structure on $\widehat{R}$ descends to a super-$C^k$-ring structure
  on the quotient ${\Bbb R}$-algebra $\widehat{R}/\widehat{I}$.
 In this case,
  $\widehat{I}$ must be a ${\Bbb Z}/2$-graded ideal of $\widehat{R}$
  (cf.\ Lemma~2.1.2).
 $\widehat{R}/\widehat{I}$ with the induced super-$C^k$-ring structure
    is called a {\it super-$C^k$-quotient} of $\widehat{R}$
   and one has the following commutative diagram
    $$
    \xymatrix{
	 &  \widehat{R}\ar[rr]^-{\widehat{q}} \ar@{->>}[d]
	     && \widehat{R}/\widehat{I}\ar@{->>}[d]\\
     &  R \ar@/^/@{.>}[u]  \ar[rr]^-{q}                     && R/I \ar@/^/@{.>}[u]    &\hspace{-3em},
	}
   $$
   where
     $\widehat{q}:\widehat{R}\rightarrow \widehat{R}/\widehat{I}$ is the quotient map  and
     $I:= \widehat{I}\cap R$ is now a $C^k$-normal ideal of the $C^k$-ring $R$.
 Note that for $\widehat{I}$ $C^k$-normal,
    the quotient super-$C^k$-ring structure on $\widehat{R}/\widehat{I}$
    is compatible with the quotient ${\Bbb R}$-algebra structure.
}\end{definition}

%
%
%
%
%
%

\medskip

\begin{definition}
 {\bf [localization of super-$C^k$-ring].} {\rm
 Let
   $\widehat{R}$ be a super-$C^k$-ring,
	         with $R\hookrightarrow \widehat{R}$ the built-in inclusion,
   $S$ be a subset of $R$,  and
   $R[S^{-1}]$	be the localization of the $C^k$-ring $R$ at $S$,
	   with the built-in $C^k$-ring-homomorphism $R\rightarrow R[S^{-1}]$.
 The {\it localization} of  $\widehat{R}$ at $S$, denoted by $\widehat{R}[S^{-1}]$,
    is the super-$C^k$-ring over $R[S^{-1}]$ defined by
	$$
	 \widehat{R}[S^{-1}]\; :=\;   \widehat{R}\otimes_R R[S^{-1}]\,.
	$$
 It goes with a built-in super-$C^k$-ring-homomorphism
  $\widehat{R}\rightarrow \widehat{R}[S^{-1}]$.
}\end{definition}

\bigskip

\begin{flushleft}
{\bf Modules over super-$C^k$-rings}
\end{flushleft}

\begin{definition} {\bf [module over super-$C^k$-ring].} {\rm
  Let $R$ be a $C^k$-ring and $\widehat{R}$ be a super-$C^k$-ring over $R$.
  Recall Definition~2.1.3.
  A {\it module $\widehat{M}$ over $\widehat{R}$}, or {\it $\widehat{R}$-module},
    is a module over $\widehat{R}$ as a superring.

  The notion of
   
     \medskip
	
    {\Large $\cdot$}
	   {\it homomorphism} $\widehat{M}_1 \rightarrow \widehat{M}_2$ of $\widehat{R}$-modules,
 
      \medskip
	
	{\Large $\cdot$}
       {\it submodule} $\widehat{M}_1 \hookrightarrow \widehat{M}_2$, (cf.\ {\it monomorphism}),
	
	  \medskip
	
	{\Large $\cdot$}
       {\it quotient module}  $\widehat{M}_1
	      \twoheadrightarrow \widehat{M}_2$, (cf.\ {\it epimorphism}), 	
	
	  \medskip
	
	{\Large $\cdot$}
	   {\it direct sum}  $\widehat{M}_1\oplus \widehat{M}_2$ of $\widehat{R}$-modules,
	
	  \medskip
	
    {\Large $\cdot$}
      {\it tensor product}
	      $\widehat{M}_1\otimes_{\widehat{R}}\widehat{M}_2$ of $\widehat{R}$-modules,
	
	  \medskip
	
    {\Large $\cdot$}
	    {\it finitely generated}:
    		if $\widehat{R}^{\,\oplus l} \twoheadrightarrow \widehat{M}$
		     exists for some $l$,
	
	  \medskip
	
    {\Large $\cdot$}
        {\it finitely presented}:
           if $\widehat{R}^{\,\oplus l^{\prime}}
		          \rightarrow  \widehat{R}^{\,\oplus l} \rightarrow  \widehat{M}\rightarrow 0$
              is exact for some $l$, $l^{\prime}$				
	
      \medskip
	
	\noindent
    are all defined as in Definition~2.1.3 for modules over a superring.

 Denote by $\ModCategory(\widehat{R})$ the category of modules over  $\widehat{R}$.	
}\end{definition}

\medskip

\begin{remark} $[\,$module over super-$C^k$-ring vs.\ module over $C^k$-ring$\,]$. {\rm
 Recall the built-in ring-homomorphism $R\rightarrow \widehat{R}$.
 Thus every $\widehat{R}$-module is canonically an $R$-module.
 The induced functor $\ModCategory(\widehat{R})\rightarrow \ModCategory(R)$ is exact.
}\end{remark}

\medskip

\begin{definition}{\bf [localization of module over super-$C^k$-ring].}   {\rm
 Let
   $\widehat{M}$ be a module over a super-$C^k$-ring $\widehat{R}$.
 Recall the built-in inclusion $R\hookrightarrow \widehat{R}$.
 Let $\widehat{R}\rightarrow \widehat{R}[S^{-1}]$
   be the localization of $\widehat{R}$ at a subset $S\subset R$.
 Then, the {\it localization of $\widehat{M}$ at $S$}, denoted by $\widehat{M}[S^{-1}]$,
   is the $\widehat{R}[S^{-1}]$-module defined by
     $$
	   \widehat{M}[S^{-1}]\;
	     :=\;  \widehat{R}[S^{-1}]\otimes_{\widehat{R}} \widehat{M}\,.
	 $$
 By construction,  it is equipped with an $\widehat{R}$-module-homomorphism
   $\widehat{M} \rightarrow \widehat{M}[S^{-1}]$.
}\end{definition}

\bigskip

\begin{flushleft}
{\bf Differential calculus on super-$C^k$-rings}
\end{flushleft}

\begin{definition} {\bf [superderivation on super-$C^k$-ring].} {\rm
 Let $\widehat{R}$ be a super-$C^k$-ring over another super-$C^k$-ring $\widehat{S}$
   with $\widehat{S}\rightarrow\widehat{R}$ the built-in super-$C^k$-ring-homomorphism.
 Then,
  a ({\it left}) {\it super-$C^k$-$\widehat{S}$-derivation $\widehat{\Theta}$}
   on $\widehat{R}$ is a map
    $$
      \widehat{\Theta} \;:\; \widehat{R}\;\longrightarrow\; \widehat{R}
    $$
   that satisfies
    $$
    \begin{array}{llll}	
     \mbox{({\it right-$\widehat{S}$-linearity})}
	  &&& \widehat{\Theta}(\widehat{r}\widehat{s}
	                                                       +\widehat{r}^{\prime}\widehat{s}^{\prime})\;
	            =\;  \widehat{\Theta}(\widehat{r})\,\widehat{s}\,
				      +\,   \widehat{\Theta}(\widehat{r}^{\prime})\,\widehat{s}^{\prime}   \,,   \\[.6ex]
	 \mbox{({\it super Leibniz rule$^{\prime}$})}
	  &&& \widehat{\Theta}(\widehat{r}\widehat{r}^{\prime})\;
	           =\; \widehat{\Theta}(\widehat{r})\, \widehat{r}^{\prime}\,
			           +\, (-1)^{|\widehat{r}||\widehat{r}^{\prime}|}\,
					           \widehat{\Theta}(\widehat{r}^{\prime})\,\widehat{r}
    \end{array}
   $$	
   for all $\widehat{s},\,\widehat{s}^{\prime}\in \widehat{S}$,
        $\widehat{r},\,\widehat{r}^{\prime} \in \widehat{R}$ homogeneous,
	 and the ${\Bbb R}$-linear extensions of these relations,  and
  $$
   \begin{array}{l}
    \hspace{-2em}\mbox{({\it chain rule})}\\[1.2ex]
    \widehat{\Theta}(h(r_1,\,\cdots\,,\,r_l))\;
	  =\;  \partial_1 h(r_1,\,\cdots\,,\,r_l)\,\widehat{\Theta}(r_1)\;
	         +\; \cdots\; +\;
			 \partial_l h(r_1,\,\cdots\,,\,r_l)\,\widehat{\Theta}(r_l)	
   \end{array}			
  $$
  for all $h\in C^k({\Bbb R}^l)$,
    $l\in {\Bbb Z}_{\ge 1}$, and $r_1,\,\cdots\,,\,r_l\in$
	     the $C^k$-ring $R \subset \widehat{R}$.
 
 Denote by $\sDer_{\widehat{S}}(\widehat{R})$
   the set of all super-$C^k$-$\widehat{S}$-derivations on $\widehat{R}$.
 Then, $\sDer_{\widehat{S}}(\widehat{R})$ is ${\Bbb Z}/2$-graded by construction.
 Furthermore,
   if $\widehat{\Theta}\in \sDer_{\widehat{S}}(\widehat{R})$,
      then so does
	    $\widehat{r}\widehat{\Theta}
		    = (-1)^{|\widehat{r}||\widehat{\Theta}|}\widehat{\Theta}\widehat{ r}$,
     with $(\widehat{r}\widehat{\Theta})(\,\cdot\,)
	                    := \widehat{r}(\widehat{\Theta}(\,\cdot\,))$ and
	         $(\widehat{\Theta}\widehat{ r})(\,\cdot\,)
			            :=  (-1)^{|\widehat{r}||\,\cdot\,|} (\widehat{\Theta}(\,\cdot\,))\widehat{r}$,
	 for $\widehat{r}\in\widehat{R}$.
 Thus,
  {\it $\sDer_{\widehat{S}}(\widehat{R})$ is naturally a (bi-)$\widehat{R}$-module},
    with $\widehat{r}\cdot \widehat{\Theta} := \widehat{r}\widehat{\Theta}$ and
            $\widehat{\Theta}\cdot \widehat{r} := \widehat{\Theta}\widehat{r}$
	        for $\widehat{\Theta}\in \sDer_{\widehat{S}}(\widehat{R})$,
			     $\widehat{r}\in \widehat{R}$.
 
 Furthermore,
  if $\widehat{\Theta},\, \widehat{\Theta}^{\prime}
         \in \sDer_{\widehat{S}}(\widehat{R})$  (homogeneous),
   then so does the {\it super Lie bracket}
   $$
     [\widehat{\Theta}, \widehat{\Theta}^{\prime}]\;
	    :=\;  \widehat{\Theta}\widehat{\Theta}^{\prime}\,
		         -\,(-1)^{|\widehat{\Theta}||\widehat{\Theta}^{\prime}|}\,
				          \widehat{\Theta}^{\prime}\widehat{\Theta}
   $$
   of $\widehat{\Theta}$ and $\widehat{\Theta}^{\prime}$.
 Thus, {\it $\sDer_{\widehat{S}}(\widehat{R})$ is naturally a super-Lie algebra}.
 
 We will denote $\sDer_{\Bbb R}(\widehat{R})$ also by $\sDer(\widehat{R})$.
}\end{definition}

\medskip

\begin{remark} $[\,$inner superderivation$\,]$. {\rm
 Similarly to the commutative case,
   all the inner superderivations of  a super-$C^k$-ring are zero.
}\end{remark}

\medskip

\begin{remark} $[\,$relation to $\Der_S(R)$$\,]$.  {\rm
 Note that the built-in inclusion $R\hookrightarrow \widehat{R}$ induces
  a natural $\widehat{R}$-module-homomorphism
   $\sDer_{\widehat{S}}(\widehat{R})
       \rightarrow \Der_S(R)\otimes_R\widehat{R} $.
}\end{remark}

\medskip

\begin{definition} {\bf [superderivation with value in module].}
{\rm
  Let
     $\widehat{R}$ be a super-$C^k$-ring over a super-$C^k$-ring $\widehat{S}$  and
     $\widehat{M}$ an $\widehat{R}$-module.
  An ${\Bbb R}$-linear map
    $$
	   d\; :\;  \widehat{R}\;  \rightarrow\; \widehat{ M}
	$$
    is called a {\it super-$C^k$-$\widehat{S}$-derivation with values in $\widehat{M}$},
  if
   $$
     \begin{array}{llll}
      \mbox{({\it right-$\widehat{S}$-linearity})}
	  &&& d(\widehat{r}\widehat{s}+\widehat{r}^{\prime}\widehat{s}^{\prime})\;
	            =\;  d(\widehat{r})\,\widehat{s}\,
				      +\,   d(\widehat{r}^{\prime})\,\widehat{s}^{\prime}   \,,   \\[.6ex]	
	  \mbox{({\it super Leibniz rule$^{\prime}$})}
	   &&& d(\widehat{r}\widehat{r}^{\prime})\;
	               =\;   d(\widehat{r})\,\widehat{r}^{\prime}\,
                    		 +\,  (-1)^{|\widehat{r}||\widehat{r}^{\prime}|}\,
							                   d(\widehat{r}^{\prime})\,\widehat{r}    \\[.6ex]
      \mbox{({\it chain rule})}	
       &&&   d (f(r_1,\,\cdots\,,\,r_n)) \;
		           =\; \sum_{i=1}^n (\partial_if)(r_1,\,\cdots,\,r_n)\cdot dr_i
	 \end{array}	
   $$
   for all $\widehat{s},\,\widehat{s}^{\prime}\in\widehat{S}$,
             $\widehat{r}$, $\widehat{r}^{\prime}\in \widehat{R}$ homogeneous,
             $f\in \cup_{n}C^k({\Bbb R}^n)$,	 and $r_i\in R$.
  Here
    $\partial_if$ is the partial derivative of $f\in C^k({\Bbb R}^n)$
    with respect  to the $i$-th coordinate of ${\Bbb R}^n$.
 In particular, $d$ is a $\widehat{S}$-module-homomorphism.
 The set $\sDer_{\widehat{S}}(\widehat{R}, \widehat{M})$
   of all super-$C^k$-$\widehat{S}$-derivation with values in $\widehat{M}$
   is naturally an $\widehat{R}$-module,
   with the multiplication defined by
      $\widehat{r}\cdot d: \widehat{r}^{\prime}\mapsto
           \widehat{r}(d(\widehat{r}^{\prime}))$ and
	  $d\cdot \widehat{r}: \widehat{r}^{\prime}
	     \mapsto  (-1)^{|\widehat{r}||\widehat{r}^{\prime}|}
		                     (d(\widehat{r}^{\prime}))\widehat{r}$
   for $\widehat{r}\in\widehat{R}$.
}\end{definition}

\medskip

\begin{definition} {\bf [module of differentials of super-$C^k$-ring].} {\rm
 Let $\widehat{R}$ be a super-$C^k$-ring over $\widehat{S}$.
 An $\widehat{R}$-module $\widehat{M}$ with a super-$C^k$-$\widehat{S}$-derivation
   $d:\widehat{R}\rightarrow M$
    is called the {\it $C^k$-cotangent module} of $\widehat{R}$ over $\widehat{S}$
   if it satisfies the following universal property:
   \begin{itemize}
    \item[{\Large $\cdot$}]
	 For any $\widehat{R}$-module $\widehat{M}^{\prime}$  and
	   super $C^k$-$\widehat{S}$-derivation
	   $d^{\prime}:\widehat{R}\rightarrow \widehat{M}^{\prime}$,
	 there exists a unique homomorphism of $\widehat{R}$-modules
	   $\psi:\widehat{M}\rightarrow \widehat{M}^{\prime}$
	  such that $\;d^{\prime}= \psi \circ d\,$.	
  	  $$
	    \xymatrix{
		 & \widehat{R} \ar[rr]^-{d} \ar[dr]_-{d^{\prime}}
		     && \widehat{M}   \ar @{.>} [dl]^-{\psi}\\
		 &&    \widehat{M}^{\prime}&&.		  		
		}
	  $$			
   \end{itemize}
   (Thus, $\widehat{M}$ is unique up to a unique $\widehat{R}$-module isomorphism.)
  We denote this $\widehat{M}$ with $d: \widehat{R}\rightarrow \widehat{M}$
     by $\Omega_{\widehat{R}/\widehat{S}}$,
   with the built-in super-$C^k$-$\widehat{S}$-derivation
   $d:\widehat{R}\rightarrow \Omega_{\widehat{R}/\widehat{S}}$ understood.
}\end{definition}

\medskip

\begin{remark}$[\,$explicit construction of $\Omega_{\widehat{R}/\widehat{S}}$$\,]$. {\rm
 The $C^k$-cotangent module $\Omega_{\widehat{R}/\widehat{S}}$
    of $\widehat{R}$ over $\widehat{S}$
   can be constructed explicitly from the $\widehat{R}$-module generated by the set
   $$
     \{ d(\widehat{r})\,|\,  \widehat{r}\in \widehat{R}   \}\,,
   $$
   subject to the relations
   $$
    \begin{array}{llll}	
	 \mbox{({\it ${\Bbb Z}/2$-grading})}
	  &&& |d(\widehat{r})|\;=\; |\widehat{r}|\,,   \\[.6ex]
     \mbox{({\it right-${\Bbb R}$-linearity})}
	  &&& d(\widehat{r}\widehat{s} + \widehat{r}^{\prime}\widehat{s}^{\prime})\;
	            =\;   d(\widehat{r})\,\widehat{s}\,
				        +\, d(\widehat{r}^{\prime})\,\widehat{s}^{\prime}\,,   \\[.6ex]
	 \mbox{({\it super Leibniz rule$^{\prime}$})}
	  &&&  d(\widehat{r}\widehat{r}^{\prime})\;
	               =\;  d(\widehat{r})\,\widehat{r}^{\prime}\,
				          +\, (-1)^{|\widehat{r}||\widehat{r}^{\prime}|}\,
						          d(\widehat{r}^{\prime})\, \widehat{r} \,,   \\[.6ex]
	 \mbox{({\it bi-$\widehat{R}$-module structure})}
	  &&&  d(\widehat{r})\,\widehat{r}^{\prime}\;
	            =\;   (-1)^{|\widehat{r}||\widehat{r}^{\prime}|}\,
				               \widehat{r}^{\prime}\,d(\widehat{r})
    \end{array}
   $$	
   for all $\widehat{s},\,\widehat{s}^{\prime}\in  \widehat{S}$,
			 $\widehat{r},\, \widehat{r}^{\prime}$,  and
  $$
   \begin{array}{llll}
    \mbox{({\it chain rule})}\hspace{3em}
      &&&  d(h(r_1,\,\cdots\,,\,r_s))   \\
	  &&&   =\;  \partial_1 h(r_1,\,\cdots\,,\,r_s)\, d(r_1)\;
	                     +\; \cdots\; +\;
			           \partial_s h(r_1,\,\cdots\,,\,r_s)\, d(r_s)	
   \end{array}			
  $$
  for all $h\in C^k({\Bbb R}^s)$,
            $s\in {\Bbb Z}_{\ge 1}$, and
			$r_1,\,\cdots\,,\,r_s\in R\subset \widehat{R}$.
 Denote the image of $d(\widehat{r})$ under the quotient by $d\widehat{r}$.
 Then, by definition, the built-in map
  $$
    \begin{array}{ccccc}
	 d & :&  \widehat{R} & \longrightarrow  & \Omega_{\widehat{R}/\widehat{S}}   \\[.6ex]
	    &&   \widehat{r}  & \longmapsto        & d\widehat{r}
	\end{array}
  $$
  is a super-$C^k$-$\widehat{S}$-derivation from $\widehat{R}$
  to $\Omega_{\widehat{R}/\widehat{S}}$.
}\end{remark}

\medskip

\begin{remark}$[\,$relation to $\Omega_{R/S}$$\,]$. {\rm
 Note that the built-in inclusion $R\hookrightarrow \widehat{R}$ induces
  a natural $\widehat{R}$-module-homomorphism
   $\Omega_{R/S}\otimes_R\widehat{R}
                   \rightarrow \Omega_{\widehat{R}/\widehat{S}}$.
}\end{remark}

\medskip

\begin{remark}
$[\,$relation between $\sDer_{\widehat{S}}(\widehat{R})$
        and $\Omega_{\widehat{R}/\widehat{S}}$$\,]$.
{\rm
 The universal property of $\Omega_{\widehat{R}/\widehat{S}}$ implies that
   there is an $\widehat{R}$-module-isomorphism
  $$
    \sDer_{\widehat{S}}(\widehat{R})\;     \longrightarrow\;
	 \Hom_{\widehat{R}}(\Omega_{\widehat{R}/\widehat{S}},\widehat{R})\,.
  $$
}\end{remark}

\bigskip

\subsection{Super-$C^k$-manifolds, super-$C^k$-ringed spaces, and super-$C^k$-schemes}

The notion of super-$C^k$-manifolds, super-$C^k$-ringed spaces, and super-$C^k$-schemes
  are introduced in this subsection.
They are a super generalization of related notions from works of 
   Eduardo Dubuc [Du],
   Dominic Joyce [Joy], 
   Anders Kock [Ko], 
   Ieke Moerdijk and Gonzalo E.\ Reyes [M-R], 
   Juan Navarro Gonz\'{a}lez and Juan Sancho de Salas  [NG-SdS]
 in $C^{\infty}$-algebraic geometry or synthetic differential geometry.
The presentation here proceeds particularly with [Joy] in mind.

\clearpage

\begin{flushleft}
{\bf Super-$C^k$-manifolds}
\end{flushleft}
The notion of `{\it supermanifolds}' to capture the geometry
   (either the space-time over which fermionic fields are defined or the symmetry group itself)
  behind supersymmetric quantum field theory
  was studied by various authors, including
    Majorie Batchelor [Bat],
    Felix Berezin and Dimitry Leites [B-L],
	Bryce DeWitt [DeW],
    Bertram Kostant [Kos],
    Alice Rogers [Rog],
    Mitchell Rothstein [Rot],
  leading to several inequivalent notions of `supermanifolds'.
  
For this note, we follow the direction of
  Yuri Manin [Man]  and Steven Shnider and Raymond Wells, Jr., [S-W].
In essence,
   \begin{itemize}
    \item[{\Large $\cdot$}]
	{\it A supermanifold is a ringed space
	          with the underlying space an ordinary manifold $(X,{\cal O}_X)$, as a $C^k$-scheme, 
	        but with the structure sheaf
			  a ${\Bbb Z}/2$-graded ${\Bbb Z}/2$-commutative ${\cal O}_X$-algebra.}
   \end{itemize}
 Such setting (cf.\ the detail below) is
   both physically direct and compatible
   and mathematically in line with Grothendieck's Algebraic Geometry.
 It suits best for our purpose of generalization to
  the notion of matrix supermanifolds to describe coincident fermionic D-branes as maps therefrom.
 
\bigskip

\begin{definition} {\bf [super-$C^k$-manifold].} {\rm
 Let
   $(X,{\cal O}_X)$ be a $C^k$-manifold  of dimension $m$ and
   ${\cal F}$ be a  locally-free sheaf of ${\cal O}_X$-modules of rank $s$.
 Then
   $$\mbox{$
     \widehat{\cal O}_X\;
	   :=\;  \bigwedge^{\bullet}_{{\cal O}_X}{\cal F}
   $}$$
  is a sheaf of superpolynomial rings over ${\cal O}_X$
  (or equivalently a sheaf of Grassmann ${\cal O}_X$-algebras).
 The ${\Bbb Z}/2$-grading of $\widehat{\cal O}_X$ is given by setting
 $$\mbox{$
    \widehat{\cal O}_X^{\,\scriptsizeeven}\;
     :=\;  \bigwedge^{\scriptsizeeven}_{{\cal O}_X}{\cal F}
     	 \hspace{2em}\mbox{and}\hspace{2em}
   \widehat{\cal O}_X^{\,\scriptsizeodd}\;
     :=\;  \bigwedge^{\scriptsizeodd}_{{\cal O}_X}{\cal F}\,.  	 	
 $}$$
 The new ringed space
  $$
    \widehat{X}\; :=\; (X,\widehat{\cal O}_X)
  $$
  is called a {\it super-$C^k$-manifold}.
 The dimension $m$  of $X$ is called the {\it even dimension} of $\widehat{X}$
   while the rank $s$ of ${\cal F}$ is called the {\it odd dimension} of $\widehat{X}$.
 By construction, there is a {\it built-in split short exact sequence} of ${\cal O}_X$-modules
   (also as $\widehat{\cal O}_X$-modules)
   $$
     \xymatrix{
        \, 0  \ar[r]    &   \widehat{\cal I}_X  \ar[r]    & \widehat{O}_X \ar[r]
	       & {\cal O}_X \ar[r]   \ar@{.>}@/^/[l]   &  0  \, ,
      }
   $$
 where $\widehat{\cal I}_X  := \bigwedge^{\mbox{\tiny $\ge$}\,1}_{\;{\cal O}_X}{\cal F}$.
 In terms of ringed spaces, one has thus
  $$
    \xymatrix{
     \;\widehat{X} \ar@<.7ex>[rr]^-{\widehat{\pi}}
	      && X\ar@<.2ex>[ll]^-{\widehat{\iota}}\;,
	}
  $$
  where $\widehat{\pi}$ is a dominant morphism and $\widehat{\iota}$ is an inclusion
  such that $\widehat{\pi} \circ \widehat{\iota}=\Id_X$.
 Note that since in this case
   $\widehat{\cal I}_X$ is the nil-radical 
    (i.e.\ the ideal sheaf of all nilpotent sections) of $\widehat{\cal O}_X$,
  $X=\widehat{X}_{\redscriptsize}\subset \widehat{X}$.
 Cf.~Definition~2.3.15.
		
 A {\it morphism}
   $$
      \hat{f}:= (f, \widehat{f}^{\sharp})\;:\;
	     \widehat{X}:= (X, \widehat{\cal O}_X)\;
		 \longrightarrow\;  \widehat{\cal O}_Y:= (Y, \widehat{\cal O}_Y)
   $$
   between super-$C^k$-manifolds is a $C^k$-map $f: X \rightarrow Y$ between $C^k$-manifolds
   together with a sheaf-homomorphism
    $$
	 \widehat{f}^{\sharp}\; :\; f^{-1}\widehat{\cal O}_Y\;
	           \longrightarrow\;  \widehat{\cal O}_X
    $$
  that fits into the following commutative diagram
   $$
    \xymatrix{
	 &  f^{-1}\widehat{\cal O}_Y\ar[rr]^-{\widehat{f}^{\sharp}} \ar[d]
	      && \widehat{\cal O}_X\ar[d]\\
     &  f^{-1}{\cal O}_Y \ar@/^/@{.>}[u]  \ar[rr]^-{f^{\sharp}}
	      && {\cal O}_X \ar@/^/@{.>}[u]    & \hspace{-3em}.
	}
   $$
  Here,
    the homomorphism
  	 $f^{\sharp}:f^{-1}{\cal O}_Y\rightarrow {\cal O}_X$ between sheaves of $C^k$-rings on $X$
    is induced  by the $C^k$-map $f:X\rightarrow Y$ between $C^k$-manifolds.
 In terms of morphisms between ringed spaces, one has the commutative diagram:
  $$
    \xymatrix{
     &  \widehat{X} \ar@<.4ex>[d]^-{\widehat{\pi}} \ar[rr]^-{\widehat{f}}
                 &&     \widehat{Y} \ar@<.4ex>[d]^-{\widehat{\pi}} 	 \\
	 &  X\ar@<.4ex>[u]^-{\widehat{\iota}} \ar[rr]^-f
	             &&     Y\ar@<.4ex>[u]^-{\widehat{\iota}}  & \hspace{-3em}.
	}
  $$
 We will call $\hat{f}$ also a {\it $C^k$-map}.
}\end{definition}

\medskip

\begin{remark} {$[$supermanifold/superspace in physics: spinor and issue of central charge$\,]$.} {\rm
 The above Definition~2.3.1 of super-$C^k$-manifolds gives a general mathematical setting.
 However, for a physical application some refinement or extension of the above definition is required:
 \begin{itemize}
  \item[(1)]
   For a physical application to {\it $N=1$ supersymmetric quantum field theories},
    the $C^k$-manifold $X$ is equipped with a Riemannian or Lorentzian metric and
    the sheaf ${\cal F}$ {\it is a sheaf of spinors}
     that arises {\it from} the associated bundle of the orthonormal-frame bundle
	 to an {\it irreducible spinor representation} of the orthonormal group related to the metric.
   
  \item[(2)]
   For a physical application to {\it $N\ge 2$ supersymmetric quantum field theories},
    not only that
      the $C^k$-manifold $X$ is equipped with a Riemannian or Lorentzian metric and
      the sheaf ${\cal F}$ {\it is a sheaf of spinors}
       that arises {\it from} the associated bundle of the orthonormal-frame bundle
       to {\it a direct sum of $N$-many irreducible spinor representations}
 	   of the orthonormal group related to the metric,
	one has but also a new issue of whether to include the (bosonic) extension to
	  the structure sheaf $\widehat{\cal O}_X$ in Definition~2.3.1
	  that takes care also of the {\it central charges} in the $N\ge 2$ supersymmetry algebra;  
	  cf.\ [So], [W-B], and [W-O].
 \end{itemize}
 Such necessary refinement or extension should be made case by case to reflect physics.
 For the current note,
  our focus is on the notion of `differentiable maps from a matrix-supermanifold to a real manifold'.
 The framework we develop (cf.\ Sec.~3 and Sec.~4) is intact
   once a refined or extended structure sheaf $\widehat{\cal O}_X$ for a supermanifold is chosen.
}\end{remark}

\medskip

\begin{remark} $[\,$real spinor vs.\ complex spinor$\,]$. {\rm
 Also to reflect physics presentation,
  it may be more convenient case by case to take ${\cal F}$ in Definition~2.3.1
    to be an ${\cal O}_X^{\,\Bbb C}$-module, rather than an ${\cal O}_X$-module,
	when one constructs the structure sheaf $\widehat{\cal O}_X$.
 See, for example, Example~2.3.8 below.
 Again, our notion of `differentiable maps from a matrix-supermanifold to a real manifold' remains intact.
}\end{remark}
 
\bigskip

The following sample list  of superspaces is meant to give mathematicians a taste of
 the role of spinor representations in physicists' notion of a supermanifold
 beyond just a ${\Bbb Z}/2$-graded ${\Bbb Z}/2$-commutative manifold.
See, for example, [Freed: Lecture 3] of Daniel Freed.
 
\bigskip

\begin{example} {\bf [superspace ${\Bbb R}^{m|s}$ as super-$C^k$-manifold].} {\rm
 This is the supermanifold of topology ${\Bbb R}^m$ and
  function ring the superpolynomial ring $C^k({\Bbb R}^m)[\theta^1,\,\cdots\,,\theta^s]$
  over the $C^k$-ring $C^k({\Bbb R}^m)$.
 In particular, ${\Bbb R}^{0|s}$ is called a {\it superpoint}.
}\end{example}

\medskip

\begin{example} {\bf [$d=1+1$, $N=(2,2)$ superspace].} {\rm
 The supermanifold ${\Bbb R}^{2|4}$
   of the underlying space the $(1+1)$-dimensional Minkowski space-time ${\Bbb M}^{1+1}$ and
   of function ring the superpolynomial ring
     $C^{\infty}({\Bbb R}^2)
       [\theta^1, \theta^2,\bar{\theta}^{\dot{1}},\bar{\theta}^{\dot{2}}]$.
 Furthermore,
  each of $\theta^1$ and $\theta^2$
  (resp.\ $\bar{\theta}^{\dot{1}}$ and $\bar{\theta}^{\dot{2}}$)
   is in the left (resp.\ right) Majorana-Weyl spinor representation of $\SO(1,1)$.
 This is a basic superspace for a $d=2$ superconformal field theory and a superstring theory.
 }\end{example}
 
\medskip

\begin{example} {\bf [$d=2+1$, $N=1$ superspace].} {\rm
 The supermanifold ${\Bbb R}^{3|2}$
   of the underlying space the $(2+1)$-dimensional Minkowski space-time ${\Bbb M}^{2+1}$ and
   of function ring the superpolynomial ring
     $C^{\infty}({\Bbb R}^3)[\theta^1, \theta^2]$.
 Furthermore,
  the tuple $(\theta^1, \theta^2)$
   is in the Majorana spinor representation of $\SO(2,1)$
   ($\simeq$ the fundamental representation of $\SL(2,{\Bbb R})$).
 This is a basic superspace for a $d=2+1$, $N=1$ supersymmetric quantum field theory.
 }\end{example}

\medskip

\begin{example} {\bf [$d=3$, $N=1$ superspace].} {\rm
 The supermanifold ${\Bbb R}^{3|4}$
   of the underlying space the $3$-dimensional Euclidean space ${\Bbb E}^3$ and
   of function ring the superpolynomial ring
     $C^{\infty}({\Bbb R}^3)[\theta^1, \theta^2,\theta^3,\theta^4]$.
 Furthermore,
  the tuple $(\theta^1, \theta^2,\theta^3,\theta^4)$
   is in the pseudo-real spinor representation of $\SO(3)$.
 Notice how the signature of the metric may influence
  the dimension of a minimal/irreducible spinor representation at the same manifold dimension;
 cf.\ Example~2.3.6.
 }\end{example}
 
\medskip

\begin{example} {\bf [$d=3+1$, $N=1$ superspace].} {\rm
 The supermanifold ${\Bbb R}^{4|4}$
   of the underlying space the $3+1$-dimensional Minkowski space-time ${\Bbb M}^{3+1}$ and
   of function ring the superpolynomial ring
     $C^{\infty}({\Bbb R}^4)[\theta^1, \theta^2,\theta^3,\theta^4]$.
 Furthermore,
  the tuple $(\theta^1, \theta^2,\theta^3,\theta^4)$
  is in the real Majorana spinor representation of $\SO(3,1)$.
 In physics, it is more convenient to consider instead
  complex Weyl spinor representations of $\SO(3,1)$
    via the isomorphism $\Spin(3,1)\simeq \SL(2,{\Bbb C})$.
 In this case we take as the function ring
  $C^{\infty}({\Bbb R}^4)\otimes_{\Bbb R}
     [\theta^1,\theta^2, \bar{\theta}^{\dot{1}},\bar{\theta}^{\dot{2}}]^{\Bbb C}$,
  in which $(\theta^1,\theta^2)$	and $(\bar{\theta}^{\dot{1}},\bar{\theta}^{\dot{2}})$
   are in complex Weyl spinor representations of opposite chirality.	
 This is a basic superspace for a $d=3+1$, $N=1$ supersymmetric quantum field theory.
 See, for example, [W-B].
 }\end{example}

\medskip

\begin{example} {\bf [$d=3+1$, $N=2$ superspace].} {\rm
The supermanifold ${\Bbb R}^{4|8}$
   of the underlying space the $3+1$-dimensional Minkowski space-time ${\Bbb M}^{3+1}$ and
   of function ring the superpolynomial ring
     $C^{\infty}({\Bbb R}^4)
	    [\theta^1, \theta^2,\theta^3,\theta^4,
		 \theta^{\prime 1},\theta^{\prime 2}, \theta^{\prime 3},\theta^{\prime 4}]$.
 Furthermore,
  each of the tuples $(\theta^1, \theta^2,\theta^3,\theta^4)$
     and $(\theta^{\prime 1},\theta^{\prime 2}, \theta^{\prime 3},\theta^{\prime 4})$
   is in the real Majorana representation of $\SO(3,1)$. 	
 Here, we ignore the central charge in the $d=3+1$, $N=2$ super-Poincar\'{e} algebra.
 As in Example~2.3.8,
   for physics, it is more convenient to take as the function ring
  $C^{\infty}({\Bbb R}^4)\otimes_{\Bbb R}
     [\theta^1,\theta^2, \bar{\theta}^{\dot{1}},\bar{\theta}^{\dot{2}},
	  \theta^{\prime 1}, \theta^{\prime 2},
	  \bar{\theta}^{\prime \dot{1}}, \bar{\theta}^{\prime \dot{2}}]^{\Bbb C}$
  through complex Weyl representations of $\SO(3,1)$.	
 This is a basic superspace for  a $d=3+1$, $N=2$ supersymmetric quantum field theory.
 }\end{example}

\bigskip

\begin{flushleft}
{\bf Super-$C^k$-ringed spaces and super-$C^k$-schemes}
\end{flushleft}

\begin{definition} {\bf [super-$C^k$-ringed space, $C^k$-map].} {\rm
 A {\it super-$C^k$-ringed space}
   $$
      \widehat{X}\; =\;  (X,{\cal O}_X, \widehat{\cal O}_X)
   $$
   is a $C^k$-ringed space $(X,{\cal O}_X)$ together with a sheaf $\widehat{\cal O}_X$ of
   super-$C^k$-rings over ${\cal O}_X$.
 By construction, it has a built-in split short exact sequence of ${\cal O}_X$-modules
  (also as $\widehat{\cal O}_X$-modules)
    $$
     \xymatrix{
        \, 0  \ar[r]    &   \widehat{\cal I}_X  \ar[r]    & \widehat{O}_X \ar[r]
	       & {\cal O}_X \ar[r]   \ar@{.>}@/^/[l]   &  0  \, .
      }
    $$
   Here, $\widehat{\cal I}_X:=\Ker(\widehat{\cal O}_X\rightarrow {\cal O}_X)$
       is an ideal sheaf of $\hat{\cal O}_X$.
 The split short exact sequence defines the a pair of morphisms between the ringed spaces
  $$
    \xymatrix{
     \;\widehat{X} \ar@<.7ex>[rr]^-{\widehat{\pi}}
	      && X\ar@<.2ex>[ll]^-{\widehat{\iota}}\;,
	}
  $$
  where $\widehat{\pi}$ is a dominant morphism and $\widehat{\iota}$ is an inclusion
  such that $\widehat{\pi} \circ \widehat{\iota}=\Id_X$.

 A {\it morphism}
   $$
      \hat{f}:= (f,f^{\sharp}, \widehat{f}^{\sharp})\;:\;
	     \widehat{X}:= (X,{\cal O}_X,\widehat{\cal O}_X)\;
		 \longrightarrow\;  \widehat{\cal O}_Y:= (Y,{\cal O}_Y,\widehat{\cal O}_Y)
   $$
   between super-$C^k$-ringed spaces is a morphism
   $(f,f^{\sharp}):(X,{\cal O}_X)\rightarrow (Y,{\cal O}_Y)$ between $C^k$-ringed spaces
   together with a sheaf-homomorphism
    $$
	 \widehat{f}^{\sharp}\; :\; f^{-1}\hat{\cal O}_Y\;
	           \longrightarrow\;  \widehat{\cal O}_X
    $$
  that fits into the following commutative diagram
   $$
    \xymatrix{
	 &  f^{-1}\widehat{\cal O}_Y\ar[rr]^-{\widehat{f}^{\sharp}} \ar[d]
	      && \widehat{\cal O}_X\ar[d]\\
     &  f^{-1}{\cal O}_Y \ar@/^/@{.>}[u]  \ar[rr]^-{f^{\sharp}}
	      && {\cal O}_X \ar@/^/@{.>}[u]    & \hspace{-3em}.
	}
   $$
 That is, a commutative diagram of morphisms between ringed spaces:
  $$
    \xymatrix{
     &  \widehat{X} \ar@<.4ex>[d]^-{\widehat{\pi}} \ar[rr]^-{\widehat{f}}
                 &&     \widehat{Y} \ar@<.4ex>[d]^-{\widehat{\pi}} 	 \\
	 &  X\ar@<.4ex>[u]^-{\widehat{\iota}} \ar[rr]^-f
	             &&     Y\ar@<.4ex>[u]^-{\widehat{\iota}}  & \hspace{-3em}.
	}
  $$
 We will call $\hat{f}$ also a {\it $C^k$-map}.
}\end{definition}

\medskip

\begin{definition} {\bf [affine super-$C^k$-scheme].} {\rm
 Let $\widehat{R}$ be a super-$C^k$-ring with the structure ring-homomorphisms
  $\xymatrix{ \widehat{R}\ar[r]   & R \ar@{.>}@/^/[l]    }$.
 The {\it affine super-$C^k$-scheme associated to $\hat{R}$}  is a super-$C^k$-ringed space
   $(X,{\cal O}_X,\widehat{\cal O}_X)$ defined as follows:
  \begin{itemize}
   \item[{\Large $\cdot$}]
    $(X,{\cal O}_X)$ is the affine $C^k$-scheme associated to $R$.
	
   \item[{\Large $\cdot$}]	
    $\widehat{\cal O}_X$ is the quasi-coherent sheaf on $X$
	   associated to $\widehat{R}$ as an $R$-module,
	  as defined in $C^k$-algebraic geometry via localizations of $\widehat{R}$ at subsets of $R$.
  \end{itemize}
 By construction,
  $\widehat{\cal O}_X$ is a sheaf of super-$C^k$-rings over ${\cal O}_X$   and
  the ring-homomorphisms $\xymatrix{ \widehat{R}\ar[r]   & R \ar@{.>}@/^/[l]    }$
    induce a split short exact sequence
	$$
     \xymatrix{
        \, 0  \ar[r]    &   \widehat{\cal I}_X  \ar[r]    & \widehat{O}_X \ar[r]
	       & {\cal O}_X \ar[r]   \ar@{.>}@/^/[l]   &  0  \,,
      }
    $$
 which defines a pair of built-in morphisms
  $$
    \xymatrix{
     \;\widehat{X} \ar@<.7ex>[rr]^-{\widehat{\pi}}
	      && X\ar@<.2ex>[ll]^-{\widehat{\iota}}\;,
	}
   $$
  where $\widehat{\pi}$ is a dominant morphism and $\widehat{\iota}$ is an inclusion
  such that $\widehat{\pi} \circ \widehat{\iota}=\Id_X$.
	
 A {\it morphism}
   $(X,{\cal O }_X, \widehat{\cal O}_X)\rightarrow (Y,{\cal O}_Y,\widehat{\cal O}_Y)$
   between affine super-$C^k$-schemes
  is defined to be a $C^k$-map
   $\widehat{f}:
     (X,{\cal O }_X,\widehat{\cal O}_X)\rightarrow (Y,{\cal O}_Y,\widehat{\cal O}_Y)$
	of the underlying super-$C^k$-ringed spaces. 
}\end{definition}

\medskip

\begin{remark} $[\,$super-$C^k$-ring vs.\ affine super-$C^k$-scheme$\,]$. {\rm 
 For our purpose, 
  we shall assume that all the $C^k$-rings in our discussion are finitely generated and germ-determined.
 In this case, the category of affine super-$C^k$-schemes
      is contravariantly equivalent to the category of super-$C^k$-rings.
 See, for example, [Joy: Proposition 4.15].	  
}\end{remark}	  
	  
\medskip

\begin{definition} {\bf [super-$C^k$-scheme].} {\rm
 A super-$C^k$-ringed space $\widehat{X}:= (X,{\cal O}_X, \widehat{\cal O}_X)$
    is called a {\it super-$C^k$-scheme}
  if $X$ admits an open-set covering $\{U_{\alpha}\}_{\alpha\in A}$
   such that
     $(U_{\alpha}, {\cal O}_X|_{U_{\alpha}}, \widehat{\cal O}_X|_{U_{\alpha}})$
	 is an affine super-$C^k$-scheme for all $\alpha\in A$.	
 By construction,
  $\widehat{\cal O}_X$ is a sheaf of super-$C^k$-rings over ${\cal O}_X$
  with a built-in a split short exact sequence
	$$
     \xymatrix{
        \, 0  \ar[r]    &   \widehat{\cal I}_X  \ar[r]    & \widehat{O}_X \ar[r]
	       & {\cal O}_X \ar[r]   \ar@{.>}@/^/[l]   &  0  \,,
      }
    $$
 which defines a pair of built-in morphisms
  $$
    \xymatrix{
     \;\widehat{X} \ar@<.7ex>[rr]^-{\widehat{\pi}}
	      && X\ar@<.2ex>[ll]^-{\widehat{\iota}}\;,
	}
   $$
  where $\widehat{\pi}$ is a dominant morphism and $\widehat{\iota}$ is an inclusion
  such that $\widehat{\pi} \circ \widehat{\iota}=\Id_X$.
		
 A {\it morphism}
   $(X,{\cal O }_X, \widehat{\cal O}_X)\rightarrow (Y,{\cal O}_Y,\widehat{\cal O}_Y)$
   between super-$C^k$-schemes
  is defined to be a $C^k$-map
   $\widehat{f}:
     (X,{\cal O }_X, \widehat{\cal O}_X)\rightarrow (Y,{\cal O}_Y,\widehat{\cal O}_Y)$
	of the underlying super-$C^k$-ringed spaces.
}\end{definition}

\medskip

\begin{remark}$[\,$super-$C^k$-scheme vs.\ equivalence class of gluing systems of super-$C^k$-rings$\,]$.
 {\rm
 Recall Remark~2.3.12.
 Under the assumption that all the $C^k$-rings in our discussion be finitely generated and germ-determined, 
 the category of super-$C^k$-schemes is contravariantly equivalent to 
  the category of equivalence classes of gluing systems of super-$C^k$-rings.	
}\end{remark}

\medskip

\begin{definition} {\bf [super-$C^k$-normal ideal sheaf and super-$C^k$-subscheme].} {\rm
 Let $\widehat{X}:= (X,{\cal O}_X, \widehat{\cal O}_X)$ be a super-$C^k$-scheme.
 A {\it super-$C^k$-normal ideal sheaf} $\widehat{\cal I}$ on $\widehat{X}$ 
   is a sheaf of super-$C^k$-normal ideals of $\widehat{\cal O}_X$.
 In this case (and only in this case), $\widehat{\cal I}$ defines 
  a {\it super-$C^k$-subschme} $\widehat{Z}:=(Z,{\cal O}_Z,\widehat{\cal O}_Z)$ 
    of $\widehat{X}$ 
    with a built-in commutative diagram of $\widehat{\cal O}_X$-modules
  $$
   \xymatrix{
    0 \ar[r]   
	   & \widehat{\cal I}\ar[r]
       & \widehat{\cal O}_X \ar[r]  \ar@{->>}[d]  
	   & \widehat{\cal O}_Z \ar[r]  \ar@{->>}[d]                                                              & 0  \\
    0 \ar[r]
       & {\cal I}:= \widehat{\cal I}_0\cap {\cal O}_X   \rule{0ex}{1.2em}  
           	   \ar[r] \ar@{^{(}->}[u]
       &{\cal O}_X \rule{0ex}{1em} \ar[r] \ar@<.7ex>@{^{(}->}[u]                    
	   &  {\cal O}_Z  \rule{0ex}{1em}\ar[r]  \ar@<.7ex>@{^{(}->}[u]          & 0
   }
  $$    
  such that both horizontal sequences are exact.  
 Here 
   $\widehat{\cal I}=\widehat{\cal I}_0+\widehat{\cal I}_1$ 
      is the decomposition of $\widehat{\cal I}$ by its homogeneous components.   
}\end{definition}

\medskip

\begin{remark} {$[$sheaf-type super-thickening$]$.} {\rm
 One may think of
    the graded-commutative scheme $\widehat{X}$
     as a {\it sheaf-type super-thickening} of the underlying $C^k$-scheme $X$,
  and the morphism $\widehat{f}:\widehat{X}\rightarrow \widehat{Y}$
   as a lifting of $C^k$-map $f:X\rightarrow Y$.
 Cf.~{\sc Figure}~3-1.
}\end{remark}

\bigskip

\subsection{Sheaves of modules and differential calculus on a super-$C^k$-scheme}	

Basic notions of sheaves of modules and differential calculus on a super-$C^k$-scheme
are collected in this subsection.

\bigskip

\begin{flushleft}
{\bf Sheaves of modules on a super-$C^k$-scheme}	
\end{flushleft}
\begin{definition} {\bf [sheaf of modules on super-$C^k$-scheme].} {\rm
 Let $\widehat{X}=(X,{\cal O}_X, \widehat{\cal O}_X)$ be a super-$C^k$-scheme.
 A {\it sheaf of modules} on $\widehat{X}$ (i.e.\ {\it $\widehat{\cal O}_X$-module})
  is defined to be a sheaf $\widehat{\cal F}$ of modules
  on the (${\Bbb Z}/2$-graded ${\Bbb Z}/2$-commutative) ringed space that underlies $\widehat{X}$.
 In particular,
  $\widehat{\cal F}$ is ${\Bbb Z}/2$-graded
     $\widehat{\cal F}=\widehat{\cal F}_0\oplus\widehat{\cal F}_1$.
  Sections of $\widehat{\cal F}_0$  (resp.\ $\widehat{\cal F}_1$) is said to be {\it even}
   (resp.\ {\it odd}).
  Sections of either $\widehat{\cal F}_0$ or $\widehat{\cal F}_1$ is said to be {\it homogeneous}.	
 The left-$\widehat{\cal O}_X$-module structure on $\widehat{\cal F}$
   induces the right-$\widehat{\cal O}_X$-structure on $\widehat{\cal F}$ and vice versa;
 thus we won't distinguish left-, right-, or bi-module structures in our discussions.
	
 The notion of
  
      \medskip
	
    {\Large $\cdot$}
	   {\it homomorphism} $\widehat{\cal F} \rightarrow \widehat{\cal G}$
	    of $\widehat{\cal O}_X$-modules,

      \smallskip
	
	{\Large $\cdot$}
       {\it submodule} $\widehat{\cal G} \hookrightarrow \widehat{\cal F}$,
	   (cf.\ {\it monomorphism}),
	
	  \smallskip
	
	{\Large $\cdot$}
       {\it quotient module}  $\widehat{\cal F}\twoheadrightarrow \widehat{\cal G}$,
	   (cf.\ {\it epimorphism}), 	
	
	  \smallskip
	
	{\Large $\cdot$}
	   {\it direct sum}  $\widehat{\cal F}\oplus \widehat{\cal G}$
	   of $\widehat{\cal O}_X$-modules,
	
	  \smallskip
	
    {\Large $\cdot$}
      {\it tensor product}
	    $\widehat{\cal F}\otimes_{\widehat{\cal O}_X}\!\widehat{\cal G}$
		of $\widehat{\cal O}_X$-modules,
	
	  \smallskip
	
    {\Large $\cdot$}
	    {\it finitely generated}:
    		if $\widehat{\cal O}_X^{\,\oplus l}  \twoheadrightarrow \widehat{\cal F}$
		     exists for some $l$,
	
      \smallskip
	
    {\Large $\cdot$}
        {\it finitely presented}:
           if $\widehat{\cal O}_X^{\,\oplus l^{\prime}}
		          \rightarrow \widehat{\cal O}_X^{\,\oplus l}
				  \rightarrow \widehat{\cal F} \rightarrow 0$
              is exact for some $l$, $l^{\prime}$				
	
      \medskip
	
	\noindent
    are all defined in the ordinary way as in commutative algebraic geometry.
  A homomorphism $\widehat{\cal F} \rightarrow \widehat{\cal G}$
	    of $\widehat{\cal O}_X$-modules is said to be {\it  even} (resp.\ {\it odd})
	if it preserves (resp.\ switches) the parity of homogeneous sections.	
	
 Denote by $\ModCategory(\widehat{X})$ the category of all $\widehat{\cal O}_X$-modules.
}\end{definition}

\medskip

\begin{definition} {\bf [push-forward and pull-back of sheaf of modules].} {\rm
 Let
    $\widehat{f}:\widehat{X}\rightarrow \widehat{Y}$ be a morphism
       between super-$C^k$-schemes  and
	$\widehat{\cal F}$ (resp.\ $\widehat{\cal G}$) be a $\widehat{\cal O}_X$-module
     (resp.\ $\widehat{\cal O}_Y$-module).
 Then, the structure sheaf-of-rings homomorphism
    $\widehat{f}^{\sharp}: f^{-1}\widehat{\cal O}_Y\rightarrow \widehat{\cal O}_X$
	renders $\widehat{\cal F}$ an $\widehat{\cal O}_Y$-module.
 It is called the {\it push-forward} of $\widehat{\cal F}$ by $\widehat{f}$ and
	 is denoted by
	   $\widehat{f}_{\ast}(\widehat{\cal F})$ or $\widehat{f}_{\ast}\widehat{\cal F}$.
 The inverse-image sheaf  $f^{-1}\widehat{\cal G}$ of $\widehat{\cal G}$ under $f$
    is an $f^{-1}\widehat{\cal O}_Y$-module.
 Define
     the {\it pull-back}
	   $\widehat{f}^{\ast}\widehat{\cal G}$	(or $\widehat{f}^{\ast}\widehat{\cal G}$)
	 of $\widehat{\cal G}$ under $\widehat{f}$
   to be the $\widehat{\cal O}_X$-module
    $f^{-1}\widehat{\cal G}\otimes_{f^{-1}\widehat{\cal O}_Y}\widehat{\cal O}_X$.
}\end{definition}

\bigskip

Let
  $\widehat{X}=(X, {\cal O}_X, \widehat{\cal O}_X)$ be an affine super-$C^k$-scheme
      associated to a super-$C^k$-ring $\widehat{R}\supset R$   and
   $\widehat{M}$ be an $\widehat{R}$-module.
 Then, the assignment
    $U\mapsto       \widehat{M}\otimes_{\widehat{R}}\widehat{\cal O}_X(U)$,
  with the restriction map $\Id_{\widehat{M}}\otimes  \rho_{UV}$ for $V\subset U$,
   where $\rho_{UV}:\widehat{\cal O}_X(U)\rightarrow \widehat{\cal O}_X(V)$
    is the restriction map of $\widehat{\cal O}_X$,
  is a presheaf on $\widehat{X}$.
 Let $\widehat{M}^{\sim}$ be its sheafification.

\bigskip
 
\begin{definition} {\bf [quasi-coherent sheaf on affine super-$C^k$-scheme].} {\rm
 The sheaf $\widehat{M}^{\sim}$ of $\widehat{\cal O}_X$-modules
   on the affine super-$C^k$-scheme $\widehat{X}$ thus obtained
  from the $\widehat{R}$-module $\widehat{M}$ is called a {\it quasi-coherent sheaf}
  on $\widehat{X}$.
}\end{definition}
 
\medskip

\begin{definition} {\bf [quasi-coherent sheaf on super-$C^k$-scheme].} {\rm
 Let $\widehat{X}=(X,{\cal O}_X, \widehat{\cal O}_X)$ be a super-$C^k$-scheme.
 An $\widehat{\cal O}_X$-module $\widehat{\cal F}$ is said to be {\it quasi-coherent}
  if $X$ admits an open-set covering $\{U_{\alpha}\}_{\alpha\in A}$
   such that
     $(U_{\alpha}, {\cal O}_X|_{U_{\alpha}}, \widehat{\cal O}_X|_{U_{\alpha}})$
	 is an affine super-$C^k$-scheme   and
     $\widehat{\cal F}|_{U_{\alpha}}$ is quasi-coherent in the sense of Definition~2.4.3
   for all $\alpha$.	
																				
 Denote by $\QCohCategory(\widehat{X})$
   the category of all quasi-coherent $\widehat{\cal O}_X$-modules.
}\end{definition}
 
\bigskip

Let $\widehat{X}=(X,{\cal O}_X, \widehat{\cal O}_X)$ be a super-$C^k$-scheme.
Recall the built-in dominant morphism $\widehat{\pi}: \widehat{X} \rightarrow X$.
The following immediate lemmas relate modules over $\widehat{X}$ and those over $X$:
 
\bigskip

\begin{lemma}
{\bf [quasi-coherent sheaf on super-$C^k$-scheme vs.\ on $C^k$-scheme].}
 An $\widehat{\cal O}_X$-module $\widehat{\cal F}$ is quasi-coherent
   if and only if $\widehat{\pi}_{\ast}(\widehat{\cal F})$ is quasi-coherent
    on $X$.	
 The push-forward functor
  $$
     \widehat{\pi}_{\ast}\;:\;
	   \ModCategory_{\widehat{X}}\;  \longrightarrow\;   \ModCategory_X
  $$
  is exact and takes $\QCohCategory(\widehat{X})$ to $\QCohCategory(X)$.
\end{lemma}
 
\medskip

\begin{lemma}{\bf [natural push-pull relation].}
 Let
    $\widehat{f}:\widehat{X}\rightarrow \widehat{Y}$ be a morphism
       between super-$C^k$-schemes  and
	$\widehat{\cal F}$ (resp.\ $\widehat{\cal G}$) be a $\widehat{\cal O}_X$-module
     (resp.\ $\widehat{\cal O}_Y$-module).
 Then, there are canonical isomorphisms
  $$
      f_{\ast}(\widehat{\pi}_{\ast}\widehat{\cal F})\;
        \simeq\;
		   \widehat{\pi}_{\ast}(\widehat{f}_{\ast}\widehat{\cal F})
       \hspace{2em}\mbox{and}\hspace{2em}		
     f^{\ast}(\widehat{\pi}_{\ast}\widehat{\cal G})\;
        \stackrel{\sim}{\rightarrow}\;
		   \widehat{\pi}_{\ast}(\widehat{f}^{\ast}\widehat{\cal G})\,.
  $$
\end{lemma}

\bigskip
 
\begin{flushleft}
{\bf Differential calculus on a super-$C^k$-scheme}	
\end{flushleft}
\begin{definition} {\bf [tangent sheaf and cotangent sheaf].}  {\rm
 Let $\widehat{X}=(X,{\cal O}_X,\widehat{\cal O}_X)$ be a super-$C^k$-scheme.
 (1)
 The presheaf on $\widehat{X}$  that associates to an open set $U\subset X$
  the $\widehat{\cal O}_X(U)$-module
  $\Der_{\Bbb R}(\widehat{\cal O}_X(U))$ is a quasi-coherent sheaf of
   $\widehat{\cal O}_X$-modules,
   denoted by ${\cal T}_{\ast}\widehat{X}$
   and called interchangeably the {\it tangent sheaf of $\widehat{X}$} or
    the {\it sheaf of super-$C^k$-derivations on $\widehat{\cal O}_X$}.
  
 (2)
 The presheaf on $\widehat{X}$  that associates to an open set $U\subset X$
  the $\widehat{\cal O}_X(U)$-module $\Omega_{\widehat{\cal O}_X(U)/{\Bbb R}}$
  is a quasi-coherent sheaf of $\widehat{\cal O}_X$-modules,
  denoted by ${\cal T}^{\ast}\widehat{X}$
   and called interchangeably the {\it cotangent sheaf of $\widehat{X}$} or
    the {\it sheaf of differentials of $\widehat{\cal O}_X$}.
 By construction, there is a canonical even map
  $d: \widehat{\cal O}_X\rightarrow {\cal T}^{\ast}\widehat{X}$
   as sheaves of ${\Bbb R}$-vector spaces on $\widehat{X}$.
 }\end{definition}

\bigskip

It follows from the local study in Sec.~2.2 that
 there is a canonical isomorphism
  $$
    {\cal T}_{\ast}\widehat{X}\;  \longrightarrow\; 
         \Homsheaf_{\widehat{\cal O}_X}
		     ({\cal T}^{\ast}\widehat{X}, \widehat{\cal O}_X)
  $$
 as $\widehat{\cal O}_X$-modules.
Cf.\  Remark 2.2.17.

\bigskip

\section{Azumaya/matrix super-$C^k$-manifolds with a fundamental\\ module}

Once the basics of super-$C^k$-rings are laid down (Sec.~2.2),
 the extension of the notion of super-$C^k$-rings to the notion of Azumaya/matrix super-$C^k$-rings
 proceeds in the same manner
 as the extension of the notion of $C^k$-rings to the notion of Azumaya/matrix $C^k$-rings
  in [L-Y3] (D(11.1)).
After localizations at a subset in the center of the rings in question and then gluings from local to global,
 the extension of the notion of super-$C^k$-manifolds
   to the notion of Azumaya/matrix super-$C^k$-manifolds with a fundamental module
 proceeds in the same manner
 as the extension of the notion of $C^k$-manifold
   to the notion of Azumaya/matrix $C^k$-manifolds with a fundamental module
 in ibidem.
A brief review is given below for the introduction of terminology and notations we need
  and the completeness of the note.

\bigskip

\begin{flushleft}
{\bf Azumaya/matrix algebras over a super-$C^k$-ring, modules, and differential calculus}
\end{flushleft}
Let
  $\widehat{R}$ be a super-$C^k$-ring  and
  $\widehat{R}^{\,\Bbb C}:= \widehat{R}\otimes_{\Bbb R}{\Bbb C}$ be its complexification.
The ({\it complex}) {\it matrix algebra}   of rank $r$ over $\widehat{R}^{\,\Bbb C}$
 is the algebra $M_{r\times r}(\widehat{R}^{\,\Bbb C})$   of $r\times r$ matrices
  with entries elements in $\widehat{R}^{\,\Bbb C}$.
The addition and the multiplication of elements of $M_{r\times r}(\widehat{R}^{\,\Bbb C})$
 are defined through the matrix addition and the matrix multiplication in the ordinary way:
 $$\mbox{$
   \left(\, \rule{0ex}{1em}\widehat{a}_{ij}\, \right)_{ij}\,
     +\,  \left(\,\widehat{b}_{ij}\,\right)_{ij}\;
     :=\;  \left(\,\widehat{a}_{ij}+\widehat{b}_{ij}\,\right)_{ij}
	 \hspace{2em}\mbox{and}\hspace{2em}
   \left(\, \rule{0ex}{1em}\widehat{a}_{ij}\, \right)_{ij}\,
     \cdot\,  \left(\, \widehat{b}_{ij}\, \right)_{ij}\;
     :=\;  \left(\sum_{k=1}^r\widehat{a}_{ik}\widehat{b}_{kj}\right)_{ij}\,,
	 $}
 $$
 where $\widehat{a}_{ij}+\widehat{b}_{ij}$
   and $\widehat{a}_{ik}\cdot\widehat{b}_{kj}$ are respectively the addition and the multiplication
   in $\widehat{R}^{\,\Bbb C}$.
As an abstract ({\it complex}) {\it Azumaya algebra} over $\widehat{R}$,
 $$
  M_{r\times r}(\widehat{R}^{\,\Bbb C})\; \simeq\;
    \frac{\widehat{R}^{\,\Bbb C}\langle\, e_i^{\;j} \,|\, 1\le i,j\le r \,\rangle   }
	 {(\,   \widehat{r}e_i^{\;j}-e_i^{\;j}\widehat{r}\,,\,
	        e_i^{\;j}e_{i^{\prime}}^{\;j^{\prime}}
			  - \delta_{i^{\prime}}^j\, e_i^{\;j^{\prime}}\;|\;
		\widehat{r}\in \widehat{R}^{\,\Bbb C}\,,\, 1 \le i, j, i^{\prime}, j^{\prime}\le r \,)}\,.
 $$
 Here
   $\widehat{R}^{\,\Bbb C}\langle\, e_i^{\;j} \,|\, 1\le i,j\le r \,\rangle$
     is the unital associative algebra generated by
	   $\widehat{R}^{\,\Bbb C}$ and the set $\{e_i^{\;j} \,|\, 1\le i,j\le r\}$    and
   $(\,   \widehat{r}e_i^{\;j}-e_i^{\;j}\widehat{r}\,,\,
	        e_i^{\;j}e_{i^{\prime}}^{\;j^{\prime}}
			  - \delta_{i^{\prime}}^j\, e_i^{\;j^{\prime}}\;|\;
		 \widehat{r}\in \widehat{R}^{\,\Bbb C}\,,\, 1 \le i, j, i^{\prime}, j^{\prime}\le r \,)$
     is the bi-ideal thereof generated by the elements indicated.
The ${\Bbb Z}/2$-grading of $M_{r\times r}(\widehat{R}^{\,\Bbb C})$
  follows from the ${\Bbb Z}/2$-grading of $\widehat{R}^{\,\Bbb C}$
  by assigning in addition the parity of $e_i^{\;j}$ to be even.
As ${\Bbb Z}/2$-graded ${\Bbb C}$-algebras, one has the isomorphism
 $$
   M_{r\times r}(\widehat{R}^{\,\Bbb C})\;
    \simeq\;   M_{r\times r}({\Bbb C})\otimes_{\Bbb C}\widehat{R}^{\,\Bbb C}\,.
 $$
  	 	
$M_{r\times r}(\widehat{R}^{\,\Bbb C})$ naturally represents
  on the free module $\widehat{R}^{\,\Bbb C}$-module
  $\widehat{F}:=(\widehat{R}^{\,\Bbb C})^{\oplus r}$ of rank $r$,
  by the matrix multiplication to the left on a column vector.
We will call $\widehat{F}$ the {\it fundamental module}  of
  $M_{r\times r}(\widehat{R}^{\,\Bbb C})$.
The dual
 $\widehat{F}^{\vee}
    := \Hom_{\widehat{R}^{\,\Bbb C}}(\widehat{F},\widehat{R}^{\,\Bbb C})$
   of $\widehat{F}$, as a $\widehat{R}^{\,\Bbb C}$-module,
 is a right-$M_{r\times r}(\widehat{R}^{\,\Bbb C})$-module.
Denote by $M_{r\times r}(\widehat{R}^{\,\Bbb C})\mbox{-}\ModCategory$
  the category of left-$M_{r\times r}(\widehat{R}^{\,\Bbb C})$-modules.
Then there are natural functors
 $$
   \begin{array}{cccl}
   \widehat{R}^{\,\Bbb C}\mbox{-}\ModCategory
     & \longrightarrow
	 & M_{r\times r}(\widehat{R}^{\,\Bbb C})\mbox{-}\ModCategory \\[.6ex]
	\widehat{M}
	 & \longmapsto
	 & \widehat{F}\otimes_{\widehat{R}^{\,\Bbb C}}\widehat{M}
   \end{array}	
 $$
  and
 $$
   \begin{array}{cccl}
   M_{r\times r}(\widehat{R}^{\,\Bbb C})\mbox{-}\ModCategory
     & \longrightarrow
	 & \widehat{R}^{\,\Bbb C}\mbox{-}\ModCategory            \\[.6ex]
	\widehat{N}    & \longmapsto
	 & \widehat{F}^{\vee}\otimes_{M_{r\times r}(\widehat{R}^{\,\Bbb C})}\widehat{N}\,.
   \end{array}	
 $$
When $\widehat{R}$ is the function-ring of a super-$C^k$-manifold,
 these two functors render $\widehat{R}^{\,\Bbb C}\mbox{-}\ModCategory$ and
  $M_{r\times r}(\widehat{R}^{\,\Bbb C})\mbox{-}\ModCategory$
 equivalent, called the {\it Morita equivalence}.

Let $\widehat{R}$ be a super-$C^k$-ring over another super-$C^k$-ring $\widehat{S}$.
Then, under the isomorphism
  $M_{r\times r}(\widehat{R}^{\,\Bbb C})
      \simeq M_{r\times r}({\Bbb C})\otimes_{\Bbb C}\widehat{R}^{\,\Bbb C}$
  between ${\Bbb C}$-algebras,
 $$
  \begin{array}{crcl}
   & \sDer_{\widehat{S}}(M_{r\times r}(\widehat{R}^{\,\Bbb C}))
    & \simeq
    & \Der_{\Bbb C}(M_{r\times r}({\Bbb C}))
		                                                     \otimes_{\Bbb C}\widehat{R}^{\,\Bbb C}
        \oplus															
		\Id_{r\times r}	   \otimes_{\Bbb C}\sDer_{\widehat{S}}(\widehat{R})^{\Bbb C}
	      											 \\[1.2ex]
   \mbox{and}\hspace{2em}
    & \Omega_{M_{r\times r}(\widehat{R}^{\,\Bbb C})/\widehat{S}^{\,\Bbb C}}
	& \simeq
	& \Omega_{M_{r\times r}({\Bbb C})/{\Bbb C}}
	                                                   \otimes_{\Bbb C}\widehat{R}^{\,\Bbb C}
         \oplus
         M_{r\times r}({\Bbb C})
		   \otimes_{\Bbb C}\Omega_{\widehat{R}/\widehat{S}}^{\,\Bbb C }\,.
		   \hspace{7em}		 	
  \end{array}
 $$
Caution that
  the former is only a $\widehat{R}^{\,\Bbb C}$-module
  while the latter is an $M_{r\times r}(\widehat{R}^{\,\Bbb C})$-module.
 
\bigskip

\begin{sremark} $[\,$center of $M_{r\times r}(\widehat{R}^{\,\Bbb C})$$\,]$. {\rm
 Despite being only ${\Bbb Z}/2$-commutative,
 conceptually it is instructive to regard $\widehat{R}^{\,\Bbb C}$
  as playing the role of the center of the matrix ring $M_{r\times r}(\widehat{R}^{\,\Bbb C})$.
}\end{sremark}

\bigskip

\begin{flushleft}
{\bf Azumaya/matrix super-$C^k$-manifolds with a fundamental module}
\end{flushleft}
The notion of Azumaya/matrix manifolds ([L-Y3] (D(11.1))) and the notion of supermanifolds (Sec.~2.3)
 can be merged into the notion of Azumaya/matrix supermanifolds:
 $$
  \xymatrix @C=-5em{
    &    &  \{\,\mbox{\it $C^k$-manifolds}\,\}\raisebox{-2ex}{\rule{0ex}{1ex}}\hspace{1em}
	                \ar@{_{(}->}[dl]  \ar@{^{(}->}[dr] &    \\
	&  \{\,\mbox{\it Azumaya/matrix $C^k$-manifolds}\,\}\raisebox{-2ex}{\rule{0ex}{1ex}}
	                \ar@{^{(}->}[dr]  &
	       & \{\,\mbox{\it super-$C^k$-manifolds}\,\}\raisebox{-2ex}{\rule{0ex}{1ex}}
		       \ar@{_{(}->}[dl]\\
    &    &  \{\,\mbox{\it Azumaya/matrix super-$C^k$-manifolds}\,\}   &  &.		
  }
 $$

\bigskip
 
\begin{sdefinition} {\bf [Azumaya/matrix super-$C^k$-manifold with a fundamental module].} {\rm
 An {\it Azumaya} (or {\it matrix}) {\it super-$C^k$-manifold with a fundamental module}
   is the following tuple:
  $$
    (X,{\cal O}_X, \widehat{\cal O}_X,
	         \widehat{\cal O}_X^{A\!z}
			    :=\Endsheaf_{\widehat{\cal O}_X^{\,\Bbb C}}(\widehat{\cal E}), {\cal E})\;
				=:\;  (\widehat{X}^{\!A\!z},\widehat{\cal E})\,,
  $$
  where
   \begin{itemize}
    \item[{\Large $\cdot$}]
	 $\widehat{X}:=(X,{\cal O}_X, \widehat{\cal O}_X)$ is a super-$C^k$-manifold;
	
	\item[{\Large $\cdot$}]
	 ${\cal E}$ is a locally free ${\cal O}_X^{\,\Bbb C}$-module of finite rank, say, $r$;
	
	\item[{\Large $\cdot$}]
	 $\widehat{\cal E}\;
	    :=\;  {\cal E}
	                 \otimes_{{\cal O}_X}\widehat{\cal O}_X\;
		 =\;   \widehat{\pi}^{\ast}{\cal E}\,$,
	  where
	  $\xymatrix{
	       \widehat{X} \ar@<-.2ex>[r]_-{\widehat{\pi}}
	        & X \ar@<-.5ex> @{_{(}->}[l]_-{\widehat{\iota}}  }$ 	
	   are the built-in morphisms for $\widehat{X}$.
   \end{itemize}
 Note that  one has the canonical isomorphism
  $$
    \widehat{\cal O}_X^{A\!z}\;
	:=\;             \Endsheaf_{\widehat{\cal O}_X^{\,\Bbb C}}(\widehat{\cal E})\;
	\simeq\;   \Endsheaf_{{\cal O}_X^{\,\Bbb C}}({\cal E})
	                                                       \otimes_{{\cal O}_X}\widehat{\cal O}_X\;
    =\;            {\cal O}_X^{A\!z}\otimes_{{\cal O}_X}\widehat{\cal O}_X
	=:\; 			\widehat{\pi}^{\ast}{\cal O}_X^{A\!z}\,.
  $$
}\end{sdefinition}

\bigskip

Built into the definition is the following commutative diagram of morphisms:
$$
  \xymatrix{
  & && \widehat{X}^{\!A\!z} \ar@<.4ex>[dll]^-{\dot{\widehat{\pi}}}
                                                             \ar[drr]^-{\dot{\pi}^{A\!z}} \\
  & X^{\!A\!z}\ar @<.4ex>@ {^{(}->}[urr]^-{\dot{\widehat{\iota}}} \ar[drr]^-{\pi^{A\!z}}
     &&&& \widehat{X}\ar@<.4ex>[dll]^-{\widehat{\pi}} \\
  & && X \ar@<.4ex>@{^{(}->}[urr]^-{\widehat{\iota}}  &&&.
  }
$$
{\sc Figure} 3-1.
%
%
\begin{figure}[htbp]
 \bigskip
  \centering
  \includegraphics[width=0.80\textwidth]{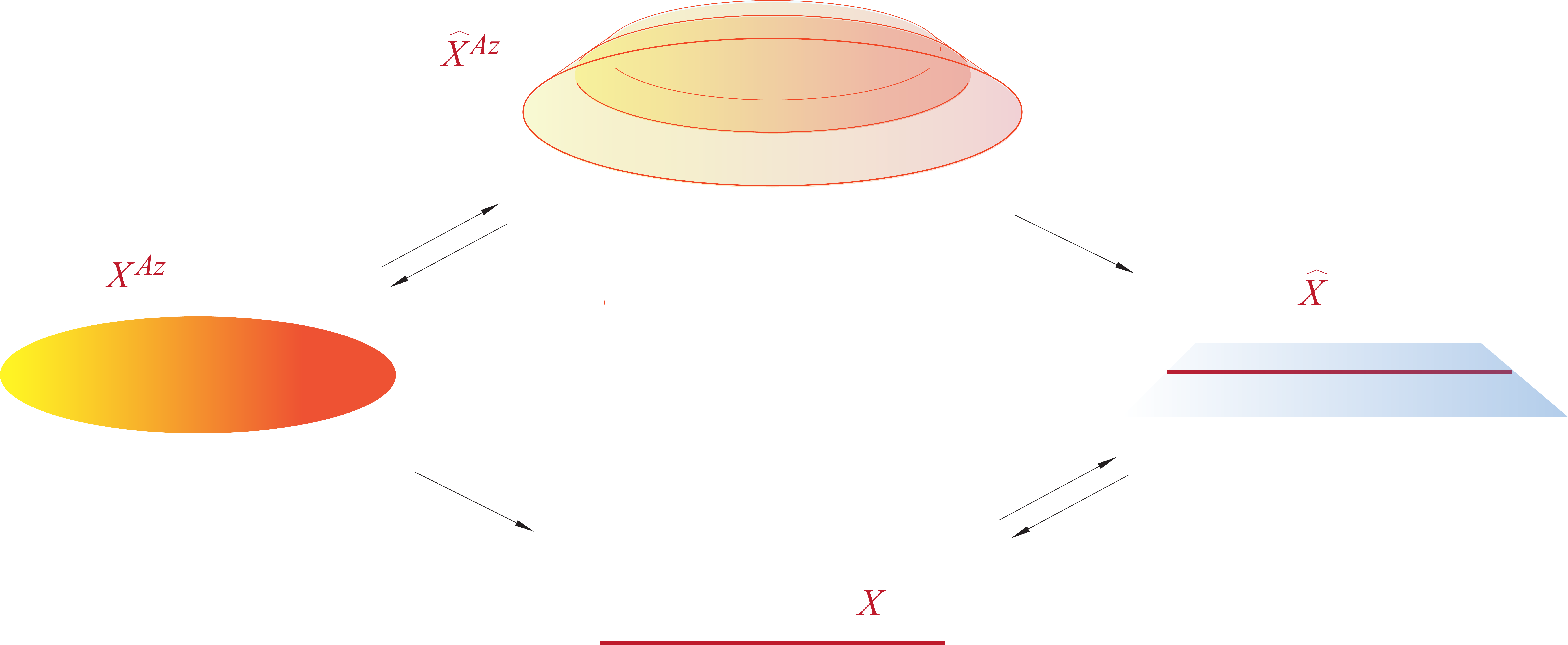}
  
  \bigskip
  \bigskip
 \centerline{\parbox{13cm}{\small\baselineskip 12pt
  {\sc Figure}~3-1.
  The built-in morphisms that underlie an Azumaya/matrix supermanifold $\widehat{X}^{\!A\!z}$.   
  It is worth emphasizing that, as in the case of Azumaya/matrix manifolds, 
    the superspace $\widehat{X}$ (and hecne $X$) should be regarded only 
	  as an auxiliary space, providing a topology underlying $\widehat{X}^{\!A\!z}$.
  The major object is the structure sheaf $\widehat{\cal O}_X^{A\!z}$, 
    which	should be thought of as a matrix-type noncommutayive cloud over $\widehat{X}$ 
	 and contains the geometrical contents that are relevant to D-branes.  	
  The built-in pair of morphisms 
    $\widehat{X}^{\!A\!z}
	   \begin{array}{c}
	    \raisebox{.4ex}{$\longrightarrow$} \\[-2.2ex] \longleftarrow \end{array}$ $X^{\!A\!z}$	 
   (resp.\  $\widehat{X}\begin{array}{c}
	                   \raisebox{.4ex}{$\longrightarrow$} \\[-2.2ex] \longleftarrow \end{array}X$)
    means that $\widehat{X}^{\!A\!z}$ is a sheaf-type thickening of $X^{\!A\!z}$ 
	 and contains $X^{\!A\!z}$ as the zero-section, as indicated,
     (resp.\  $\widehat{X}$ is a sheaf-type thickening of $X$ 
	    and contains $X$ as the zero-section, as indicated). 
  In the illustration, 
    the bluish shade indicates a ``{\it super}ization"  
	  while the orangish shade indicates ``{\it matrix}ization"; 
	and their combination is indicated by a bluish orangish shade.
  }}
\end{figure}

\bigskip

\begin{sremark}$[\,$smearing of matrix-points$\,]$. {\rm
 Conceptually, it is instructive to regards an Azumaya/ matrix super-$C^k$-manifold
  either as a smearing of unfixed matrix points over a super-$C^k$-manifold
 or as a smearing of unfixed matrix superpoint over a $C^k$-manifold.
} \end{sremark}

\medskip	

\begin{sdefinition} {\bf [tangent sheaf].} {\rm
 Continuing the notation in Definition~3.2.
 The {\it tangent sheaf} ${\cal T}_{\ast}\widehat{X}^{\!A\!z}$ of $\widehat{X}^{\!A\!z}$
   is the $\widehat{\cal O}_X$-module which assigns to each open $U\subset X$
   the $\widehat{\cal O}_X(U)$-module
   $\sDer_{\Bbb C}(\widehat{\cal O}_X^{A\!z}(U))$.
}\end{sdefinition}

\medskip

\begin{sdefinition} {\bf [cotangent sheaf].} {\rm
 Continuing the notation in Definition~3.2.
 The {\it cotangent sheaf} ${\cal T}^{\ast}\widehat{X}^{\!A\!z}$ of $\widehat{X}^{\!A\!z}$
   is the $\widehat{\cal O}_X^{A\!z}$-module which assigns to each open $U\subset X$
   the $\widehat{\cal O}_X^{A\!z}(U)$-module
   $\Omega_{\widehat{\cal O}_X^{A\!z}(U)/{\Bbb C}}$.
}\end{sdefinition}

\bigskip

Through the Morita equivalence,
 general (left) $\widehat{\cal O}_X^{A\!z}$-modules on $\widehat{X}^{\!A\!z}$
  can be obtained by the tensor
  $\widehat{\cal E}\otimes_{\widehat{\cal O}_X}\widehat{\cal F}$
  of the fundamental module $\widehat{\cal E}$ with $\widehat{\cal O}_X$-modules $\widehat{\cal F}$.

\bigskip

\begin{flushleft}
{\bf Remarks on general endomorphism-ringed super-$C^k$-schemes and
         differential calculus thereupon}	
\end{flushleft}
Let
 $\widehat{X}:= (X,{\cal O}_X,\widehat{\cal O}_X)$ be a super-$C^k$-scheme,
   with ${\cal O}_X$ a sheaf of finitely-generated germ-determined $C^k$-rings,    and
 $\widehat{\cal F}$ be a finitely-presented quasi-coherent $\widehat{\cal O}_X^{\,\Bbb C}$-module
   on $\widehat{X}$.
Then, the sheaf  $\Endsheaf_{\widehat{\cal O}_X^{\,\Bbb C}}(\widehat{\cal F})$
 of endomorphisms  of $\widehat{\cal F}$ is a quasi-coherent $\widehat{\cal O}_X^{\,\Bbb C}$-module
 that is finitely presentable as well.
Thus,
 if one assumes further that the built-in sheaf-of-rings homomorphism
  $\widehat{\cal O}_X^{\,\Bbb C}
      \rightarrow \Endsheaf_{\widehat{\cal O}_X^{\,\Bbb C}}(\widehat{\cal F})$
  is injective; i.e.\ $\widehat{\cal F}$ is supported on the whole $\widehat{X}$,
  and defines
   $$
      \widehat{\cal O}_X^{\nc}\;
	  :=\;  \Endsheaf_{\widehat{\cal O}_X^{\,\Bbb C}}(\widehat{\cal F})\,,
   $$
 then $\widehat{\cal O}_X^{\nc}$ is a sheaf of rings that is a finitely-presentable (noncommutative)
 algebraic extension of the sheaf of rings $\widehat{\cal O}_X^{\,\Bbb C}$.
This defines a new ringed topological space
 $$
   \widehat{X}^{\nc}\; :=\; (X,{\cal O}_X,\widehat{\cal O}_X,\widehat{\cal O}_X^{\nc})\,,
 $$
 which may be named {\it endomorphism-ringed super-$C^k$-scheme with a fundamental module}.
The Azumaya/matrix case $\widehat{X}^{\!A\!z}$ corresponds
  to the case ${\cal F}$ is in addition locally free.
Notions or constructions for $\widehat{X}^{\!A\!z}$ that depend only on
  the finite-presentability or algebraicness of  the extension
    $\widehat{\cal O}_X\hookrightarrow \widehat{\cal O}_X^{A\!z} $
 generalize immediately to the general endomorphism-ringed space $\widehat{X}^{\nc}$.	
In particular, the basics:
  \begin{itemize}
   \item[{\Large $\cdot$}]
    the notion of $\widehat{X}^{\nc}$
	  as an equivalence class of gluing systems of rings,

   \item[{\Large $\cdot$}]	
   the functor
     $\widehat{\cal F}\otimes_{\widehat{\cal O}_X}(\,\cdot\,):
	   \widehat{\cal O}_X\mbox{-}\ModCategory
	    \rightarrow  \widehat{\cal O}_X^{\nc}\mbox{-}\ModCategory$
   that turn an $\widehat{\cal O}_X$-module to a (left) $\widehat{\cal O}_X^{\nc}$-module,
    
   \item[{\Large $\cdot$}]
   the tangent sheaf ${\cal T}_{\ast}\widehat{X}^{\nc}$   and
   the cotangent sheaf ${\cal T}^{\ast}\widehat{X}^{\nc}$
   of  $\widehat{X}^{\nc}$
  \end{itemize}	
 can all be defined/establised.
And it remains instructive to regard the whole $\widehat{\cal O}_X$,
   despite only ${\Bbb Z}/2$-commutative,
 as the center of $\widehat{\cal O}_X^{\nc}$.

\bigskip

\section{Differentiable maps from an Azumaya/matrix supermanifold with a fundamental module
                    to a real manifold}

With the preparations in Sec.~2	and Sec.~3,
we now come to the main theme of the current note:				
the notion of `{\it differentiable map} from an Azumaya/matrix supermanifold with a fundamental module'
 that generalizes the setting in [L-Y3] (D(11.1)) for the case of Azumaya/matrix manifolds.

\bigskip

\subsection{A local study: The affine case}

The local study in this subsection is the foundation to the general notion of a differentiable map
 from an Azumaya/matrix supermanifold with a fundamental module to a real manifold.

\bigskip

\begin{flushleft}
{\bf From the aspect of function-rings and modules}
\end{flushleft}
\begin{definition} {\bf [admissible homomorphism from $C^k$-ring to Azumaya/matrix super-$C^k$-ring].}
{\rm
 Let
  \begin{itemize}
    \item[{\Large $\cdot$}]
      $U\subset {\Bbb R}^m$ and $V\subset {\Bbb R}^n$ be open sets,
	
	\item[{\Large $\cdot$}]
     $E$ be a complex $C^k$ vector bundle of rank $r$ on $U$,
	  for our purpose we may assume that $E$ is trivial,

	\item[{\Large $\cdot$}]
      $\xymatrix{\widehat{U} \ar@<-.2ex>[r]_-{\widehat{\pi}}
	                         & U \ar@<-.5ex> @{_{(}->}[l]_-{\widehat{\iota}}}$
       be a super-$C^k$-manifold supported on $U$ associated to the super-$C^k$-polynomial ring
        $C^k(U)[\theta^1,\,\cdots\,,\theta^s]$,
	    where $\theta^\alpha\theta^{\beta}+\theta^{\beta}\theta^{\alpha}=0$
	     for $1\le \alpha,\,\beta\le s$,
		
    \item[{\Large $\cdot$}]		
	 $\pr_{\widehat{U}}:\widehat{U}\times V\rightarrow \widehat{U}$,
	 $\pr_V:\widehat{U}\times V\rightarrow V$
	  be the projection maps, and
		
    \item[{\Large $\cdot$}]		
     $\widehat{E}:= \widehat{\pi}^{\ast}E$ be the pull-back complex vector bundle on $\widehat{U}$.
  \end{itemize}
 Then the endomorphism algebra of $\widehat{E}$ over  $\widehat{U}$
    is isomorphic to the $r\times r$ matrix ring\\
	$M_{r\times r}(C^k(U)[\theta^1,\,\cdots\,,\theta^s])$
	   over $C^k(U)[\theta^1,\,\cdots\,,\theta^s]^{\Bbb C}$.
 With this identification,
   a ring-homomorphism
    $$
	  \widehat{\varphi}^{\sharp}\;:\; C^k(V)\;
	      \longrightarrow\;  M_{r\times r}(C^k(U)[\theta^1,\,\cdots\,,\theta^s]^{\Bbb C})
	$$
 	 over ${\Bbb R}\hookrightarrow{\Bbb C}$  is said to be {\it $C^k$-admissible}
   if the following diagram of ring-homomorphisms
    $$
	  \xymatrix{
	     M_{r\times r }(C^k(U)[\theta^1,\,\cdots\,,\theta^s]^{\Bbb C})
		    && C^k(V)\ar[ll]_-{\widehat{\varphi}^{\sharp}}\\
         C^k(U)[\theta^1,\,\cdots\,,\theta^s]	\rule{0em}{1.2em}
       		    \ar@{^{(}->}[u]^-{\dot{\pi}^{A\!z,\sharp}}		
	   }
	$$
    extends to a commutative diagram of ring-homomorphisms
	  (over ${\Bbb R}$ or ${\Bbb R}\hookrightarrow {\Bbb C}$ when applicable)
    $$
	  \xymatrix{
	     M_{r\times r }(C^k(U)[\theta^1,\,\cdots\,,\theta^s]^{\Bbb C})
		    && C^k(V)\ar[ll]_-{\widehat{\varphi}^{\sharp}}
			                        \ar@{_{(}->}[d]^-{pr_V^{\sharp}}  \\
         \hspace{1ex}C^k(U)[\theta^1,\,\cdots\,,\theta^s]	\rule{0em}{1.2em}\hspace{1ex}
       		    \ar@{^{(}->}[u]^-{\dot{\pi}^{A\!z,\sharp}}		
				\ar@{^{(}->}[rr]_-{pr_{\widehat{U}}^{\sharp}}
			&& C^k(U\times V)[\theta^1,\,\cdots\,,\theta^s]
			         \ar[llu]_-{\tilde{\widehat{\varphi}}^{\sharp}}
	   }
	$$
    such that the following two conditions are satisfied:
	\begin{itemize}
	  \item[(1)]
	    The image of $\tilde{\widehat{\varphi}}^{\sharp}$
	    $$
		  \Image \tilde{\widehat{\varphi}}^{\sharp}\;
		    \simeq\;  \frac{C^k(U\times V)[\theta^1,\,\cdots\,,\theta^s]}
			                          {\Ker\tilde{\widehat{\varphi}}^{\sharp}}
		$$
        admits a quotient super-$C^k$-ring structure from the super-$C^k$-ring
		 $C^k(U\times V)[\theta^1,\,\cdots\,,\theta^s]$.
		
      \item[(2)]		
	   Regard $\dot{\pi}^{A\!z,\sharp}$ and $\widehat{\varphi}^{\sharp}$
	     now as ring-homomorphism to $\Image\tilde{\widehat{\varphi}}^{\sharp}$;
	   then, with respect to the super-$C^k$-ring structure in Condition (1),
	   both ring-homomorphisms
	    $$
		  \begin{array}{c}
		    \dot{\pi}^{A\!z,\sharp}\; :\;
			   C^k(U)[\theta^1,\,\cdots\,,\theta^s]\;
			     \longrightarrow\; \Image\tilde{\widehat{\varphi}}^{\sharp}\;,  \\[1.2ex]
		    \widehat{\varphi}^{\sharp}\;:\;
			    C^k(V)\; \longrightarrow\; \Image\tilde{\widehat{\varphi}}^{\sharp}						 
		  \end{array}		
		$$
	    are super-$C^k$-ring-homomorphisms.	
	\end{itemize}
 Note that
   since $C^k(U\times V)[\theta^1,\,\cdots\,,\theta^s]$ is the push-out
    of $C^k(U)[\theta^1,\,\cdots\,,\theta^s]$ and $C^k(V)$ in the category of super-$C^k$-rings,
  $\tilde{\widehat{\varphi}}^{\sharp}$ is unique if exists.
 In this case, one may think of  $\Image\tilde{\widehat{\varphi}}^{\sharp}$   as
 {\it the super-$C^k$-ring generated by $C^k(U)[\theta^1,\,\cdots\,,\theta^s]$
          and $\Image\widehat{\varphi}^{\sharp}$}
   in $M_{r\times r}(C^k(U)[\theta^1,\,\cdots\,,\theta^s]^{\Bbb C})$.
 In notation,
   $$
    A_{\widehat{\varphi}}\;
       :=\; 	  C^k(U)[\theta^1,\,\cdots\,,\theta^s]     \langle \Image \hat{\varphi}^{\sharp}\rangle\;
	   :=\;     \Image\tilde{\widehat{\varphi}}^{\sharp}\,.	
   $$ 	
}\end{definition}

\bigskip

Bringing both
     $A_{\widehat{\varphi}}$ and
     the module of sections of the fundamental vector bundle $\widehat{E}$ over $\widehat{U}$
   into the picture,
 one now has the following full diagram for a $C^k$-admissible ring-homomorphism
  $\widehat{\varphi}: C^k(V)\rightarrow
    M_{r\times r}(C^k(U)[\theta^1,\,\cdots\,,\theta^s]^{\Bbb C})$
  over ${\Bbb R}\hookrightarrow {\Bbb C}\,$:
   $${\footnotesize
   \xymatrix{
     (C^k(U)[\theta^1,\,\cdots\,,\theta^s]^{\Bbb C})^{\oplus r}\\
	& M_{r\times r}(C^k(U)[\theta^1,\,\cdots\,,\theta^s]^{\Bbb C})  \ar@{~>}[lu] \\
    &    A_{\widehat{\varphi}}\,:=\,\rule{0ex}{3ex}
         		C^k(U)[\theta^1,\,\cdots\,,\theta^s]\langle\Image\widehat{\varphi}^{\sharp}\rangle
				      \ar@{^{(}->}[u]_-{\sigma_{\widehat{\varphi}}^{\sharp}}
					  \ar@{~>}@/^2ex/[luu]
                   				&&& C^k(V)\ar[lllu]_-{\widehat{\varphi}^{\sharp}}
								                           \ar[lll]^-{f_{\widehat{\varphi}}^{\sharp}}
														   \ar@{_{(}->}[d]^-{pr_V^{\sharp}}\\
    &  \;\;C^k(U)[\theta^1,\,\cdots\,,\theta^s]\;\; \rule{0ex}{3ex}
                        	       \ar@{^{(}->}[u]_-{\pi_{\widehat{\varphi}}^{\sharp}}
	                               \ar@{~>}@/^5ex/[luuu]
						           \ar@{^{(}->}[rrr]_-{pr_{\widehat{U}}^{\sharp}}
         &&&  \; C^k(U\times V)[\theta^1,\,\cdots\,,\theta^s]
		                           \ar@{->>}[ulll]^-{\tilde{\widehat{\varphi}}^{\sharp}}\;.
   }  }
  $$
 
\bigskip

\begin{remark} $[\,$$k=\infty$$\,]$. {\rm
 For the case $k=\infty$,
  Condition (1)  in Definition~4.1.1 is always satisfied and, hence, redundant.
}\end{remark}

\bigskip

\begin{flushleft}
{\bf From the aspect of super-$C^k$-schemes and sheaves}
\end{flushleft}
Continuing the notations of the previous theme.
Let
 \begin{itemize}
  \item[{\Large $\cdot$}]
   ${\cal O}_U$ be the sheaf of $C^k$-functions on $U$,
   ${\cal O}_V$ be the sheaf of $C^k$-functions on $V$,
   $\widehat{\cal O}_U$ be the structure sheaf of $\widehat{U}$ as the affine super-$C^k$-scheme
     associated to the super-$C^k$-ring $C^k(U)[\theta^1,\,\cdots\,,\theta^s]$,

  \item[{\Large $\cdot$}]
   ${\cal A}_{\widehat{\varphi}}$ be the $\widehat{\cal O}_U$-algebra
   associated to the $C^k(U)[\theta^1,\,\cdots\,,\theta^s]$-algebra $A_{\widehat{\varphi}}$,
	
  \item[{\Large $\cdot$}]
   ${\cal E}$ be the sheaf of $C^k$-sections of $E$ over $U$,
   $\widehat{\cal E}:= \widehat{\pi}^{\ast}{\cal E}$
     the induced locally free $\widehat{\cal O}_U$-module,

  \item[{\Large $\cdot$}]
   $\widehat{\cal O}_U^{A\!z}
       := \Endsheaf_{\widehat{\cal O}_U^{\,\Bbb C}}(\widehat{\cal E})$
   the structure sheaf of the Azumaya/matrix super-$C^k$-manifold $\widehat{U}^{\!A\!z}$,
   which realizes $\widehat{\cal E}$ as the fundamental module on $\widehat{U}^{\!A\!z}$.
 \end{itemize}
Let
 $$
    \widehat{\varphi}^{\sharp}\; :\;
	   C^k(V)\; \longrightarrow\;
         M_{r\times r}(C^k(U)[\theta^1,\,\cdots\,,\theta^s]^{\Bbb C})
 $$
 be a $C^k$-admissible ring-homomorphism over ${\Bbb R}\hookrightarrow {\Bbb C}$
 as in the previous theme.
Then the full diagram of ring-homomorphisms associated to $\widehat{\varphi}$ in the previous theme
 can be translated into the following diagram of maps between spaces:
$$
   \xymatrix{
    \; \widehat{\cal E} \ar@{.>}[rd]     \ar@{.>}@/_1ex/[rdd]    \ar@{.>}@/_2ex/[rddd]   \\
      & \widehat{U}^{\!A\!z}
	                                    \ar[rrrd]^-{\widehat{\varphi}}
	                                    \ar@{->>}[d]^-{\sigma_{\widehat{\varphi}}}   \\
      &\;\; \widehat{U}_{\widehat{\varphi}}\;\;
	                                   \ar[rrr]_-{f_{\widehat{\varphi}}}
	                                   \ar@{_{(}->} [rrrd]_-{\tilde{\widehat{\varphi}}}
                                	   \ar@{->>}[d]^-{\pi_{\widehat{\varphi}}} &&& \;\;V\;\; \\
	  &\;\; \widehat{U}\;\;
	     &&& \;\; \widehat{U}\times V
		                   \ar@{->>}[u]_-{pr_V} \ar@{->>}[lll]^-{pr_{\widehat{U}}}  \;.
    }
  $$
 Notice that in the above diagram,
  the maps $\pi_{\widehat{\varphi}}$, $f_{\widehat{\varphi}}$, and
    $\tilde{\widehat{\varphi}}$
   are now maps between super-$C^k$-ringed spaces in the sense of Definition~2.3.10
 while the `maps' $\widehat{\varphi}$ and $\sigma_{\widehat{\varphi}}$
   are only conceptual without real contents and are defined solely contravariantly
   through the ring-homomorphisms
    $\widehat{\varphi}^{\sharp}$ and $\sigma_{\widehat{\varphi}}^{\sharp}$
   respectively.
   
\bigskip

\begin{definition}
{\bf [$C^k$-map from Azumaya/matrix super-$C^k$-manifold $\widehat{U}^{\!A\!z}$].} {\rm
 We shall call
  $\widehat{\varphi}:(\widehat{U}^{\!A\!z},\widehat{\cal E} )\rightarrow V$
  in the above diagram
  a  {\it $k$-times differentiable map} (in brief,  {\it $C^k$-map})
  from the super-$C^k$-manifold with a fundamental module
     $(\widehat{U}^{\!A\!z},\widehat{\cal E})$ to the $C^k$-manifold $V$.
 Through the underlying admissible ring-homomorphism $\widehat{\varphi}^{\sharp}$
    in the previous theme,
  $\widehat{\cal E}$	becomes an ${\cal O}_Y$-module,
  called the {\it push-forward} of $\widehat{\cal E}$ under $\widehat{\varphi}$   and
  denoted by $\widehat{\varphi}_{\ast}(\widehat{\cal E})$  or
                       $\widehat{\varphi}_{\ast}\widehat{\cal E}$.
 The {\it image} of $\widehat{\varphi}: \widehat{U}^{\!A\!z}\rightarrow V$,
   in notation $\Image\widehat{\varphi}$ or $\widehat{\varphi}(\widehat{U}^{\!A\!z})$,
 is the $C^k$-subscheme of $V$ defined by the ideal $\Ker(\widehat{\varphi})$ of $C^k(V)$.
}\end{definition}

\medskip

\begin{lemma}
{\bf [image $\Image\widehat{\varphi}$
           vs.\ support $\Supp(\widehat{\varphi}_{\ast}(\widehat{\cal E}))$].}
 Continuing the notation in Definition~4.1.3.
 The image $\Image(\widehat{\varphi})$ of the $C^k$-map $\widehat{\varphi}$
    is identical to the $C^k$-scheme-theoretical support
      $\Supp(\widehat{\varphi}_{\ast}(\widehat{\cal E}))$
	  of the push-forward $\widehat{\varphi}_{\ast}(\widehat{\cal E})$.
\end{lemma}

\medskip

\begin{definition}
{\bf [surrogate of $\widehat{U}^{\!A\!z}$ specified by $\widehat{\varphi}$].}
{\rm
 Continuing the discussion.
 The super-$C^k$-scheme $\widehat{U}_{\widehat{\varphi}}$
    together with the built-in $C^k$-maps
    $$
     \xymatrix{
       \widehat{U}_{\widehat{\varphi}}\ar[rrr]^-{f_{\widehat{\varphi}}}
   	                                                       \ar@{->>}[d]^-{\pi_{\widehat{\varphi}}} &&& V                \\
	   \widehat{U}
      }
    $$
  is called the {\it surrogate of} (the noncommutative) $\widehat{U}^{\!A\!z}$
   {\it specified by $\widehat{\varphi}:\widehat{U}^{\!A\!z}\rightarrow V$}.   	
 Caution that in general there is {\it no} $C^k$-map $\widehat{U}\rightarrow V$	
   that makes the diagram commute.
} \end{definition}

\bigskip

\begin{flushleft}
{\bf The role of the fundamental module $\widehat{\cal E}$}
\end{flushleft}
Let
  $$
    \widehat{\varphi}\;:\; (\widehat{U}^{\!A\!z},\widehat{\cal E})\;
	   \longrightarrow\; V
  $$
 be a $C^k$-map defined by a $C^k$-admissible ring-homomorphism 	
  $$
    \widehat{\varphi}^{\sharp}\; :\;
	   C^k(V)\; \longrightarrow\;
         M_{r\times r}(C^k(U)[\theta^1,\,\cdots\,,\theta^s]^{\Bbb C})
   $$
  over ${\Bbb R}\hookrightarrow {\Bbb C}$ after fixing a trivialization of ${\cal E}$ on $U$.
Then
 the same argument as in [L-Y3: Sec.\ 5.2] (D(11.1)) implies by construction  the following properties
 related to the fundamental module $\widehat{\cal E}$  on the super-$C^k$-manifold $\widehat{U}$:
 \begin{itemize}
  \item[(1)]
   The fundamental module $\widehat{\cal E}:= \widehat{\pi}^{\ast}{\cal E}$,
       first on $\widehat{U}$,
	 is also naturally an ${\cal A}_{\widehat{\varphi}}$-module
	          on $\widehat{U}_{\widehat{\varphi}}$  and
	  an $\widehat{\cal O}_U^{A\!z}$-module on $\widehat{U}^{\!A\!z}$.
   We will denote them all by $\widehat{\cal E}$.	
	
  \item[(2)]
  The ${\cal O}_V$-modules
       $\varphi_{\ast}(\widehat{\cal E})$ and
       $f_{\widehat{\varphi},\ast}(\widehat{\cal E})$
   are canonically isomorphic.
 \end{itemize}
 
\bigskip
 
\begin{definition} {\bf [graph of $\widehat{\varphi}$].} {\rm
 The {\it graph} of the $C^k$-map
  $\widehat{\varphi}:(\widehat{U}^{\!A\!z},\widehat{\cal E})\rightarrow V$
  is defined to be the ${\cal O}_{\widehat{U}\times V}$-module
   $$
     \tilde{\widehat{\cal E}}_{\widehat{\varphi}}\;
	  :=\;  \tilde{\widehat{\varphi}}_{\ast}(\widehat{\cal E})
   $$
   on the product-space  $\widehat{U}\times V$.
}\end{definition}

\medskip
 
\begin{definition}
{\bf [$C^k$-admissible ${\cal O}_{\widehat{M}}$-module].}   {\rm
 Let $\widehat{M}$ be a super-$C^k$-manifold.
 An ${\cal O}_{\widehat{M}}$-module ${\widehat{\cal F}}$
       is said to be {\it $C^k$-admissible}
   if the annihilator ideal sheaf
		 $\Ker({\cal O}_{\widehat{M}}
		       \rightarrow
		           \Endsheaf_{{\cal O}_{\widehat{M}}}(\widehat{\cal F}))$
      is super-$C^k$-normal.
 In this case,
  $\Ker({\cal O}_{\widehat{M}}
		       \rightarrow
		           \Endsheaf_{{\cal O}_{\widehat{M}}}(\widehat{\cal F}))$
   defines a super-$C^k$-subscheme structure for $\Supp(\widehat{\cal F})$
    	   in $\widehat{M}$; 		
  i.e., the quotient map of ${\cal O}_{\widehat{M}}$, as a sheaf of super-$C^k$-rings,
   $$
     {\cal O}_{\widehat{M}}\;
	    \longrightarrow\;
		  {\cal O}_{\scriptsizeSupp(\widehat{\cal F})}
		      :={\cal O_{\widehat{M}}}/
			            \Ker({\cal O}_{\widehat{M}}
		                  \rightarrow
		                    \Endsheaf_{{\cal O}_{\widehat{M}}}(\widehat{\cal F}))
   $$
   induces a sheaf-of-super-$C^k$-rings structure
     on ${\cal O}_{\scriptsizeSupp(\widehat{\cal F})}$.
}\end{definition}

\medskip

\begin{remark} $[\,$case $k=\infty$$\,]$. {\rm
 For a super-$C^{\infty}$-manifold $\widehat{M}$,
 every ${\cal O}_{\widehat{M}}$-module $\widehat{\cal F}$
  is $C^{\infty}$-admissible.
}\end{remark}

\medskip

\begin{lemma}
{\bf [basic properties of $\tilde{\widehat{\cal E}}_{\widehat{\varphi}}$].}
 The graph $\tilde{\widehat{\cal E}}_{\widehat{\varphi}}$ of $\widehat{\varphi}$
  has the following properties:
  \begin{itemize}
   \item[$(1)$]
     $\tilde{\widehat{\cal E}}_{\widehat{\varphi}}$ is
       a $C^k$-admissible ${\cal O}_{\widehat{U}\times V}^{\,\Bbb C}$-module;
	  its super-$C^k$-scheme-theoretical support
	     $\Supp(\tilde{\widehat{\cal E}}_{\widehat{\varphi}})$
	   is isomorphic to the surrogate $ \widehat{U}_{\widehat{\varphi}}$
	     of $\widehat{U}^{\!A\!z}$ specified by $\widehat{\varphi}$.
     In particular,
	  $\tilde{\widehat{\cal E}}_{\widehat{\varphi}}$
	  is of relative dimension $0$ over $\widehat{U}$	
		   	
   \item[$(2)$]	
    There is a canonical isomorphism
	 $\widehat{\cal E}
	    \stackrel{\sim}{\longrightarrow}
      	  \pr_{\widehat{U},\ast}(\tilde{\widehat{\cal E}})$
	 of ${\cal O}_{\widehat{U}}^{\,\Bbb C}$-modules.
   In particular, $\tilde{\widehat{\cal E}}_{\widehat{\varphi}}$
      is flat over $\widehat{U}$, of relative complex length $r$.
	
   \item[$(3)$]	 	
    There is a canonical exact sequence of
	 ${\cal O}_{\widehat{U}\times V}^{\,\Bbb C}$-modules
	 $$
	  \pr_{\widehat{U}}^{\ast}(\widehat{\cal E})\;
	    \longrightarrow\;  \tilde{\widehat{\cal E}}_{\widehat{\varphi}}\;
		\longrightarrow\;  0\,.
	 $$
	
   \item[$(4)$]
    The $\widehat{\cal O}_V$-modules
      $\pr_{V,\ast}(\tilde{\widehat{\cal E}}_{\widehat{\varphi}})$
      and $\widehat{\varphi}_{\ast}(\widehat{\cal E})$
     are canonically isomorphic.	 	
  \end{itemize}
\end{lemma}

\bigskip
 
Conversely, one has the following lemma of reconstruction,
 which follows the same argument as
    in [L-L-S-Y] (D(2)) for the algebraic case  and
    	[L-Y3] (D(11.1)) for the $C^k$-case:

\bigskip

\begin{lemma}
{\bf [reconstructing $C^k$-map via ${\cal O}_{\widehat{U}\times Y}^{\,\Bbb C}$-module].}
 Let $\tilde{\widehat{\cal E}}$ be an ${\cal O}_{\widehat{U}\times V}^{\,\Bbb C}$-module
  that is $C^k$-admissible,
     and of relative dimension $0$, of finite relative complex length, and flat over $\widehat{U}$.
 For the moment, we assume further that
   $\pr_{\widehat{U},\ast}(\tilde{\widehat{\cal E}})$	is trivial.
 Let ${\widehat{\cal E}}:=  \pr_{\widehat{U},\ast}(\tilde{\widehat{\cal E}})$.
 Then
  $\tilde{\widehat{\cal E}}$ specifies a $C^k$-map
   $\widehat{\varphi}:(\widehat{U}^{\!A\!z},\widehat{\cal E})\rightarrow V$
    whose graph $\tilde{\widehat{\cal E}}_{\widehat{\varphi}}$
      is canonically isomorphic to  $\tilde{\widehat{\cal E}}$.
\end{lemma}

\begin{proof}
 Note that
  the ${\cal O}_V$-action on $\tilde{\widehat{\cal E}}$ via $\pr_V^{\,\sharp}$
    and the ${\cal O}_{\widehat{U}}^{\,\Bbb C}$-action on $\tilde{\widehat{\cal E}}$
	     via $\pr_{\widehat{U}}^{\,\sharp}$ commute
   since
     the image of $\pr_V^{\,\sharp}:{\cal O}_V \rightarrow {\cal O}_{\widehat{U}\times V}$
	 lies in the center of ${\cal O}_{\widehat{U}\times V}$.
 Thus, 		
    the ${\cal O}_V$-action on $\tilde{\widehat{\cal E}}$ via $\pr_V^{\,\sharp}$
      induces a $C^k$-admissible
	 $\widehat{\varphi}^{\sharp}:
	    {\cal O}_V    \rightarrow
		   \Endsheaf_{{\cal O}_{\widehat{U}}^{\,\Bbb C}}(\widehat{\cal E})$,
   which defines a $C^k$-map\\
    $\widehat{\varphi}:
	  (U,
	      \widehat{\cal O}_U^{\!A\!z}
		     :=\Endsheaf_{{\cal O}_{\widehat{U}}^{\,\Bbb C}}(\widehat{\cal E}),
          \widehat{\cal E})
		  \rightarrow V$.		
\end{proof}
 	
\bigskip

\begin{remark} {$[\,$when $\widehat{U}=U$$\,]$.} {\rm
 For ${\cal S}=0$ the zero-${\cal O}_U$-module, $\widehat{U}=U$;   and
 all the settings/objects/statements in this subsection for the super-$C^k$-case
  reduce to the corresponding settings/objects/statements in
  [L-Y3: Sec.~5.1 \& Sec.~5.2] (D(11.1)) for the $C^k$ case.
}\end{remark}

\bigskip

\begin{flushleft}
{\bf The induced $C^k$-map $\varphi:(U^{\!A\!z},{\cal E})\rightarrow V$}
\end{flushleft}
By the post-composition with $\dot{\widehat{\iota}}^{\sharp}$
    of the built-in ring-homomorphisms of the matrix super-$C^k$-ring
  $$
      \xymatrix{
	    \;\; M_{r\times r}(C^k(U)[\theta^1,\,\cdots\,,\theta^s]^{\Bbb C}  )\;
    		      \ar@<.45ex>@{->>}[rr]^-{\dot{\widehat{\iota}}^{\sharp}}
	       && \;  M_{r\times r}(C^k(U)^{\Bbb C})
		                \ar@<.3ex> @{_{(}->}[ll]^-{\dot{\widehat{\pi}}^{\sharp}}\;,
      }
  $$ 	
 a ring-homomorphism
  $$
    \widehat{\varphi}^{\sharp}\;:\; C^k(V)\;\longrightarrow\;
	      M_{r\times r}(C^k(U)[\theta^1,\,\cdots\,,\theta^s]^{\Bbb C})
  $$
  induces a ring-homomorphism
  $$
   \varphi^{\sharp}\;:\;  C^k(V)\; \longrightarrow\; M_{r\times r}(C^k(U)^{\Bbb C})\,.
  $$
It follows by construction that:
  
\bigskip

\begin{lemma} {\bf [induced $C^k$-admissible ring-homomorphism].}
  If $\widehat{\varphi}^{\sharp}$ is $C^k$-admissible,
   then $\varphi^{\sharp}$ is $C^k$-admissible as well.
\end{lemma}

\bigskip
  
In this case,
 the $C^k$-map
   $$
     \widehat{\varphi}\;:\;
	   (\widehat{U}^{\!A\!z},\widehat{\cal E})\, :=\,
	  (\widehat{U},
	     \widehat{\cal O}_U^{A\!z}:
		     =\Endsheaf_{\widehat{\cal O}_U^{\,{\Bbb C}}}(\widehat{\cal E}),
		 \widehat{\cal E})\;
	  \longrightarrow\; (V,{\cal O}_V)
   $$
   defined by $\widehat{\varphi}^{\sharp}$
 induces a $C^k$-map
   $$
     \varphi\;:\;
	    (U^{\!A\!z},{\cal E})\,:=\,
	   (U,
	    {\cal O}_U^{A\!z}:= \Endsheaf_{{\cal O}_U^{\,\Bbb C}}({\cal E}),
		{\cal E})\;
	    \longrightarrow\;   (V,{\cal O}_V)		
   $$
    defined by $\varphi^{\sharp}$
   that fits into the following commutative diagram
  $$
    \xymatrix{
     &  \widehat{U}^{A\!z}\ar[rrr]^-{\widehat{\varphi}}    &&&  V    \\
     & U^{A\!z}\rule{0em}{1.2em}
           	 \ar@{^{(}->}[u]^-{\dot{\widehat{\iota}}}
	         \ar[urrr]_-{\varphi}  &&&	&,
    }
  $$
whose full detail is given in the commutative diagram below:  (Cf.\ [L-Y3] (D(11.1)).)
$$
   \xymatrix{
    \; \widehat{\cal E}=\widehat{\pi}^{\ast}{\cal E}
	          \ar@{.>}[rrrd]     \ar@{.>}@/_1ex/[rrdd]    \ar@{.>}@/_2ex/[rddd]    \\
      &&& \widehat{U}^{\!A\!z}                            	
	                                    \ar[rrrrd]^-{\widehat{\varphi}}
	                                    \ar@{->>}[ld]_-{\sigma_{\widehat{\varphi}}}
										\ar@<.3ex>@{->>}'[d]'[dd][dddd]^(.3){\dot{\widehat{\pi}}}
										\\
      &&\;\; \widehat{U}_{\widehat{\varphi}}\;\;
	                                   \ar[rrrrr]^(.43){f_{\widehat{\varphi}}}
	                                   \ar@{_{(}->} [rrrrd]^(.6){\tilde{\widehat{\varphi}}}
                                	   \ar@{->>}[ld]_-{\pi_{\widehat{\varphi}}}
									   \ar@<.3ex>@{->>}'[d][dddd]^(.3){\widehat{\pi}}
		 & \raisebox{-1.5em}{ \rule{0em}{2em}} &&&& \;\;V\;\; \ar@{=}[dddd]   \\
	  &\;\; \widehat{U}\;\;     \ar@<.3ex>@{->>}[dddd]^(.42){\widehat{\pi}}
	     &\rule{0ex}{1ex}  &  \rule{0ex}{1ex}
		 &&& \;\; \widehat{U}\times V
		                   \ar@{->>}[ru]_-{pr_V}
						   \ar@{->>}[lllll]^(.4){pr_{\widehat{U}}}
						   \ar@<.3ex>@{->>}[dddd]^(.4){\widehat{\pi}}
						   \;  \\
    \; {\cal E}
	          \ar@{.>}[rrrd]     \ar@{.>}@/_1ex/[rrdd]    \ar@{.>}@/_2ex/[rddd]    \\
      &&& U^{\!A\!z}\rule{0em}{1.2em}
	                                    \ar[rrrrd]^-{\varphi}
	                                    \ar@{->>}[ld]^-{\sigma_{\varphi}}
										\ar@<.3ex>@{^{(}->}'[uu]^(.7){\dot{\widehat{\iota}}}'[uuu][uuuu]
										\\
      &&\;\; U_{\varphi}\rule{0em}{1.2em}\;\;
	                                   \ar'[rrrr]^(.6){f_{\varphi}}[rrrrr]
	                                   \ar@{_{(}->} [rrrrd]^(.6){\tilde{\varphi}}
                                	   \ar@{->>}[ld]^-{\pi_{\varphi}}
									   \ar@<.3ex>@{^{(}->}'[uuu]^(.7){\widehat{\iota}} [uuuu]
		 &&&& \rule{2ex}{0ex}  & \;\;V\;\;    \\
	  &\;\; U\rule{0em}{1.2em}\;\;
	            \ar@<.3ex>@{^{(}->}[uuuu]^(.58){\widehat{\iota}}
	     &&&&& \;\; U\times V \rule{0em}{1.2em}
		                   \ar@{->>}[ru]_-{pr_V}
						   \ar@{->>}[lllll]^(.4){pr_{U}}
                           \ar@<.3ex>@{^{(}->}[uuuu]^(.6){\widehat{\iota}} 		
			        &\;\;\;\;\;\;. 			
    }
  $$

Further observation that
   $\Ker\dot{\widehat{\iota}}^{\sharp}$ is a nilpotent bi-ideal of
   $M_{r\times r}(C^k(U)[\theta^1,\,\cdots\,,\theta^s]^{\Bbb C})$
 implies then:

\bigskip

\begin{lemma} {\bf [images differ by nilpotency].}
 As $C^k$-subschemes of $V$,
  $$
   (\Image\widehat{\varphi})_{\redscriptsize}\;=\; (\Image{\varphi})_{\redscriptsize}\,.
  $$
\end{lemma}
  
\medskip

\begin{remark} {$[\,$for general ${\cal E}$ and $\widehat{U}$$\,]$.} {\rm
 {To} make the discussion in this subsection more explicit and notationwise simpler,
   we choose
      the complex vector bundle $E$ over $U$ to be trivial (and trivialized)   and
	  the super-$C^k$-manifold $\widehat{U}$ to be of product type $U\times \widehat{p}$,
	  where $\widehat{p}$ is a super-point.
 For general $E$ and $\widehat{U}$,
  $$
   {\cal O}_{\widehat{U}}\; :=\;  \widehat{\cal O}_U \;
      =\;  \mbox{$\bigwedge^{\bullet}{\cal S}$}
  $$
   for some locally free ${\cal O}_U$-module ${\cal S}$
     associated to a vector bundle $S$ on $U$ of rank, say, $s$.
  Let
	\begin{itemize}
	  \item[{\Large $\cdot$}]
        $E^{\vee}$ be the dual complex bundle of $E$   and\\
		$\End_U(E)=E\otimes E^{\vee}$
  		   be the bundle of (complex) endomorphisms of $E$.
    \end{itemize}
 Then
    \begin{itemize}
	 \item[{\Large $\cdot$}]
	  $\widehat{\cal O}_U(U)\;=\;  C^k(\bigwedge^{\bullet}S )$,
	
	 \item[{\Large $\cdot$}]
     $\widehat{\cal E}(U)$ is canonically isomorphic to
	    $C^k(E)\otimes_{C^k(U)}C^k(\bigwedge^{\bullet}S)$, 	
	
     \item[{\Large $\cdot$}]
	 $\widehat{\cal O}_U^{A\!z}(U)$ is canonically isomorphic to
	   $C^k(\End_U(E))\otimes_{C^k(U)}C^k(\bigwedge^{\bullet}S)$.		
    \end{itemize}
  And the argument in this subsection goes through with the following replacements:
  \begin{itemize}
   \item[{\Large $\cdot$}]
     the $C^k(U)[\theta^1,\,\cdots\,\theta^s]$-module	
	 $\,(C^k(U)[\theta^1,\,\cdots\,,\theta^s]^{\Bbb C})^{\oplus r}$   \\[1.2ex]
	$\Longrightarrow\;$
	the $C^k(\bigwedge^{\bullet}S)$-module
	 $\,C^k(E)\otimes_{C^k(U)}C^k(\bigwedge^{\bullet}S)$;
	
   \item[{\Large $\cdot$}]
     the $C^k(U)[\theta^1,\,\cdots\,,\theta^s]$-algebra	
	 $\,M_{r\times r}(C^k(U)[\theta^1,\,\cdots\,\theta^s]^{\Bbb C})$ \\[1.2ex]
	$\Longrightarrow\;$
     the  $C^k(\bigwedge^{\bullet}S)$-algebra
	   $\,C^k(\End_U(E))\otimes_{C^k(U)}C^k(\bigwedge^{\bullet}S)$.
  \end{itemize}
}\end{remark}

\bigskip

\subsection{Differentiable maps from an Azumaya/matrix supermanifold with\\ a fundamental module
                           to a real manifold}
						
In Sec.~4.1 we see that locally the notion of a $C^k$-map from an Azumaya/matrix super-$C^k$-manifold
  is fundamentally the same as the notion of a $C^k$-map from an Azumaya/matrix $C^k$-manifold.
The only difference is the ring involved, which may modify the exact presentation but not the underlying concept.
It follows that
 \begin{itemize}
  \item[{\Large $\cdot$}] {\it
   Gluing from local to global with respect to the topology of the (auxiliary) manifold $X$,
   all the settings/discussions in [L-Y3: Sec.~5.3] (D(11.1))
   can be adapted without work to the current super-case.}
 \end{itemize}
The essential details are given in this subsection for the completeness of discussion.

\bigskip

\subsubsection{Aspect I [fundamental]: Maps as gluing systems of ring-homomorphisms}

The notion of a differentiable map
   $$
     \widehat{\varphi}\;:\;
     (\widehat{X},
         \widehat{\cal O}_X^{\!A\!z}
    	   := \Endsheaf_{\widehat{\cal O}_X^{\,\Bbb C}}(\widehat{\cal E}),
 	     \widehat{\cal E})\;
      \longrightarrow\;  Y
   $$
   from an Azumaya/matrix supermanifold with a fundamental module to a real manifold 	
 follows from the notion of `{\sl morphisms between spaces}' studied in
 [L-Y1: Sec.\ 1.2$\;$ {\sl A noncommutative space as a gluing system of rings}] (D(1)).

\bigskip

\begin{flushleft}
{\bf  The fundamental aspect of $C^k$-maps from Azumaya manifolds}
\end{flushleft}
\begin{ssdefinition}{\bf [gluing system of $C^k$-admissible ring-homomorphisms].}
{\rm
 Let
  \begin{itemize}
    \item[{\Large $\cdot$}]
	 $(X,{\cal O}_X)$ be a $C^k$-manifold, with the structure sheaf
	   ${\cal O}_X$ of $C^k$-functions on $X$,
	
	\item[{\Large $\cdot$}]
	 ${\cal E}$ be a locally free ${\cal O}_X^{\,\Bbb C}$-module of finite rank on $X$,
	
    \item[{\Large $\cdot$}]
     $\widehat{X}$ be a super-$C^k$-manifold with the structure sheaf
	   $\widehat{\cal O}_X:= \bigwedge^{\bullet}{\cal S}$
	    for some locally free ${\cal O}_X$-module ${\cal S}$ of finite rank on $X$,
		
	\item[{\Large $\cdot$}]
	 $\widehat{\cal E}:= {\cal E}\otimes_{{\cal O}_X} \widehat{\cal O}_X$,
	
    \item[{\Large $\cdot$}]	
     $(\widehat{X}^{\!A\!z},\widehat{\cal E})\;
	    :=\;
       (\widehat{X},
            \widehat{\cal O}_X^{\!A\!z}
    	      := \Endsheaf_{\widehat{\cal O}_X^{\,\Bbb C}}(\widehat{\cal E}),
 	     \widehat{\cal E})$
      be an Azumaya/matrix super-$C^k$-manifold with a fundamental module,
	  	
    \item[{\Large $\cdot$}]	
     $(Y,{\cal O}_Y)$ be a $C^k$-manifold,
	  with the structure sheaf ${\cal O}_Y$ of $C^k$-functions on $Y$.
  \end{itemize}
 A {\it (contravariant) gluing system of $C^k$-admissible ring-homomorphisms}
      over ${\Bbb R}\hookrightarrow{\Bbb C}$ related to $(\widehat{X}^{\!A\!z},Y)$
  consists of the following data:
  \begin{itemize}
    \item[{\Large $\cdot$}]
	 ({\it local charts on $\widehat{X}^{\!A\!z}$})\hspace{1.1em}
     an open cover ${\cal U}=\{U_{\alpha}\}_{\alpha\in A}$ on $X$,
	
	\item[{\Large $\cdot$}]
	 ({\it local charts on $Y$})\hspace{2em}
     an open cover ${\cal V}=\{V_{\beta}\}_{\beta\in B}$ on $Y$,
	
	\item[{\Large $\cdot$}]
     a gluing system $\widehat{\Phi}^{\sharp}$ of $C^k$-admissible ring-homomorphisms
       from $\{C^k(V_{\beta})\}_{\beta}$\\
       to $\{C^k(\Endsheaf_{{\cal O}_{U_{\alpha}}^{\,\Bbb C}}
	                                               ({\cal E}_{U_{\alpha}}))
				    \otimes_{C^k(U_{\alpha})}
					C^k(\bigwedge^{\bullet}{\cal S}_{U_{\alpha}})\}_{\alpha}$
       over ${\Bbb R}\hookrightarrow {\Bbb C}$, which consists of
	    \begin{itemize}
	      \item[{\Large $\cdot$}]
	      ({\it specification of a target-chart
		               for each local chart on $\widehat{X}^{\!A\!z}$})\hspace{2em}
	       a map $\sigma:A\rightarrow B$,
	
	      \item[{\Large $\cdot$}]
         ({\it  differentiable map from charts on $\widehat{X}^{\!A\!z}$ to charts on $Y$})\\
	       a $C^k$-admissible ring-homomorphism over ${\Bbb R}\hookrightarrow {\Bbb C}$
	       $$
	         \widehat{\phi}^{\,\sharp}_{\alpha, \sigma(\alpha)}\;:\;
   	         C^k(V_{\sigma(\alpha)}) \;
			 \longrightarrow\;
		     C^k(\Endsheaf_{{\cal O}_{U_{\alpha}}^{\,\Bbb C}}
			                                                                                   ({\cal E}_{U_{\alpha}}))
		        \otimes_{C^k(U_{\alpha})}
			    C^k(\mbox{$\bigwedge$}^{\bullet}{\cal S}_{U_{\alpha}})
	       $$	
	       for each $\alpha\in A$
		\end{itemize}
	   that satisfy   
		 \begin{itemize}
      	  \item[{\Large $\cdot$}]
		  ({\it gluing/identification of maps at overlapped charts on $\widehat{X}^{\!A\!z}$})\\
		   for each pair $(\alpha_1$, $\alpha_2)\in A\times A$,
		   \begin{itemize}
		     \item[(G1)]
		      $(\widehat{\phi}_{\alpha, \sigma(\alpha_1)})_{\ast}
		         (\widehat{\cal E}_{U_{\alpha_1}\cap\, U_{\alpha_2}})$
			   is completely supported in
			   $V_{\sigma(\alpha_1)}\cap V_{\sigma(\alpha_2)}
			      \subset V_{\sigma(\alpha_1)}$,
						
             \item[(G2)]		
			 recall the $C^k$-admissible ring-homomorphism over ${\Bbb R}\hookrightarrow {\Bbb C}$	
			  {\footnotesize
	          $$\hspace{-8em}
		        \widehat{\phi}_{\alpha_1\alpha_2,\, \sigma(\alpha_1)\sigma(\alpha_2)}
				                            ^{\,\sharp}\;:\;
				   C^k(V_{\sigma(\alpha_1)}\cap V_{\sigma(\alpha_2)})\;
				     \longrightarrow\;
				   C^k(\Endsheaf_{{\cal O}_{U_{\alpha_1}\cap\, U_{\alpha_2}}^{\,\Bbb C}}
				              ({\cal E}_{U_{\alpha_1}\cap\,U_{\alpha_2}}))	 		
                     \otimes_{C^k(U_{\alpha_1}\cap\, U_{\alpha_2})}
					 C^k(\mbox{$\bigwedge$}^{\bullet}
					              {\cal S}_{U_{\alpha_1}\cap\, U_{\alpha_2}})			
		      $$}
			  induced by $\phi_{\alpha_1,\sigma(\alpha_1)}$,
			then 			
	            $$
		          \widehat{\phi}_{\alpha_1\alpha_2,\,\sigma(\alpha_1)\sigma(\alpha_2)}
				                                  ^{\,\sharp}\;
				   =\; \widehat{\phi}_{\alpha_2\alpha_1,\,\sigma(\alpha_2)\sigma(\alpha_1)}
				                                          ^{\,\sharp}\,.
		        $$
		   \end{itemize}
         \end{itemize}
     \end{itemize}		
} \end{ssdefinition}	
	
\bigskip

\begin{ssdefinition}{\bf [equivalent systems].} {\rm
 A  gluing system $({\cal U}^{\prime},{\cal V}^{\prime},  \widehat{\Phi}^{\prime\sharp})$
   is said to be a
 {\it refinement} of another gluing system $({\cal U},{\cal V},  \widehat{\Phi}^{\sharp})$,
    in notation
	  $({\cal U}^{\prime},{\cal V}^{\prime}, \widehat{\Phi}^{\prime\sharp})
            \preccurlyeq({\cal U},{\cal V}, \widehat{\Phi}^{\sharp})$,	
  if
     \begin{itemize}	
	  \item[{\Large $\cdot$}]
	   ${\cal U}^{\prime}
	       = \{U^{\prime}_{\alpha^{\prime}}\}_{\alpha^{\prime}\in A^{\prime}}$
		is a refinement of  ${\cal U}= \{U_{\alpha}\}_{\alpha\in A}$,
		with a map $\tau:A^{\prime}\rightarrow A$ that labels inclusions
		  $U^{\prime}_{\alpha^{\prime}}\hookrightarrow U_{\tau(\alpha^{\prime})}$; 		
       similarly,		
	   ${\cal V}^{\prime}
	       = \{V^{\prime}_{\beta^{\prime}}\}_{\beta^{\prime}\in B^{\prime}}$
		is a refinement of ${\cal V}= \{V_{\beta}\}_{\beta\in B}$,
		 with a map $\upsilon:B^{\prime}\rightarrow B$  that labels inclusions
		   $V^{\prime}_{\beta^{\prime}}
		        \hookrightarrow V_{\upsilon(\beta^{\prime})}$;
        the maps between the index sets $A$, $B$, $A^{\prime}$, and $B^{\prime}$ satisfy
		  the commutative diagram
		  $$
            \xymatrix{
			 &  A^{\prime} \ar[rr]^-{\sigma^{\prime}} \ar[d]_-{\tau}
			     && B^{\prime}\ar[d]^-{\upsilon} \\
		     &  A \ar[rr]^-{\sigma}  && B  &. 			
			}
          $$		
							
      \item[{\Large $\cdot$}]		 
	    the $C^k$-admissible ring-homomorphism
		  $$
		    \widehat{\phi}^{\prime\sharp}
			                               _{\alpha^{\prime},\, \sigma^{\prime}(\alpha^{\prime})}\; :\;
		     C^k(V_{\sigma^{\prime}(\alpha^{\prime})})\;  \longrightarrow\;
			  C^k(\Endsheaf_{{\cal O}_{U_{\alpha^{\prime}}}^{\,\Bbb C}}
			           ({\cal E}_{U_{\alpha^{\prime}}})  )
				\otimes_{C^k(U_{\alpha^{\prime}})}
				C^k(\mbox{$\bigwedge$}^{\bullet}{\cal S}_{U_{\alpha^{\prime}}})
		  $$
			in $\widehat{\Phi}^{\prime\sharp}$
         coincides with the $C^k$-admissible ring-homomorphism
		   $$
    		  C^k(V_{\sigma^{\prime}(\alpha^{\prime})}) \;  \longrightarrow\;
			  C^k(\Endsheaf_{{\cal O}_{U_{\alpha^{\prime}}}^{\,\Bbb C}}
			           ({\cal E}_{U_{\alpha^{\prime}}})  )
			      \otimes_{C^k(U_{\alpha^{\prime}})}
				 C^k(\mbox{$\bigwedge$}^{\bullet}{\cal S}_{U_{\alpha^{\prime}}})
		   $$
		    induced by
	       $$\hspace{-1.2em}
		    \mbox{\footnotesize
		     $\widehat{\phi}^{\,\sharp}
		       _{\tau(\alpha^{\prime}),\,\sigma(\tau(\alpha^{\prime}))}  \,
			   =\,
		       \widehat{\phi}^{\,\sharp}
		        _{\tau(\alpha^{\prime}),\, \upsilon(\sigma^{\prime}(\alpha^{\prime}))}\; :\;
		     C^k(V_{\upsilon(\sigma^{\prime}(\alpha^{\prime}))}) \; \longrightarrow\;
			  C^k(\Endsheaf_{{\cal O}_{U_{\tau(\alpha^{\prime}) }}^{\,\Bbb C}}
			           ({\cal E}_{U_{\tau(\alpha^{\prime}) }})  )
			    \otimes_{C^k(U_{\tau(\alpha^{\prime})})}
			   C^k(\bigwedge^{\bullet}{\cal S}_{U_{\tau(\alpha^{\prime})}})$
			  }	
		   $$
			in $\widehat{\Phi}^{\sharp}$ from the inclusions
			 $U_{\alpha^{\prime}}\hookrightarrow U_{\tau(\alpha^{\prime})}$		 and
			 $V_{\sigma^{\prime}(\alpha^{\prime})} \hookrightarrow
			     V_{\upsilon(\sigma^{\prime}(\alpha^{\prime}))  }$.
     \end{itemize}
  Two gluing systems
     $({\cal U}_1,{\cal V}_1, \widehat{\Phi}_1^{\sharp})$
	  and $({\cal U}_2,{\cal V}_2,  \widehat{\Phi}_2^{\sharp})$	
     are said to be {\it equivalent} if they have a common refinement. 	
}\end{ssdefinition}
  
\bigskip

\begin{ssdefinition}{\bf [differentiable map as equivalence class of gluing systems].}
{\rm
 We denote an equivalence class of contravariant gluing systems of $C^k$-admissible ring-homomorphisms
 compactly as
   $$
     \varphi^{\sharp}\;:\; {\cal O}_Y\; \longrightarrow \;
	   \widehat{\cal O}_X^{A\!z}\,
	      :=\, \Endsheaf_{\widehat{\cal O}_X^{\,\Bbb C}}(\widehat{\cal E})\,.
   $$
  This defines a {\it $k$-times differentiable map} (i.e., {\it $C^k$-map})
   $$
     \widehat{\varphi}\;:\;
	  (\widehat{X},
	    \widehat{\cal O}_X^{A\!z}
		     :=\Endsheaf_{\widehat{\cal O}_X^{\,\Bbb C}}(\widehat{\cal E}),
		  \widehat{\cal E})\;
     \longrightarrow\;  Y\,.
   $$
}\end{ssdefinition}

\bigskip

The $C^k$-admissible ring-homomorphism
   $$
	  \widehat{\phi}^{\,\sharp}_{\alpha, \sigma(\alpha)}\;:\;
   	    C^k(V_{\sigma(\alpha)}) \;\longrightarrow\;
		C^k(\Endsheaf_{{\cal O}_{U_{\alpha}}^{\,\Bbb C}}({\cal E}_{U_{\alpha}}))
		 \otimes_{C^k(U_{\alpha})}
		C^k(\mbox{$\bigwedge$}^{\bullet}{\cal S}_{U_{\alpha}})
   $$
   renders
    $C^k({\cal E}_{U_{\alpha}})
	   \otimes_{C^k(U_{\alpha})}
	   C^k(\bigwedge^{\bullet}{\cal S}_{U_{\alpha}})$
   a $C^k(V_{\sigma(\alpha)})$-module.
Passing to germs of $C^k$-sections (with respect to the topology of $X$),
 this defines a  sheaf of ${\cal O}_{V_{\sigma(\alpha)}}$-modules,
 denoted by
  $(\widehat{\phi}_{\alpha,\,\sigma(\alpha)})_{\ast}
        (\widehat{\cal E}_{U_{\alpha}})$.
The following lemma/definition follows by construction:

\bigskip

\begin{sslemma-definition}
 {\bf [push-forward
            $\widehat{\varphi}_{\ast}(\widehat{\cal E})$ under $\widehat{\varphi}$].}
{\rm
  {\it The collection of sheaves on local charts
    $\{(\widehat{\phi}_{\alpha,\,\sigma(\alpha)})_{\ast}
	       (\widehat{\cal E}_{U_{\alpha}})\}_{\alpha\in A}$
	 glue to a sheaf of ${\cal O}_Y$-modules on $Y$.
   It is independent of the contravariant gluing system of
      $C^k$-admissible ring-homomorphisms over ${\Bbb R}\hookrightarrow {\Bbb C}$
     that represents $\widehat{\varphi}$.}
    It is called the {\it push-forward of $\widehat{\cal E}$ under $\widehat{\varphi}$} and
	  is denoted by $\widehat{\varphi}_{\ast}(\widehat{\cal E})$.
}\end{sslemma-definition}

\medskip

\begin{sslemma-definition}
 {\bf [surrogate of $\widehat{X}^{\!A\!z}$ specified by $\widehat{\varphi}$].} {\rm
   (Cf.~Definition~4.1.5.)
   {\it The collection of local surrogates $\widehat{U}_{\varphi_{\alpha,\sigma(\alpha)}}$
              of $\widehat{U}_{\alpha}^{\!A\!z}$
			  specified by $\widehat{\varphi}_{\alpha,\sigma(\alpha)}$
	  glue to a super-$C^k$-scheme over $\widehat{X}$.
   It is independent of the contravariant gluing system of
      $C^k$-admissible ring-homomorphisms over ${\Bbb R}\hookrightarrow {\Bbb C}$
     that represents $\widehat{\varphi}$.}
    It is called the {\it surrogate of $\widehat{X}^{\!A\!z}$ specified by $\widehat{\varphi}$};
	 in notation, $\widehat{X}_{\widehat{\varphi}}$.
}\end{sslemma-definition}

\bigskip

Readers are referred to [L-Y3: Remark 5.3.1.8] (D(11.1)),
   with  a straightforward adaptation to the current super-case,
 for three conceptually important remarks on Definition~4.2.1.3.

\bigskip

\begin{flushleft}
{\bf The equivalent affine setting}
\end{flushleft}
Recall that the $C^k$-manifolds  $(X,{\cal O}_X)$ and $(Y,{\cal O}_Y)$
 are affine $C^k$-schemes
 associated respectively to the $C^k$-rings $C^k(X)$ and $C^k(Y)$
 in the context of $C^k$-algebraic geometry.
Observe also that
 \begin{itemize}
  \item[{\Large $\cdot$}]
   {\it As an ${\cal O}_X$-module,
    the sheaf $\widehat{\cal O}_X^{A\!z}$ of ${\cal O}_X^{\,\Bbb C}$-algebras is quasi-coherent.
	Explicitly, it is the quasi-coherent sheaf on the affine $C^k$-scheme $(X,{\cal O}_X)$ associated to the
	$C^k(X)^{\,\Bbb C}$-module
	  $C^k(\Endsheaf_{{\cal O}_X^{\,\Bbb C}}({\cal E}))
	     \otimes_{C^k(X)}C^k(\bigwedge^{\bullet}{\cal S})$.}
 \end{itemize}
This implies that
  
 
\bigskip

\begin{sslemma-definition}
{\bf [$C^k$-map in affine setting].}
 The equivalence class
  $$
    \widehat{\varphi}^{\sharp}\; :\;  {\cal O}_Y\;
	    \longrightarrow\;
		\widehat{\cal O}_X^{A\!z}
		:=  \Endsheaf_{\widehat{\cal O}_X^{\,\Bbb C}}(\widehat{\cal E})
  $$
   of gluing systems of
      $C^k$-admissible ring-homomorphisms over ${\Bbb R}\hookrightarrow{\Bbb C}$
   in Definition~4.2.1.3 defines a $C^k$-admissible ring-homomorphism,
    still denoted by $\widehat{\varphi}^{\sharp}$,
   $$
     \widehat{\varphi}^{\sharp}\;:\; C^k(Y)\; \longrightarrow\;
	      C^k(\widehat{X}^{\!A\!z})
		     := C^k(\Endsheaf_{{\cal O}_X^{\,\Bbb C}}({\cal E}))
			       \otimes_{C^k(X)}
				   C^k(\mbox{$\bigwedge$}^{\bullet}{\cal S})\,		                                      
   $$
   over ${\Bbb R}\hookrightarrow{\Bbb C}$.
  Conversely, any $C^k$-admissible ring-homomorphism
      $\widehat{\varphi}^{\sharp}: C^k(Y)\rightarrow  C^k(\widehat{X}^{\!A\!z})$
	     over ${\Bbb R}\rightarrow {\Bbb C}$
	 defines an equivalence class
	 $\widehat{\varphi}^{\sharp}: {\cal O}_Y \rightarrow  \widehat{\cal O}_X^{A\!z}$
	  of contravariant gluing systems of $C^k$-admissible ring-homomorphisms
	  over ${\Bbb R}\hookrightarrow {\Bbb C}$ associated to $(\widehat{X}^{\!A\!z},Y)$.
  It follows that the notion of a $C^k$-map
	$$
      \widehat{\varphi}\;:\;
	    (X,
		   \widehat{\cal O}_X^{A\!z}
		         :=\Endsheaf_{\widehat{\cal O}_X^{\,\Bbb C}}(\widehat{\cal E}),
		    \widehat{\cal E})\;
     \longrightarrow\;  Y\,.
    $$
     in Definition~4.2.1.3
     can be equivalently defined by a $C^k$-admissible ring-homomorphism\\
      $\widehat{\varphi}^{\sharp}: C^k(Y)\rightarrow  C^k(\widehat{X}^{\!A\!z})$
     over ${\Bbb R}\hookrightarrow{\Bbb C}$.	
\end{sslemma-definition}

%
	
\bigskip

\begin{ssremark} {$[\,$when $\widehat{X}=X$$\,]$.} {\rm
 For ${\cal S}=0$ the zero-${\cal O}_X$-module, $\widehat{X}=X$;   and
 all the settings/objects/statements in this subsection for the super-$C^k$-case in the current subsection
  reduce to the corresponding settings/objects/statements in
  [L-Y3: Sec.~5.3.1] (D(11.1)) for the $C^k$ case.
 And hence similarly,
  Sec.~4.2.2, Sec.~4.2.3, Sec.~4.2.4 of the current note
  to [L-Y3: Sec.~5.3.2, Sec.~5.3.3, Sec.~5.3.4] (D(11.1)).
}\end{ssremark}

%
%
%
%

\bigskip

\begin{flushleft}
{\bf The induced $C^k$-map $\varphi:(X^{\!A\!z},{\cal E})\rightarrow Y$}
\end{flushleft}
Continuing the notation in Definition~4.2.1.1.
Let
 $$
  \begin{array}{l}
    \left(\rule{0em}{1.2em}\right.
	{\cal U}= \{U_{\alpha}\}_{\alpha\in A}\,,\;
          {\cal V}= \{V_{\beta}\}_{\beta\in B}\,,	             \\[1.2ex]
	   \hspace{2.4em}	
	   \widehat{\Phi}^{\,\sharp}
		  = \left(\rule{0em}{.9em}\right.\sigma:A\rightarrow B\,,\;
		         \{
				   \widehat{\phi}^{\,\sharp}_{\alpha,\sigma(\alpha)}:
				       C^k(V_{\sigma(\alpha)})
					   \rightarrow
					  C^k(\Endsheaf_{{\cal O}_{U_{\alpha}}^{\,\Bbb C}}({\cal E}))
                        \otimes_{C^k(U_{\alpha})}					
                         C^k(\bigwedge^{\bullet}{\cal S}_{U_{\alpha}})
						 \}_{\alpha\in A}\left.\rule{0em}{0.9em}\right)
     \left.\rule{0em}{1.2em}\right )
  \end{array}						
 $$
 be a contravariant gluing system of $C^k$-admissible ring-homomorphisms
 over ${\Bbb R}\hookrightarrow{\Bbb C}$
 associated to $(\widehat{X}^{\!A\!z},Y)$.
Recall the surjective ring-homomorphism
 $$
    \dot{\widehat{\iota}}^{\,\sharp}\; : \;
	  C^k(\Endsheaf_{{\cal O}_{U_{\alpha}}^{\,\Bbb C}}({\cal E}))
	   \otimes_{C^k(U_{\alpha})}
	   C^k(\mbox{$\bigwedge$}^{\bullet}{\cal S}_{U_{\alpha}})\;
	  \longrightarrow\;
      C^k(\Endsheaf_{{\cal O}_{U_{\alpha}}^{\,\Bbb C}}({\cal E}))	
 $$
 for every $\alpha\in A$  and
let
 $$
   \phi^{\,\sharp}_{\alpha,\sigma(\alpha)}\
    :=\;       \dot{\widehat{\iota}}\,
	                \circ\, \widehat{\phi}^{\,\sharp}_{\alpha,\sigma(\alpha)}\,.
 $$
Then,
 $$
  \begin{array}{l}
    \left(\rule{0em}{1.2em}\right.
	{\cal U}= \{U_{\alpha}\}_{\alpha\in A}\,,\;
          {\cal V}= \{V_{\beta}\}_{\beta\in B}\,,	             \\[1.2ex]
	   \hspace{2.4em}	
	    \Phi^{\,\sharp}
		  = \left(\rule{0em}{.9em}\right.\sigma:A\rightarrow B\,,\;
		         \{
				   \phi^{\,\sharp}_{\alpha,\sigma(\alpha)}:
				       C^k(V_{\sigma(\alpha)})
					   \rightarrow
					  C^k(\Endsheaf_{{\cal O}_{U_{\alpha}}^{\,\Bbb C}}({\cal E}))
						 \}_{\alpha\in A}\left.\rule{0em}{0.9em}\right)
     \left.\rule{0em}{1.2em}\right )
  \end{array}						
 $$
 becomes a contravariant gluing system of $C^k$-admissible ring-homomorphisms
 over ${\Bbb R}\hookrightarrow{\Bbb C}$
 associated to $(X^{\!A\!z},Y)$.
Furthermore,
  if $({\cal U}_1,{\cal V}_1, \widehat{\Phi}_1^{\,\sharp})$  and
     $({\cal U}_2,{\cal V}_2, \widehat{\Phi}_2^{\,\sharp})$
	 are equivalent,
then, so are their associated gluing systems
  $({\cal U}_1,{\cal V}_1, \Phi_1^{\,\sharp})$  and
  $({\cal U}_2,{\cal V}_2, \Phi_2^{\,\sharp})$. 	
It follows that
 
\bigskip

\begin{ssproposition} {\bf [induced $C^k$-map $\varphi:(X^{\!A\!z},{\cal E})\rightarrow Y$].}
 $\widehat{\varphi}^{\,\sharp}:{\cal O}_Y  \rightarrow  \widehat{\cal O}_X^{A\!z}$
   defines an accompanying $\varphi^{\,\sharp}:{\cal O}_Y \rightarrow {\cal O}_X^{A\!z}$
   through the post-composition with
   $\dot{\widehat{\iota}}^{\,\sharp}:
       \widehat{\cal O}_X^{A\!z}\rightarrow {\cal O}_X^{A\!z}$;
   that is, a commutative diagram
   $$
     \xymatrix{
      &  \widehat{\cal O}_X^{A\!z}  \ar@{->>}[d]_-{\dot{\widehat{\iota}}^{\,\sharp}}
	       &&&  {\cal O}_Y
		                 \ar[lll]_-{\widehat{\varphi}^{\,\sharp}} \ar[llld]^-{\varphi^{\,\sharp}}\\
      & {\cal O}_X^{A\!z}&&&	&,
     }
   $$
 or, equivalently, a commutative diagram of $C^k$-maps
   $$
    \xymatrix{
     &  \widehat{X}^{\!A\!z}\ar[rrr]^-{\widehat{\varphi}}    &&&  Y    \\
     & X^{\!A\!z}\rule{0em}{1.2em}
           	 \ar@{^{(}->}[u]^-{\dot{\widehat{\iota}}}
	         \ar[urrr]_-{\varphi}  &&&	&,
    }
  $$
whose full detail is given in the commutative diagram below:
(Cf.\ Diagrams after Lemma~4.1.12.)

$$
   \xymatrix{
    \; \widehat{\cal E}=\widehat{\pi}^{\ast}{\cal E}
	          \ar@{.>}[rrrd]     \ar@{.>}@/_1ex/[rrdd]    \ar@{.>}@/_2ex/[rddd]    \\
      &&& \widehat{X}^{\!A\!z}                            	
	                                    \ar[rrrrd]^-{\widehat{\varphi}}
	                                    \ar@{->>}[ld]_-{\sigma_{\widehat{\varphi}}}
										\ar@<.3ex>@{->>}'[d]'[dd][dddd]^(.3){\dot{\widehat{\pi}}}
										\\
      &&\;\; \widehat{X}_{\widehat{\varphi}}\;\;
	                                   \ar[rrrrr]^(.43){f_{\widehat{\varphi}}}
	                                   \ar@{_{(}->} [rrrrd]^(.6){\tilde{\widehat{\varphi}}}
                                	   \ar@{->>}[ld]_-{\pi_{\widehat{\varphi}}}
									   \ar@<.3ex>@{->>}'[d][dddd]^(.3){\widehat{\pi}}
		 & \raisebox{-1.5em}{ \rule{0em}{2em}} &&&& \;\;  Y \;\; \ar@{=}[dddd]   \\
	  &\;\; \widehat{X}\;\;     \ar@<.3ex>@{->>}[dddd]^(.42){\widehat{\pi}}
	     &\rule{0ex}{1ex}  &  \rule{0ex}{1ex}
		 &&& \;\; \widehat{X}\times Y
		                   \ar@{->>}[ru]_-{pr_Y}
						   \ar@{->>}[lllll]^(.4){pr_{\widehat{X}}}
						   \ar@<.3ex>@{->>}[dddd]^(.4){\widehat{\pi}}
						   \;  \\
    \; {\cal E}
	          \ar@{.>}[rrrd]     \ar@{.>}@/_1ex/[rrdd]    \ar@{.>}@/_2ex/[rddd]    \\
      &&& X^{\!A\!z}\rule{0em}{1.2em}
	                                    \ar[rrrrd]^-{\varphi}
	                                    \ar@{->>}[ld]^-{\sigma_{\varphi}}
										\ar@<.3ex>@{^{(}->}'[uu]^(.7){\dot{\widehat{\iota}}}'[uuu][uuuu]
										\\
      &&\;\; X_{\varphi}\rule{0em}{1.2em}\;\;
	                                   \ar'[rrrr]^(.6){f_{\varphi}}[rrrrr]
	                                   \ar@{_{(}->} [rrrrd]^(.6){\tilde{\varphi}}
                                	   \ar@{->>}[ld]^-{\pi_{\varphi}}
									   \ar@<.3ex>@{^{(}->}'[uuu]^(.7){\widehat{\iota}} [uuuu]
		 &&&& \rule{2ex}{0ex}  & \;\; Y\;\;    \\
	  &\;\; X\rule{0em}{1.2em}\;\;
	            \ar@<.3ex>@{^{(}->}[uuuu]^(.58){\widehat{\iota}}
	     &&&&& \;\; X\times Y \rule{0em}{1.2em}
		                   \ar@{->>}[ru]_-{pr_Y}
						   \ar@{->>}[lllll]^(.4){pr_X}
                           \ar@<.3ex>@{^{(}->}[uuuu]^(.6){\widehat{\iota}} 		
			        &\;\;\;\;\;\;. 			
    }
  $$
\end{ssproposition}

\bigskip

Since $\Ker\dot{\widehat{\iota}}^{\,\sharp}$
  is a nilpotent ideal sheaf of $\widehat{\cal O}_X^{A\!z}$ on $\widehat{X}^{A\!z}$,
one has

\bigskip

\begin{sscorollary}
 As $C^k$-subschemes of $Y$,
 $$
  (\Image\widehat{\varphi})_{\redscriptsize}\;
    =\; (\Image{\varphi})_{\redscriptsize}\,.
 $$
\end{sscorollary}

\bigskip

\subsubsection{Aspect II: The graph of a differentiable map}

Similar to the studies
  [L-L-S-Y: Sec.\ 2.2] (D(2)) and [L-Y2:  Sec.\ 2.2] (D(6))
 in the algebro-geometric setting and
  [L-Y3: Sec.\ 5.3] (D(11.1)) in the $C^k$-algebro-geometric/synthetic-differential-topological setting,
the graph of a differentiable map
 $\widehat{\varphi}:
   (X,\widehat{\cal O}_X^{A\!z}
            :=\Endsheaf_{\widehat{\cal O}_X^{\,\Bbb C}}(\widehat{\cal E}),
	   \widehat{\cal E})
     \rightarrow Y$
  is a sheaf $\tilde{\widehat{\cal E}}_{\widehat{\varphi}}$
   of ${\cal O}_{\widehat{X}\times Y}^{\,\Bbb C}$-modules
  on the super-$C^k$-manifold $\widehat{X}\times Y$ with special properties.	
And $\widehat{\varphi}$ can be recovered from its graph.
Details are given below for the current super-synthetic-differential-topological setting.

\bigskip

\begin{flushleft}
{\bf Graph of a differentiable map
          $\widehat{\varphi}:(\widehat{X}^{\!A\!z}, \widehat{\cal E})\rightarrow Y$}
\end{flushleft}
It follows from the local study in Sec.~4.1 that
 an equivalence class
  $$
    \widehat{\varphi}^{\sharp}\; :\;  {\cal O}_Y\;
	    \longrightarrow\;
		\widehat{\cal O}_X^{A\!z}
		   :=  \Endsheaf_{\widehat{\cal O}_X^{\,\Bbb C}}(\widehat{\cal E})
  $$
   of gluing systems of
      $C^k$-admissible ring-homomorphisms over ${\Bbb R}\hookrightarrow{\Bbb C}$
 extends canonically to an equivalence class
  $$
    \tilde{\widehat{\varphi}}^{\sharp}\; :\;  {\cal O}_{\widehat{X}\times Y}\;
	    \longrightarrow\;   \widehat{\cal O}_X^{A\!z}
  $$
   of gluing systems of ring-homomorphisms over ${\Bbb R}\hookrightarrow{\Bbb C}$
 that defines canonically to a map
  $$
    \tilde{\widehat{\varphi}}\; :\;
    (X,
	    \widehat{\cal O}_X^{A\!z}
		   :=\Endsheaf_{\widehat{\cal O}_X^{\,\Bbb C}}(\widehat{\cal E}),
		\widehat{\cal E})\;
      \longrightarrow\;  \widehat{X}\times Y\,,
  $$
 making the following diagram commute:
  $$
   \xymatrix{
    & (X,
	        \widehat{\cal O}_X^{A\!z}
			   :=\Endsheaf_{\widehat{\cal O}_X^{\,\Bbb C}}(\widehat{\cal E}),
			\widehat{\cal E})
        \ar[rr]^-{\tilde{\widehat{\varphi}}} \ar[rrd]_-{\widehat{\varphi}}
	    && \widehat{X}\times Y \ar[d]^-{pr_Y} \\
	& && Y  &.
	}
  $$
Here $\pr_Y:\widehat{X}\times Y\rightarrow Y$ is the projection map to $Y$.

\bigskip

\begin{ssdefinition} {\bf [graph of $\widehat{\varphi}$].}  {\rm
 The {\it graph} of a $C^k$-map
   $\widehat{\varphi}:(\widehat{X}^{\!A\!z},\widehat{\cal E})\rightarrow Y$
  is a sheaf $\tilde{\widehat{\cal E}}_{\widehat{\varphi}}$
  of ${\cal O}_{\widehat{X}\times Y}^{\,\Bbb C}$-modules,
  defined by
  $$
    \tilde{\widehat{\cal E}}_{\widehat{\varphi}}\;
	   :=\; \tilde{\widehat{\varphi}}_{\ast}(\widehat{\cal E})\,.
  $$
}\end{ssdefinition}

\bigskip

The following basic properties of $\tilde{\widehat{\cal E}}_{\widehat{\varphi}}$
 follow directly from the local study in Sec.~4.1:

\bigskip

\begin{sslemma}
{\bf [basic properties of $\tilde{\widehat{\cal E}}_{\widehat{\varphi}}$].}
 The graph $\tilde{\widehat{\cal E}}_{\widehat{\varphi}}$ of $\widehat{\varphi}$
  has the following properties:
  \begin{itemize}
   \item[$(1)$]
     $\tilde{\widehat{\cal E}}_{\widehat{\varphi}}$ is
       a $C^k$-admissible ${\cal O}_{\widehat{X}\times Y}^{\,\Bbb C}$-module;
	  its super-$C^k$-scheme-theoretical support
	     $\Supp(\tilde{\widehat{\cal E}}_{\widehat{\varphi}})$
	   is isomorphic to the surrogate $ \widehat{X}_{\widehat{\varphi}}$
	     of $\widehat{X}^{\!A\!z}$ specified by $\widehat{\varphi}$.
     In particular,
	  $\tilde{\widehat{\cal E}}_{\widehat{\varphi}}$
	  is of relative dimension $0$ over $\widehat{X}$	
		   	
   \item[$(2)$]	
    There is a canonical isomorphism
	 $\widehat{\cal E}
	    \stackrel{\sim}{\longrightarrow}
      	  \pr_{\widehat{X},\ast}(\tilde{\widehat{\cal E}})$
	 of ${\cal O}_{\widehat{X}}^{\,\Bbb C}$-modules.
   In particular, $\tilde{\widehat{\cal E}}_{\widehat{\varphi}}$
      is flat over $\widehat{X}$, of relative complex length $r$.
	
   \item[$(3)$]	 	
    There is a canonical exact sequence of
	 ${\cal O}_{\widehat{X}\times Y}^{\,\Bbb C}$-modules
	 $$
	  \pr_{\widehat{X}}^{\ast}(\widehat{\cal E})\;
	    \longrightarrow\;  \tilde{\widehat{\cal E}}_{\widehat{\varphi}}\;
		\longrightarrow\;  0\,.
	 $$
	
   \item[$(4)$]
    The $\widehat{\cal O}_Y$-modules
      $\pr_{Y,\ast}(\tilde{\widehat{\cal E}}_{\widehat{\varphi}})$
      and $\widehat{\varphi}_{\ast}(\widehat{\cal E})$
     are canonically isomorphic.	 	
  \end{itemize}
\end{sslemma}
 
\medskip

\begin{sslemma} {\bf [presentation of graph of $\widehat{\varphi}$].}
 Continuing the notation.
 The graph $\tilde{\widehat{\cal E}}_{\widehat{\varphi}}$ of $\widehat{\varphi}$
   admits a presentation given by a natural isomorphism
  $$
    \tilde{\widehat{\cal E}}_{\widehat{\varphi}}\;
     \simeq\;
    	 \pr_{\widehat{X}}^{\ast}(\widehat{\cal E})
                 \left/  \left(
				      (\pr_Y^{\ast}(f)
					      -\pr_{\widehat{X}}^{\ast}(\widehat{\varphi}^{\sharp}(f))\,:\, f\in C^k(Y))
					       \cdot \pr_{\widehat{X}}^{\ast}(\widehat{\cal E})				
				           \right) \right.\,.
  $$
\end{sslemma}

\medskip

\begin{ssremark}
{$[$presentation of $\tilde{\widehat{\cal E}}_{\widehat{\varphi}}$
          in local trivialization of $\widehat{\cal E}$$\,]$.}\\
{\rm
 Note that
  with respect to a local trivialization
  ${\Bbb C}^{\oplus r}
        \otimes_{\Bbb R}
	 (C^k(U)\otimes_{\Bbb R}C^k(\bigwedge^{\bullet}{\cal S}_U))$
   of $\widehat{\cal E}$ and,
   hence, a local trivialization
   $$\mbox{$
    {\Bbb C}^{\oplus r}\otimes_{\Bbb R}C^k(U\times Y)
       \otimes_{\Bbb R}
	   C^k(\bigwedge^{\bullet}{\cal S}_U)\;
	   \simeq\;
	   {\Bbb C}^{\oplus r}\otimes_{\Bbb R}(C^k(U)
            \otimes_{\Bbb R}
	         C^k(\bigwedge^{\bullet}{\cal S}_U))\otimes_{C^k}C^k(Y))\,,
      $}			
  $$
   where $\otimes_{C^k}$ is the $C^k$-push-out,
   of $\pr_{\widehat{X}}^{\ast}(\widehat{\cal E})$
   on $\widehat{X}\times Y$ restricted to over $\widehat{U}\subset \widehat{X}$,
 the subsheaf
 $(\pr_Y^{\ast}(f)-\pr_{\widehat{X}}^{\ast}
      (\widehat{\varphi}^{\sharp}(f))\,:\, f\in C^k(Y))
					       \cdot \pr_{\widehat{X}}^{\ast}(\widehat{\cal E})$
  in the above lemma is generated (as an ${\cal O}_{\widehat{X}\times Y}^{\,\Bbb C}$-module)
  by elements of the form
  $$
     v\otimes f \; -\; (\widehat{\varphi}^{\sharp}(f)(v))\otimes 1\,,
	    \hspace{2em}f\in C^k(Y)\,.
  $$
 Here, $v$ represents an $r\times 1$ column-vector
  with coefficients in
   $C^k(U)\otimes_{\Bbb R}C^k(\bigwedge^{\bullet}{\cal S}_U)$.
}\end{ssremark}

\bigskip

\begin{flushleft}
{\bf Recovering $\widehat{\varphi}: (\widehat{X}^{\!A\!z}, \widehat{\cal E})\rightarrow Y$
          from an ${\cal O}_{\widehat{X}\times Y}^{\,\Bbb C}$-module}
\end{flushleft}
Conversely,
 let
 $\widehat{X}=(X, \widehat{\cal O}_X:=\bigwedge^{\bullet}{\cal S}) $
     be a super-$C^k$-manifold,
 $Y$ be a $C^k$-manifold,   and
 $\tilde{\widehat{\cal E}}$
  be a sheaf of  ${\cal O}_{\widehat{X}\times Y}^{\,\Bbb C}$-modules \
  with the following properties:
   \begin{itemize}
    \item[]
	  \begin{itemize}
        \item[(M1)]
		 The annihilator ideal sheaf
		   $\Ker({\cal O}_{\widehat{X}\times Y}\rightarrow
		                          \Endsheaf_{{\cal O}_{\widehat{X}\times Y}}
								         (\tilde{\widehat{\cal E}}))$
		   is super-$C^k$-normal in ${\cal O}_{\widehat{X}\times Y}$;
          thus,
   		  $\Supp(\tilde{\widehat{\cal E}})$
    		  is a super-$C^k$-subscheme of the super-$C^k$-manifold $\widehat{X}\times Y$.
         Assume that
		  $\Supp(\tilde{\widehat{\cal E}})$ is of relative dimension $0$ over $\widehat{X}$.
		
 	    \item[(M2)]
         The push-forward
		  $\widehat{\cal E}:= \pr_{\widehat{X},\ast}(\tilde{\widehat{\cal E}})$
	       is a locally free ${\cal O}_{\widehat{X}}^{\,\Bbb C}$-module of finite rank, say, $r$.
	 %
	 \end{itemize}	
   \end{itemize}
Then
 $$
    (X,
	    \widehat{\cal O}_X^{A\!z}
		     :=\Endsheaf_{\widehat{\cal O}_X^{\,\Bbb C}}(\widehat{\cal E}),
	    \widehat{\cal E})
 $$
   is an Azumaya/matrix super-$C^k$-manifold with a fundamental module  and
 $\tilde{\widehat{\cal E}}$ defines an equivalence class
  $$
    \widehat{\varphi}^{\sharp}:{\cal O}_Y\;    \longrightarrow\; \widehat{\cal O}_X^{A\!z}
  $$
  of contravariant gluing systems of $C^k$-admissible ring-homomorphisms
   over ${\Bbb R}\hookrightarrow{\Bbb C}$ related to $(\widehat{X}^{\!A\!z},Y)$
  as follows:
  \begin{itemize}
      \item[(1)]
       Let $\widehat{U}\subset \widehat{X}$ be an open set from an atlas of $X$
	   such that
   	   $\pr_Y(\Supp(\tilde{\widehat{\cal E}}_{\widehat{U}}))$
	      is contained in an open set $V\subset Y$ that lies in an atlas of $Y$.
	   Here, we treat $\tilde{\widehat{\cal E}}$ also as a sheaf over $\widehat{X}$    and
	      $\tilde{\widehat{\cal E}}_{\widehat{U}}
		     := \tilde{\widehat{\cal E}}|_{\widehat{U}\times Y}$
  			is the restriction of $\tilde{\widehat{\cal E}}$ to over $\widehat{U}$.		
		
      \item[(2)]
        Let $f\in C^k(V)$.
	    Then
		  the multiplication by $\pr_Y^{\ast}(f)\in C^k(\widehat{U}\times V)$
	       induces an endomorphism
		   $\tilde{\widehat{\alpha}}_f:
		      \tilde{\widehat{\cal E}}_{\widehat{U}}
			    \rightarrow  \tilde{\widehat{\cal E}}_{\widehat{U}}$
		     as an ${\cal O}_{\widehat{U}\times V}^{\,\Bbb C}$-module.
        Since
		  $\widehat{U}\times V\supset \Supp(\tilde{\widehat{\cal E}}_{\widehat{U}})$
		   and $\pr_Y^{\ast}(f)$ lies in the center of ${\cal O}_{\widehat{U}\times Y}$,
          $$
		    \widehat{\alpha}_f \;
			 :=\; {\pr_{\widehat{X}}}_{\ast}(\tilde{\widehat{\alpha}}_f)
		  $$
		  defines in turn
	      a $C^k$-endomorphism of the $\widehat{\cal O}_U^{\,\Bbb C}$-module
		  $\widehat{\cal E}_U$;
	      i.e.\ $\widehat{\alpha}_f \in \widehat{\cal O}_X^{A\!z}(U)$.
        This defines a ring-homomorphism
	      $\widehat{\varphi}^{\sharp}:C^k(V)\rightarrow \widehat{\cal O}_X^{A\!z}(U)$
	      over ${\Bbb R}\hookrightarrow{\Bbb C}$,
	      with $f\mapsto \widehat{\alpha}_f$.
        By construction, $\widehat{\varphi}^{\sharp}$ is $C^k$-admissible.		
		
      \item[(3)]
	   Compatibility of the system of $C^k$-admissible ring-homomorphisms
  	    $\widehat{\varphi}^{\sharp}:C^k(V)\rightarrow \widehat{\cal O}_X^{A\!z}(U)$
	     over ${\Bbb R}\hookrightarrow{\Bbb C}$
        with gluings follows directly from the construction.		
  \end{itemize}
In this way, $\tilde{\widehat{\cal E}}$  defines a $C^k$-map
  $\widehat{\varphi}:(\widehat{X}^{\!A\!z}, \widehat{\cal E})\rightarrow Y$.
  
By construction,
 the graph $\tilde{\widehat{\cal E}}_{\widehat{\varphi}}$
  of the $C^k$-map $\widehat{\varphi}$ associated to $\tilde{\widehat{\cal E}}$
  is canonically isomorphic to $\tilde{\widehat{\cal E}}$.	
This gives an equivalence of the two notions/categories:
 $$
  \begin{array}{c} \\[-1.2ex]
    \fbox{ $C^k$-maps $\widehat{\varphi}:(\widehat{X}^{\!A\!z}, \widehat{\cal E})
	  \rightarrow Y$ \raisebox{-1.4ex}{\rule{0em}{1.8em}}}\;
	   \Longleftrightarrow\;
    \fbox{ ${\cal O}_{\widehat{X}\times Y}^{\,\Bbb C}$-modules
	              $\tilde{\widehat{\cal E}}$  that satisfy (M1) and (M2) }	   \\[2ex]
  \end{array}	
 $$
({\sc Figure}~4-2-2-1.)
%
%
\begin{figure} [htbp]
 \bigskip
 \centering
 \includegraphics[width=0.80\textwidth]{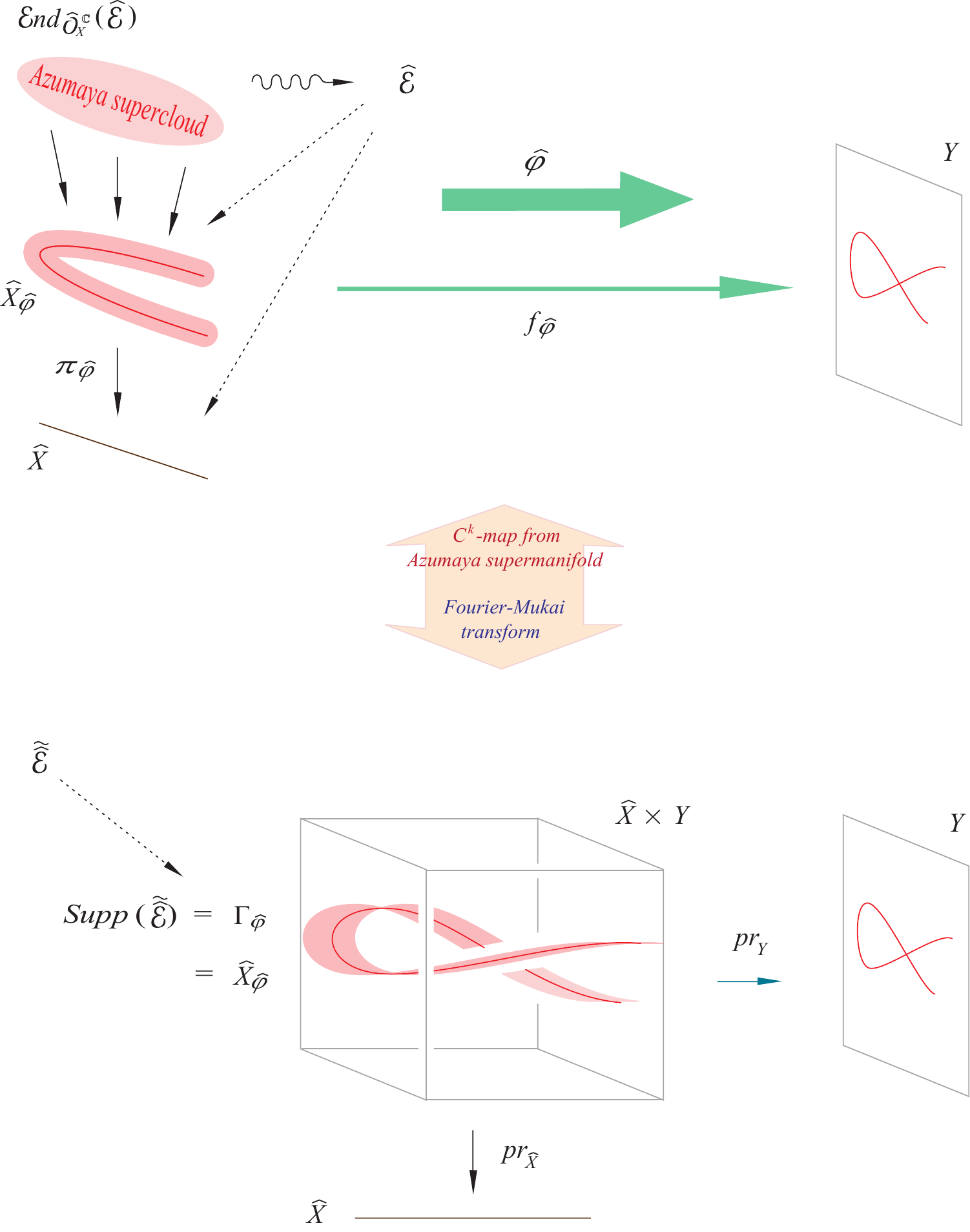}
 
 \vspace{4em}
 \centerline{\parbox{13cm}{\small\baselineskip 12pt
  {\sc Figure}~4-2-2-1.
    The equivalence between
      a $C^k$-map $\widehat{\varphi}$
     	  from an Azumaya/matrix super-$C^k$-manifold with a fundamental module
	    $(\widehat{X}, 
		        \widehat{\cal O}_X^{A\!z}
		         :=\Endsheaf_{\widehat{\cal O}_X^{\,\Bbb C}}(\widehat{\cal E}),
		          \widehat{\cal E})$
	    to a $C^k$-manifold $Y$    and
	  a special kind of Fourier-Mukai transform
	    $\tilde{\widehat{\cal E}}\in \ModCategory^{\Bbb C}(\widehat{X}\times Y)$ 
		from $\widehat{X}$ to $Y$. 		
   Here, $\ModCategory^{\Bbb C}(\widehat{X}\times Y)$
        is the category of ${\cal O}_{\widehat{X}\times Y}^{\,\Bbb C}$-modules.
      }}
 \bigskip
\end{figure}

\bigskip

\subsubsection{Aspect III: From maps to the stack of D0-branes}

Aspect II of a $C^k$-map
 $\widehat{\varphi}:(\widehat{X}^{\!A\!z}, \widehat{\cal E})\rightarrow Y$
   discussed in Sec.~4.2.2
 brings out a third aspect of $\widehat{\varphi}$, which we now explain.

\bigskip
 
\begin{flushleft}
{\bf An Azumaya/matrix supermanifold  $\widehat{X}^{\!A\!z}$
           as a smearing of unfixed Azumaya/matrix points over the underlying supermanifold $\widehat{X}$}
\end{flushleft}
We shall illuminate this in three steps.
First, recall the following example in complex analysis:

\bigskip

\begin{ssexample} {\bf [real contour $\gamma$ in complex line ${\Bbb C}^1$].}
{\rm ([L-Y3 : Example~5.3.3.1] (D(11.1)).)
 A  differentiable contour in the complex line ${\Bbb C}^1$ with complex coordinate $z=x+\sqrt{-1}y$
   is a differentiable map
   $$
     \gamma\; =\; \gamma_x + \sqrt{-1}\gamma_y\; :\; [0,1]\; \longrightarrow\;  {\Bbb C}^1\,.
   $$
 There is no issue about this, if $\gamma$ is treated as  a map between point-sets with a manifold structure:
    from the interval $[0,1]$ to the underlying real $2$-space ${\Bbb R}^2$ of ${\Bbb C}^1$
	with coordinates $(x,y)$.
 However,  in terms of function rings, some care needs to be taken.
 While there is a built-in ring-homomorphism ${\Bbb R}\hookrightarrow{\Bbb C}$ over ${\Bbb R}$,
   there exists no ring-homomorphism ${\Bbb C}\rightarrow{\Bbb R}$ with $0\mapsto 0$ and $1\mapsto 1$.
 If follows that there is no ring-homomorphism
   $\gamma^{\sharp}: C^k({\Bbb C}^1)^{\Bbb C}\rightarrow C^k([0,1])$,
   where $C^k({\Bbb C}^1)^{\Bbb C}$ is the algebra of complex-valued $C^k$-functions on ${\Bbb C}^1$.
 {To} remedy this, one should first complexify $C^k([0,1])$ to
   $$
      C^k([0,1])^{\Bbb C}\;  :=\;   C^k([0,1])\otimes_{\Bbb R}{\Bbb C}\,;
   $$
   then there is a well-defined algebra-homomorphism over ${\Bbb C}$
   $$
	  \gamma^{\sharp}\;:\; C^k({\Bbb C}^1)^{\Bbb C}\;
	       \longrightarrow\;    C^k([0,1])^{\Bbb C}
   $$
   by the pull-back of functions via $\gamma$.
 Here comes the guiding question:
    \begin{itemize}
	  \item[{\bf Q.}]   {\it What
	    is the geometric meaning of the above algebraic operation?}	
	\end{itemize}
 The answer comes from an input to differential topology from algebraic geometry.
 
 By definition, a point with function field ${\Bbb R}$ is an {\it ${\Bbb R}$-point}
  while a point with function field ${\Bbb C}$ is a {\it ${\Bbb C}$-point}.
 Topologically they are the same but algebraically they are different,  as already indicated by
   $$
     {\Bbb R}\; \hookrightarrow\;  {\Bbb C}\,, \hspace{1em}\mbox{while}\hspace{1em}
	   {\Bbb C}\; \hspace{2ex}/\hspace{-3ex}\longrightarrow\;  {\Bbb R}\,,
   $$
   which means algebrao-geometrically, concerning the existence of a map from one to the other,
   $$
       \mbox{\it ${\Bbb C}$-point}\; \longrightarrow\;  \mbox{\it ${\Bbb R}$-point}\,,
	    \hspace{1em}\mbox{while}\hspace{1em}
	   \mbox{\it ${\Bbb R}$-point}\; \hspace{2ex}/\hspace{-3ex}
	     \longrightarrow\;  \mbox{\it ${\Bbb C}$-point}\,.
   $$
  By replacing $C^k([0,1])$ by its complexification $C^k([0,1])^{\Bbb C}$,
   we promote each original ${\Bbb R}$-points on $[0,1]$ to a ${\Bbb C}$-point.
  In other words, we smear ${\Bbb C}$-points along the interval $[0,1]$.
 The map $\gamma$ now simply specifies a $C^k$ $[0,1]$-family of ${\Bbb C}$-points on ${\Bbb C}^1$
    by associating to each ${\Bbb C}$-point on $[0,1]$ a ${\Bbb C}$-point on ${\Bbb C}^1$,
	which is now allowed algebro-geometrically.
 This concludes the example
 
\noindent\hspace{40.7em}$\square$
}\end{ssexample}

\bigskip

Let $p^{A\!z}$ be a point with function ring isomorphic to the endomorphism algebra
 $\End_{\Bbb C}({\Bbb C}^{\oplus r})$
Then, recall from [L-Y3] (D(11.1))
  that  by exactly the same reasoning and geometric pictures as in Example~4.2.3.1,
   with $(\,\cdots\,)\otimes_{\Bbb R}{\Bbb C}$
            replaced by $(\,\cdots\,)\otimes_{{\Bbb R}}\End_{\Bbb C}(E)$
			  locally  
   where $E$ is a ${\Bbb C}$-vector space and $(\,\cdots\,)$ is the $C^k$-ring in question,
 one has	{\small
 $$
  \begin{array}{lcl} \\[-1.2ex]
    \fbox{Azumayanized manifold
	  $(X, {\cal O}_X
	                \otimes_{\Bbb R}\End_{\Bbb C}({\Bbb C}^{\oplus r}))$  }\;
	   & \Longleftrightarrow\;
       & \fbox{the smearing of {\it fixed} $p^{A\!z}$'s along $X$}	   \\[2ex]
    \fbox{general Azumaya manifold $(X^{\!A\!z},{\cal E})$}\;
	   & \Longleftrightarrow\;
       & \fbox{a smearing of {\it unfixed} $p^{A\!z}$'s along $X$}   \\[2ex]	
  \end{array}	
 $$   }
 
Finally, by the same reasoning but with $(\,\cdots\,)$  above replaced by the super-$C^k$-ring in question,
 one has completely analogously	{\small
 $$
  \hspace{-1ex}
  \begin{array}{lcl} \\[-1.2ex]
    \fbox{Azumayanized supermanifold
	  $(X, \widehat{\cal O}_X
	                \otimes_{\Bbb R}\End_{\Bbb C}({\Bbb C}^{\oplus r}))$  }\;
	   & \Longleftrightarrow\;
       & \fbox{the smearing of {\it fixed} $p^{A\!z}$'s along $\widehat{X}$}	   \\[2ex]
    \fbox{general Azumaya supermanifold $(\widehat{X}^{\!A\!z}, \widehat{\cal E})$}\;
	   & \Longleftrightarrow\;
       & \fbox{a smearing of {\it unfixed} $p^{A\!z}$'s along $\widehat{X}$}   \\[2ex]	
  \end{array}	
 $$   }
  
\bigskip

\begin{ssremark}$[\,$another smearing$\,]$. {\rm
 Though for our purpose the above viewpoint is preferred,
  there is however a second viewpoint:
 \begin{itemize}
   \item[{\Large $\cdot$}]
    Let $\widehat{p}^{A\!z}$ be an Azumaya/matrix superpoint with function ring
	  $\End_{\Bbb C}({\Bbb C}^{\oplus r})
	      \otimes_{\Bbb R} \bigwedge^{\bullet}({\Bbb R}^{\oplus s})$.
    Then $\widehat{X}^{\!A\!z}$	can be regarded as smearing unfixed $\widehat{p}^{A\!z}$'s
      along $X$ as well.	
 \end{itemize}
}\end{ssremark}

\clearpage

\begin{flushleft}
{\bf A $C^k$-map
          $\widehat{\varphi}:(\widehat{X}^{\!A\!z},\widehat{\cal E})\rightarrow Y$
       as smearing D0-branes on $Y$ along $\widehat{X}$}
\end{flushleft}
{To} press on along this line, we have to list two objects that are studied in algebraic geometry and
 yet their counter-objects are much less known/studied in differential topology/geometry:
\begin{itemize}
 \item[(1)] [{\it Quot-schemes}$\,$] \hspace{1em}
   Grothendieck's $\Quot$-scheme $\Quot_Y^{r}(({\cal O}_Y^{\,\Bbb C})^{\oplus r})$
    of $0$-dimensional quotient sheaves of $({\cal O}_Y^{\,{\Bbb C}})^{\oplus r}$
	of complex-length $r$.
   This is the parameter space of differentiable maps from the fixed Azumaya point
	  $(p^{A\!z},{\Bbb C}^{\oplus r})$  to $Y$;
	cf.\  [L-Y3: Sec.~3, Lemma/Definition~5.3.1.9,  Sec.~5.3.2] (D(11.1)).
   In other words, it parameterizes D0-branes ${\cal F}$ on $Y$
     (where ${\cal F}$ is a complex $0$-dimensional sheaf on $Y$ of complex length $r$)
     that is decorated with an isomorphism
	  ${\Bbb C}^{\oplus r}\stackrel{\sim}{\rightarrow}C^k({\cal F})$ over ${\Bbb C}$. 	
    	
 \item[(2)] [{\it Quotient stacks}$\,$] \hspace{1em}
  The general linear group $\GL_r({\Bbb C})$
    acts on  $\Quot_Y^{r}({\cal O}_Y^{\,\Bbb C})^{\oplus r})$
	by its tautological action on the ${\Bbb C}^{\oplus r}$-factor in the canonical isomorphism	
   $({\cal O}_Y^{\,{\Bbb C}})^{\oplus r}
       \simeq {\cal O}_Y\otimes_{\Bbb R}{\Bbb C}^{\oplus r}$.
 This defines a quotient stack
   $[ \Quot_Y^{r}(({\cal O}_Y^{\,\Bbb C})^{\oplus r})/\GL_r({\Bbb C})  ]$,
   which now parameterizes differentiable maps $\varphi$ from unfixed Azumaya points
	  $(p^{A\!z}, E)$, where $E$ is a ${\Bbb C}$-vector space of rank $r$,  to $Y$.
   In other words,  $[ \Quot_Y^{r}(({\cal O}_Y^{\,\Bbb C})^{\oplus r})/\GL_r({\Bbb C})  ]$
     is precisely the moduli stack ${\frak M}_r^{\,0^{A\!z^{\!f}}}\!\!(Y)$
	 of D0-branes of complex length $r$ on $Y$,
    realized as complex $0$-dimensional sheaves on $Y$ of complex length $r$ via push-forwards
	$\varphi_{\ast}(E)$, from	[L-Y3: Definition~5.3.1.5] (D(11.1));
   cf. Definition~4.2.1.3 and Remark~4.2.1.7.
   
\end{itemize}

Recall from Sec.~4.2.2 that a differentiable map
  $\widehat{\varphi}:(\widehat{X}^{\!A\!z}, \widehat{\cal E})\rightarrow Y$
  is completely encoded by its graph $\tilde{\widehat{\cal E}}_{\widehat{\varphi}}$
  on $\widehat{X}\times Y$.
Over any super-$C^k$-subscheme $\widehat{Z}\subset \widehat{X}$,
    $\tilde{\widehat{\cal E}}_{\widehat{\varphi}}|_{\widehat{Z}\times Y}$
  is simply a flat $\widehat{Z}$-family of
     $0$-dimensional ${\cal O}_Y^{\,{\Bbb C}}$-modules of complex length $r$.
Despite missing the details of these parameter ``spaces", it follows from their definition
 as functors or sheaves of groupoids over the category of super-$C^k$-schemes,
   with suitable Grothendieck topology, that
 $$
  \begin{array}{c} \\[-1.2ex]
    \fbox{\rule{0em}{1.2em}
	   $\,C^k$-maps
	     $\widehat{\varphi}:(\widehat{X}^{\!A\!z}, \widehat{\cal E})\rightarrow Y$ }\;
	   \Longleftrightarrow\;
    \fbox{ admissible maps
	     $\widehat{X}\rightarrow {\frak M}_r^{\,0^{A\!z^{\!f}}}\!\!(Y)$ }	   \\[2ex]
  \end{array}	
 $$
This matches perfectly with
   the picture of an Azumaya/matrix supermanifold $(\widehat{X}^{\!A\!z}, \widehat{\cal E})$
      as a smearing of unfixed Azumaya/matrix points $p^{A\!z}$ along $\widehat{X}$
 since, then, a map $\widehat{X}^{\!A\!z}\rightarrow Y$
  is nothing but an $\widehat{X}$-family of maps
  $p^{A\!z}\rightarrow Y$, which is exactly the map
  $\widehat{X} \rightarrow {\frak M}_r^{\,0^{A\!z^{\!f}}}\!\!(Y)$.
 Cf.~{\sc Figure}~4-2-3-1.
%
%
\begin{figure}[htbp]
 \bigskip
  \centering
  \includegraphics[width=0.80\textwidth]{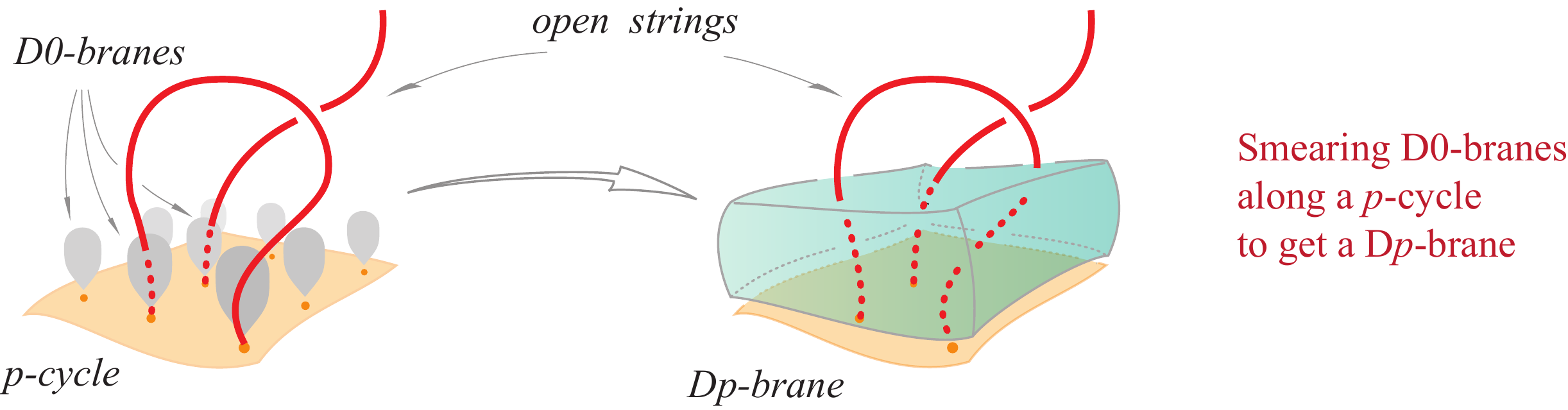}
 
  \bigskip
  \bigskip
 \centerline{\parbox{13cm}{\small\baselineskip 12pt
  {\sc Figure}~4-2-3-1. ([L-L-S-Y: {\sc Figure}~3-1-1].)
  The original stringy operational definition of D-branes as objects
     in the target space(-time) $Y$ of fundamental strings
   where end-points of open-strings can and have to stay
  suggests that smearing D$0$-branes along
    a (real) $p$-dimensional submanifold $X$ in $Y$ renders $X$ a D$p$-brane.	
 Such a smearing can be generalized to the supercase and 
   is realized as a map from the supermanifold $\widehat{X}$ (over and containing $X$)
   to the stack ${\frak M}^{D0}(Y)$ of D$0$-branes on $Y$.
 In the figure, the Chan-Paton sheaf ${\cal E}$ that carries
   the index information on the end-points of open strings
   is indicated by a shaded cloud.   
  }}
\end{figure}

\bigskip

\subsubsection{Aspect IV: From associated $\GL_r({\Bbb C})$-equivariant maps}

Again, we list a parallel issue in differential topology that remains to be understood:
\begin{itemize}
 \item[(1)] [{\it Fibered product}$\,$] \hspace{1em}
  The notion of the fibered product of stratified singular spaces with a structure sheaf
  and its generalization to stacks needs to be developed.
\end{itemize}
Subject to this missing detail, from the very meaning of a quotient stack, it is natural to anticipate that
 any natural definition of the notion of fibered product should lift a map
  $$
    \widehat{X}\;   \longrightarrow\;
   	{\frak M}_r^{\,0^{A\!z^{\!f}}}\!\!(Y)\,
       =\, [ \Quot_Y^{r}(({\cal O}_Y^{\,\Bbb C})^{\oplus r})/\GL_r({\Bbb C})  ]
  $$
 lift to a $\GL_r({\Bbb C})$-equivariant map
  $$
     P_{\widehat{X}}\;
	   \longrightarrow\;  \Quot_Y^r(({\cal O}_Y^{\,{\Bbb C}})^{\oplus r})\,,
   $$
 where $P_{\widehat{X}}$ is a principal $\GL_r({\Bbb C})$-bundle over $\widehat{X}$
  from the fibered product
  $$
   P_{\widehat{X}}\;
    :=\;   \widehat{X}  \times_{{\frak M}_r^{\,0^{A\!z^{\!f}}}\!\!(Y)}
	             \Quot_Y^{r}(({\cal O}_Y^{\,\Bbb C})^{\oplus r})
  $$
Conversely, any latter map should define a former map.
Together with Aspect III in Sec.~4.2.3, this gives a correspondence
 $$
  \begin{array}{l} \\[-1.2ex]
    \fbox{$\,C^k$-maps
	  $\widehat{\varphi}:(\widehat{X}^{\!A\!z}, \widehat{\cal E})\rightarrow Y$} \\[2ex]
	   \hspace{6em}\Longleftrightarrow\;
    \fbox{admissible $\GL_r({\Bbb C})$-equivariant  maps
	           $P_{\widehat{X}}
			      \rightarrow \Quot_Y^r(({\cal O}_Y^{\,{\Bbb C}})^{\oplus r})$}	   \\[2ex]
  \end{array}	
 $$

\bigskip

\subsection{Remarks on differentiable maps from a general endomorphism-ringed super-$C^k$-scheme
                          with a fundamental module to a real supermanifold} 	

Recall the last theme
  `{\sl Remarks on general endomorphism-ringed super-$C^k$-schemes and
         differential calculus thereupon}'	
   of Sec.~3,
 in which the notion of {\it endomorphism-ringed super-$C^k$-schemes with a fundamental module}
 is introduced, which generalizes the notion of Azumaya/matrix super-$C^k$-schemes with a fundamental module.
Similar notion of
 `differentiable maps from an endomorphism-ringed super-$C^k$-scheme with a fundamental module
 to a real supermanifold'
 is readily there from a generalization of the discussions in Sec.~4.1 and Sec.~4.2.
  
Consider the local study first.
Let
 \begin{itemize}
  \item[{\Large $\cdot$}]
   $R\simeq C^k({\Bbb R}^n)/I$ (i.e.\ $R$ a finitely generated $C^k$-ring) and\\
   $\mbox{\hspace{-1.2ex}}$
   $\xymatrix{\widehat{R}\ar@<.3ex>[r]  & R\ar@<.3ex>[l]}$ be a super-$C^k$-ring over $R$,
 
  \item[{\Large $\cdot$}]
   $S=C^k(V)$ where $V$ be an open set of ${\Bbb R}^n$ and\\
   $\mbox{\hspace{-1.2ex}}$
   $\xymatrix{\widehat{S}\ar@<.3ex>[r]  & S\ar@<.3ex>[l]}$ be a superpolynomial ring over $S$.
 \end{itemize}
Then
 \begin{itemize}
  \item[{\Large $\cdot$}]
   The {\it push-out} $R\otimes_{C^k}S$ of $R$ and $S$ exists in the category of $C^k$-rings
     with  $R\otimes_{C^k}S  \simeq C^k({\Bbb R}^n\times V)/I$
     and {\it with the built-in $C^k$-ring-homomorphisms
    $\pr_1^{\,\sharp}:R\hookrightarrow R\otimes_{C^k}S$ and
    $\pr_2^{\,\sharp}:S\hookrightarrow R\otimes_{C^k}S$}
    coincident with pulling back of functions via the projections $\pr_1$ and $\pr_2$
	of the product space to its factors.
   Here $I \hookrightarrow  C^k({\Bbb R}^n \times V)$
    via the canonical inclusion $C^k({\Bbb R}^n)\subset C^k({\Bbb R}^n \times V)$.	
	
  \item[{\Large $\cdot$}]	
   The {\it push-out} $\widehat{R}\otimes_{sC^k}\widehat{S}$
    	of $\widehat{R}$ and $\widehat{S}$ exists in the category of super-$C^k$-rings
  {\it with the built-in super-$C^k$-ring-homomorphisms
    $\widehat{\pr}_1^{\,\sharp}:\widehat{R}
	    \hookrightarrow \widehat{R}\otimes_{sC^k}\widehat{S}$		
		and
    $\widehat{\pr}_2^{\,\sharp}:\widehat{S}
	   \hookrightarrow \widehat{R}\otimes_{sC^k}\widehat{S}$}.
	
 \item[{\Large $\cdot$}]	
  One has a built-in commutative diagram of morphisms
  $$
   \xymatrix{
      &&&&&&& \;\;  \widehat{S} \;\; \ar@<.3ex>[dd]
                                                   	                \ar@{_{(}->}[dl]_-{\widehat{pr}_2^{\sharp}}  \\
	  &\;\; \widehat{R}\;\;
                	  \ar@<.3ex>[dd]
					  \ar@{^{(}->}[rrrrr]^-{\widehat{pr}_1^{\sharp}} 	                                      
	     &\rule{0ex}{1ex}  &  \rule{0ex}{1ex}
		 &&& \;\; \widehat{R}\otimes_{sC^k}\widehat{S}\ar@<.3ex>[dd]\;  \\
      &&
		 &&&&& \;\; S \ar@<.3ex>[uu]
		                         \ar@{_{(}->}[dl]_-{pr_2^{\sharp}}  \;\;    \\
	  &\;\; R\;\;  \ar@<.3ex>[uu]
				        \ar@{^{(}->}[rrrrr]^-{pr_1^{\sharp}} 	
	     &&&&& \;\; R\otimes_{C^k} S  \ar@<.3ex>[uu]		
			        &\;\;\;\;\;\;\;\;. 			
    }
  $$
 \end{itemize}	
 
Let
 \begin{itemize}
  \item[{\Large $\cdot$}]
   $\widehat{M}$ be a finitely generated $\widehat{R}^{\,\Bbb C}$-module   and
   
  \item[{\Large $\cdot$}]
   $\End_{\widehat{R}^{\,\Bbb C}}(\widehat{M})$
   be the $\widehat{R}^{\,\Bbb C}$-algebra
   of $\widehat{R}^{\,\Bbb C}$-endomorphisms of $\widehat{M}$.
 \end{itemize}	
We {\it assume that
  the built-in $\widehat{R}$-algebra-homomorphism
  $\widehat{R}\rightarrow \End_{\widehat{R}^{\,\Bbb C}}(\widehat{M})$ is injective}.
We will denote this inclusion by $\widehat{\pi}^{\sharp}$.
Then, the key notion in the whole setting is the following:

\bigskip

\begin{definition}{\bf [$C^k$-admissible superring-homomorphism].} {\rm
(Cf.\ [L-Y3: Definition~5.1.2] (D(11.1)) and
           Definition~4.1.1 in Sec.~4.1.)
 A superring-homomorphism
  $$
    \widehat{\varphi}^{\sharp}\;:\;
      \widehat{S}\;\longrightarrow\; \End_{\widehat{R}^{\,\Bbb C}}(\widehat{M})
  $$	
  over ${\Bbb R}\hookrightarrow {\Bbb C}$
  is said to be {\it $C^k$-admissible}
 if the arrow $\widehat{\varphi}^{\sharp}$ extends to the following commutative diagram
  of superring-homomorphisms
    $$
	  \xymatrix{
	    \End_{\widehat{R}^{\,\Bbb C}}(\widehat{M})
		    && \widehat{S}\ar[ll]_-{\varphi^{\sharp}}\raisebox{-.6ex}{\rule{0ex}{1ex}}
			                        \ar@{_{(}->}[d]^-{\widehat{pr}_2^{\sharp}}  \\
         \hspace{1ex}\widehat{R}	\rule{0em}{1.2em}\hspace{1ex}
       		    \ar@{^{(}->}[u]^-{\widehat{\pi}^{\sharp}}		
				\ar@{^{(}->}[rr]_-{\widehat{pr}_1^{\sharp}}
			&& \widehat{R}\otimes_{sC^k}\widehat{S}
			         \ar[llu]_-{\tilde{\widehat{\varphi}}^{\sharp}}
	   }
	$$	
    such that
	 \begin{itemize}
	  \item[{\Large $\cdot$}]
	   $\Ker(\tilde{\widehat{\varphi}}^{\sharp})$ is super-$C^k$-normal and hence
       $\Image\tilde{\widehat{\varphi}}^{\sharp}$	
	     can be equipped with a quotient super-$C^k$-ring structure
	     from that of $\widehat{R}\otimes_{sC^k}\widehat{S}$
		 via $\tilde{\widehat{\varphi}}^{\sharp}$.
	
	 \item[{\Large $\cdot$}]
	  With respect to the super-$C^k$-ring structure on $\Image\tilde{\widehat{\varphi}}^{\sharp} $,
	   $\widehat{\varphi}^{\sharp}$ is a super-$C^k$-ring-homomorphism
	    as a superring-homomorphism
	    $\widehat{S}\rightarrow \Image\tilde{\widehat{\varphi}}^{\sharp}$.	
	 \end{itemize}
}\end{definition}

\bigskip
  
Moving on to the global study.
Let
    $\widehat{X}:=(X, {\cal O}_X, \widehat{\cal O}_X)$ be a super-$C^k$-scheme   and
	 %
	 %
	$\widehat{\cal F}$
	  be a finitely generated quasi-coherent $\widehat{\cal O}_X^{\,\Bbb C}$-module.
We shall {\it assume that the built-in $\widehat{\cal O}_X$-algebra-homomorphism
   $\widehat{\cal O}_X
       \rightarrow \Endsheaf_{\widehat{\cal O}_X^{\,\Bbb C}}(\widehat{\cal F})$
   is injective}.
Consider the endomorphism-ringed super-$C^k$-scheme with a fundamental module
   $$
     (\widehat{X}^{\nc}, \widehat{\cal F})\;
	    :=\; (\widehat{X},
     		       \widehat{\cal O}^{\nc}_X
            		:=\Endsheaf_{\widehat{\cal O}_X^{\,\Bbb C}}(\widehat{\cal F}),
                  \widehat{\cal F})\,.
   $$
   
\bigskip

\begin{definition}    {\bf [differentiable map].} {\rm
(Cf.\ [L-Y3: Definition~5.3.1.5] (D(11.1))   and
           Definition~4.2.1.3 in Sec.~4.2.1.)
 Let $\widehat{Y}$ be a super-$C^k$-manifold.						
 A {\it $k$-times differentiable map} (i.e.\ {\it $C^k$-map})
   $$
     \widehat{\varphi}\; :\;
	   (\widehat{X}^{\nc}, \widehat{\cal F})\; \longrightarrow\; \widehat{Y}
   $$
    is defined contravariantly
    as an {\it equivalence class of gluing systems of $C^k$-admissible superring-homomorphisms},
	in notation,
   $$
     \widehat{\varphi}^{\sharp}\; :\;
	   \widehat{\cal O}_Y\; \longrightarrow\;  \widehat{\cal O}_X^{\nc}\,,
   $$
  exactly as in Definition~4.2.1.1, Definition~4.2.1.2, and Definition~4.2.1.3 in Sec.~4.2.1.
}\end{definition}					

\bigskip					
					
As in
     [L-Y3: Sec.~5.3.1] (D(11.1))    and
     Sec.~4.2 for the case of $C^k$-maps from an Azumaya/matrix super-$C^k$-manifold to a $C^k$-manifold,
 one has the following well-defined basic notions:
 \begin{itemize}
   \item[{\Large $\cdot$}]	
    the {\it surrogate} $\widehat{X}_{\widehat{\varphi}}$ of $\widehat{X}^{\nc}$
  	specified by $\widehat{\varphi}$,

   \item[{\Large $\cdot$}]
    the {\it push-forward}  $\widehat{\varphi}_{\ast}\widehat{\cal F}$ of $\widehat{\cal F}$
	to $\widehat{Y}$,
	
   \item[{\Large $\cdot$}]	
    the {\it graph} $\tilde{\widehat{\cal F}}_{\widehat{\varphi}}$ of $\widehat{\varphi}$,
	which is an $\widehat{\cal O}_{X\times Y}^{\,\Bbb C}$-module
	on $\widehat{X}\times\widehat{Y}$.
 \end{itemize}	
One has now {\it Aspect I} and {\it Aspect II} for  $\varphi$.
When $\widehat{\cal F}$ is in addition locally free and $\widehat{Y}=Y$,
 one recovers the notion of $C^k$-map in Sec.~4.2.1
 and has in addition Aspect III and Aspect IV.
 
Recall the built-in inclusion $\widehat{\iota}: X\hookrightarrow \widehat{X}$ and let
 $$
  {\cal F}\; := \; \widehat{\iota}^{\ast}\widehat{\cal F}
 $$
 be the restriction of $\widehat{\cal F}$ to $X$.
Then, $\widehat{\varphi}$ induces a $C^k$-map
 $$
   \varphi\; :\;
     (X^{\nc},{\cal F})
	   := (X,
	          {\cal O}_X^{\nc}:=\Endsheaf_{{\cal O}_X^{\,\Bbb C}}({\cal F}),
			    {\cal F})\;
      \longrightarrow\;  Y				
 $$
 with a built-in commutative diagrams of morphisms
 $$
   \xymatrix{
    \; \widehat{\cal F}
	          \ar@{.>}[rrrd]     \ar@{.>}@/_1ex/[rrdd]    \ar@{.>}@/_2ex/[rddd]    \\
      &&& \widehat{X}^{\nc}                            	
	                                    \ar[rrrrd]^-{\widehat{\varphi}}
	                                    \ar@{->>}[ld]_-{\sigma_{\widehat{\varphi}}}
										\ar@<.3ex>@{->>}'[d]'[dd][dddd]^(.3){\dot{\widehat{\pi}}}
										\\
      &&\;\; \widehat{X}_{\widehat{\varphi}}\;\;
	                                   \ar[rrrrr]^(.43){f_{\widehat{\varphi}}}
	                                   \ar@{_{(}->} [rrrrd]^(.6){\tilde{\widehat{\varphi}}}
                                	   \ar@{->>}[ld]_-{\pi_{\widehat{\varphi}}}
									   \ar@<.3ex>@{->>}'[d][dddd]^(.3){\widehat{\pi}}
		 & \raisebox{-1.5em}{ \rule{0em}{2em}}
		 &&&& \;\;  \widehat{Y} \;\;  \ar@<.3ex>@{->>}[dddd]^(.42){\widehat{\pi}} \\
	  &\;\; \widehat{X}\;\;     \ar@<.3ex>@{->>}[dddd]^(.42){\widehat{\pi}}
	     &\rule{0ex}{1ex}  &  \rule{0ex}{1ex}
		 &&& \;\; \widehat{X}\times \widehat{Y}
		                   \ar@{->>}[ru]_-{pr_{\widehat{Y}}}
						   \ar@{->>}[lllll]^(.4){pr_{\widehat{X}}}
						   \ar@<.3ex>@{->>}[dddd]^(.4){\widehat{\pi}}
						   \;  \\
    \; {\cal F}:= \widehat{\iota}^{\ast}\widehat{\cal F}
	          \ar@{.>}[rrrd]     \ar@{.>}@/_1ex/[rrdd]    \ar@{.>}@/_2ex/[rddd]    \\
      &&& X^{\nc}\rule{0em}{1.2em}
	                                    \ar[rrrrd]^-{\varphi}
	                                    \ar@{->>}[ld]^-{\sigma_{\varphi}}
										\ar@<.3ex>@{^{(}->}'[uu]^(.7){\dot{\widehat{\iota}}}'[uuu][uuuu]
										\\
      &&\;\; X_{\varphi}\rule{0em}{1.2em}\;\;
	                                   \ar'[rrrr]^(.6){f_{\varphi}}[rrrrr]
	                                   \ar@{_{(}->} [rrrrd]^(.6){\tilde{\varphi}}
                                	   \ar@{->>}[ld]^-{\pi_{\varphi}}
									   \ar@<.3ex>@{^{(}->}'[uuu]^(.7){\widehat{\iota}} [uuuu]
		 &&&& \rule{2ex}{0ex}
		 & \;\; Y \rule{0em}{1.2em}\;\;
		              \ar@<.3ex>@{^{(}->}[uuuu]^(.58){\widehat{\iota}}		 \\
	  &\;\; X\rule{0em}{1.2em}\;\;
	            \ar@<.3ex>@{^{(}->}[uuuu]^(.58){\widehat{\iota}}
	     &&&&& \;\; X\times Y \rule{0em}{1.2em}
		                   \ar@{->>}[ru]_-{pr_Y}
						   \ar@{->>}[lllll]^(.4){pr_X}
                           \ar@<.3ex>@{^{(}->}[uuuu]^(.6){\widehat{\iota}} 		
			        &\;\;\;\;\;\;. 			
    }
  $$
 as in Proposition~4.2.1.8.

\bigskip

\section{A glimpse of super D-branes, as dynamical objects, and the Higgs mechanism in the current setting}

We give in this section a glimpse of
  super D-branes, as dynamical objects in string theory,  and
  the Higgs mechanism on D-branes
 in the current setting.
It serves to give readers
  a taste of applications to string theory
	and
 a bridge to sequels of the current note.

\bigskip

\subsection{Fermionic D-branes as fundamental/dynamical objects in string theory}

There are two versions of fermionic (fundamental, either open or closed) strings:
\begin{itemize}
  \item[(1)] {\it Ramond-Neveu-Schwarz} ({\it RNS}) {\it fermionic string},
   for which world-sheet spinors are manifestly involved
   ([N-S] of  Andr\'{e} Neveu and John Schwarz
   and [Ra] of Pierre Ramond);
   
  \item[(2)]  {\it Green-Schwarz} ({\it GS}) {\it fermionic string},
   for which space-time spinors are manifestly involved\\ ([G-S] of Michael Green and John Schwarz).
\end{itemize}
Mathematicians are referred particularly to [G-S-W: Chap.\  4 \& Chap.\ 5] of Green, Schwarz, and Witten
 for thorough explanations.
Once having the notion of differentiable maps from Azumaya/matrix manifold to a real manifold
  ([L-Y3] (D(11.1)))
  and its super-extension (Sec.~4.2 of the current note),
it takes no additional work to give a prototypical definition of {\it fermionic D-branes}
  in the style of either Ramond-Neveu-Schwarz or Green-Schwarz fermionic string
 once one understands
   the meaning of such fermionic strings from the viewpoint of Grothendieck's Algebraic Geometry.

\bigskip

\begin{flushleft}
{\bf Ramond-Neveu-Schwarz fermionic string and Green-Schwarz fermionic string
          from the viewpoint of Grothendieck's Algebraic Geometry}
\end{flushleft}
This discussion in this theme follows [G-S-W: Chap.\ 4 \& Chap.\ 5]
   (with possibly some mild change of notations to be compatible with the current note)  and [Ha: Chap.\ II].
Let
  ${\Bbb M}^{(d-1)+1}$   be the $d$-dimensional Minkowski space-time with coordinates
    $y:= (y^{\mu})_{\mu}=(y^0, y^1,\,\cdots\,, y^{d-1})$ and
  $\Sigma\simeq {\Bbb R}^1\times S^1$ or ${\Bbb R}^1\times [0,2\pi]$ be a string world-sheet 	
   with coordinates $\sigma:=(\sigma^0,\sigma^1)$.
   
\bigskip

\noindent
$(a)$ {\it Ramond-Neveu-Schwarz} ({\it RNS}) {\it fermionic string}

\medskip

\noindent
In this setting,
 there are both bosonic (world-sheet scalar) fields $y^{\mu}(\sigma)$
   and fermionic (world-sheet spinor) fields $\psi^{\mu}(\sigma)$
 on the string world-sheet $\Sigma$ for $\mu=0,1, \,\cdots\,, d-1$.
The former  collectively describe a map $f:\Sigma \rightarrow {\Bbb M}^{(d-1)+1}$ and
the latter as its superpartner.

Consider the supermanifold $\widehat{\Sigma}$ that have the same topology as $\Sigma$
 but with additional Grassmann coordinates $\theta := (\theta^A)_A =(\theta^1, \theta^2)$
   forming $2$-component Majorana spinor on $\Sigma$.
Then, after adding auxiliary (nondynamical) fields $B^{\mu}(\sigma)$ to the world-sheet,
these fields on $\Sigma$ can be grouped to superfields:(Cf.\ [G-S-W: Sec.\ 4.1.2; Eq.\ (4.1.16)].)
 $$
	 Y^{\mu}(\sigma) \;
		=\; y^{\mu}(\sigma)\,+\,\bar{\theta}\psi^{\mu}(\sigma)\,
		        +\, \frac{1}{2}\,\bar{\theta}\theta\,B^{\mu}(\sigma)\,.		
 $$	
  
 {From} the viewpoint of Grothendieck's Algebraic Geometry,
  a map $\widehat{f} : \widehat{\Sigma}\rightarrow {\Bbb M}^{(d-1)+1}$ is specified
  contravariantly by a homomorphism
   $$
    \begin{array}{ccccc}
      \widehat{f}^{\sharp}  &  :
	    & C^{\infty}({\Bbb M}^{(d-1)+1})
        & \longrightarrow     &   C^{\infty}(\widehat{\Sigma})\\[1.2ex]
	 && y^{\mu}  & \longmapsto   & \widehat{f}^{\sharp}(y^{\mu})
	\end{array}
   $$
   of the function rings in question.
 Since $C^{\infty}(\widehat{\Sigma})=C^{\infty}(\Sigma)[\theta^1,\theta^2]$
     a superpolynomial ring over $C^{\infty}(\Sigma)$,
 $\hat{f}^{\sharp}(y^{\mu})$ must be of the form
  $$
    \widehat{f}^{\sharp}(y^{\mu})\;
	  =\; f^{\mu}(\sigma)\,
	          +\,\bar{\theta}\psi^{\mu}(\sigma)\,
		          +\, \frac{1}{2}\,\bar{\theta}\theta\,B^{\mu}(\sigma)\,,
  $$
  which is exactly the previous quoted expression [G-S-W: Eq.\ (4.1.16)].
 In conclusion,
  \begin{itemize}
   \item[{\Large $\cdot$}]
    {\it A Ramond-Neveu-Schwarz fermionic string moving in a Minkowski space-time ${\Bbb M}^{(d-1)+1}$
	    as studied in {\rm [G-S-W: Chap.\ 4]}
      can be described by
	   a map $\widehat{f}: \widehat{\Sigma} \rightarrow {\Bbb M}^{(d-1)+1}$
	    in the sense of Grothendieck's Algebraic Geometry.}
  \end{itemize}
   
\bigskip

\noindent
$(b)$ {\it Green-Schwarz} ({\it GS}) {\it fermionic string}

\medskip

\noindent
In this setting,
 in addition to the ordinary bosonic (world-sheet scalar) fields
   $y^{\mu}(\sigma)$, $\mu=0,1,\,\cdots\,, d-1$, on $\Sigma$
   that collectively describe a map $f: \Sigma\rightarrow {\Bbb M}^{(d-1)+1}$,
 there are also a set of {\it world-sheet scalar yet mutually anticommuting} fields
   $\theta^{Aa}(\sigma)$, $A=1,\,\cdots\,, N$ and $a=1,\,\cdots\,, s$, on $\Sigma$.
Here $s$ is the dimension of a spinor representation of the Lorentz group $\SO(d-1,1)$
 of the target Minkowski space-time ${\Bbb M}^{(d-1)+1}$.

{\it Differential geometrically} intuitively, one would think of these (world-sheet scalar) fields on $\Sigma$
 collectively as follows:
 \begin{itemize}
  \item[{\Large $\cdot$}]
   Let $\widehat{\Bbb M}^{(d-1)+1}$ be a superspace
     with coordinates
	   the original coordinates $y:=(y^{\mu})_{\mu}$ of ${\Bbb M}^{(d-1)+1}$
	  and additional anticommuting coordinates $\theta^{Aa}$,
	     $A=1,\,\cdots\,, N$ and $a=1,\,\cdots\,, s$,
     such that
   	 each tuple $(\theta^{A1}, \,\cdots\,, \theta^{As})$, $A=1,\,\cdots\,,N$,
   	   is in a spinor representation of the Lorentz group $\SO(d-1,1)$,
     	   the symmetry of the space-time ${\Bbb M}^{(d-1)+1}$ .
   Note that $\widehat{\Bbb M}^{(d-1)+1}\simeq {\Bbb R}^{d|Ns}$ as supermanifolds.
	
  \item[{\Large $\cdot$}]
   The collection $(y^{\mu}(\sigma), \theta^{Aa}(\sigma))_{\mu, A, a}$
    of (world-sheet scalar) fields on $\Sigma$ describe collectively
    a map $\widehat{f}: \Sigma \rightarrow \widehat{\Bbb M}^{(d-1)+1}$.
  In other words,
  a Green-Schwarz fermionic string moving in ${\Bbb M}^{(d-1)+1}$  is described
   by a map from an ordinary world-sheet to a super-Minkowski space-time.
 \end{itemize}
 
However, {\it algebraic geometrically} some revision to this naive differential geometric picture has to be made.
 \begin{itemize}
  \item[{\Large $\cdot$}]
   One would like a contravariant  equivalence between spaces and their function ring:
    $$
	 \begin{array}{ccccc}
       \widehat{f}& : & \Sigma  & \longrightarrow & \widehat{\Bbb M}^{(d-1)+1}
	 \end{array}
    $$
   with
   $$
    \begin{array}{cccccl}
      \widehat{f}^{\sharp}& :
	    &  C^{\infty}({\Bbb M}^{(d-1)+1})[\theta^{Aa}\,:\,  1\le A \le N,\, 1\le a\le s]
		& \longrightarrow &  C^{\infty}(\Sigma) \\[1.2ex]
	  &&	y^{\mu}       & \longmapsto     &   y^{\mu}(\sigma) \\[1.2ex]
	  &&    \theta^{Aa}& \longmapsto     &  ?                                         &.
	 \end{array}
   $$
  Here,
    $C^{\infty}({\Bbb M}^{(d-1)+1})[\theta^{Aa}\,:\,  1\le A \le N,\, 1\le a\le s]$
       is the superpolynomial ring over the $C^{\infty}$-ring $C^{\infty}({\Bbb M}^{(d-1)+1})$
	  with anticummuting generators in $\{\theta^{Aa}\}_{A, a}$.
   
  \item[{\Large $\cdot$}]
   The natural candidate for $\widehat{f}(\theta^{Aa})$ is certainly the world-sheet scalar field
     $\theta^{Aa}(\sigma)$ regarded as an element in the function-ring  of $\Sigma$.
   However, the anticommuting nature of fields $\theta^{Aa}$, $1\le A\le N$ and $1\le a\le s$,
    among themselves forbids them to lie in $C^{\infty}(\Sigma)$.

  \item[{\Large $\cdot$}]	
   The way out of this from the viewpoint of Grothendieck's Algebraic Geometry
    is to extend the world-sheet $\Sigma$ also
    to a superworld-sheet $\widehat{\Sigma}$ with the function ring the superpolynomial ring
    $C^{\infty}(\Sigma)[\theta^{\prime Aa}\,:\, 1\le A\le N,\, 1\le a\le s]$.
	
  \item[{\Large $\cdot$}]	
   One now has a well-defined super-$C^{\infty}$-ring-homomorphism
   $$
    \begin{array}{cccccl}
      \widehat{f}^{\sharp}& :
	   & C^{\infty}({\Bbb M}^{(d-1)+1})[\,\theta^{Aa}\,:\,  A,\, a\,]
	   & \longrightarrow
	   &  C^{\infty}(\Sigma)[\,\theta^{\prime Aa}\,:\, A,\, a\,]       \\[1.2ex]
	  &&	y^{\mu}       & \longmapsto     &   y^{\mu}(\sigma) \\[1.2ex]
	  &&    \theta^{Aa}& \longmapsto     &   \theta^{Aa}(\sigma)                                         &.
	 \end{array}
   $$
   
  \item[{\Large $\cdot$}]
   Furthermore, since all the fields $\theta^{Aa}(\sigma)$ are dynamical,
   in comparison with the setting for the RNS fermionic string, it is reasonable to require in addition that
   $$
    \widehat{f}^{\sharp}(\theta^{Aa})\; =\; \theta^{Aa}(\sigma)\;
	  \in \; \Span_{C^{\infty}(\Sigma)}\{\,\theta^{\prime Aa}\,|\, A, a\,\}\,.
   $$
 \end{itemize}

In conclusion,
  \begin{itemize}
   \item[{\Large $\cdot$}] {\it
    Assuming the notation from the above discussion.
	A Green-Schwarz fermionic string moving in a Minkowski space-time ${\Bbb M}^{(d-1)+1}$
	    as studied in {\rm [G-S-W: Chap.\ 5]}
      can be described in the sense of Grothendieck's Algebraic Geometry	
	  by a map $\widehat{f}: \widehat{\Sigma} \rightarrow  \widehat{\Bbb M}^{(d-1)+1}$,
	   defined by a super-$C^{\infty}$-ring-homomorphism
      $$
       \begin{array}{cccccl}
         \widehat{f}^{\sharp}& :
	      & C^{\infty}({\Bbb M}^{(d-1)+1})[\,\theta^{Aa}\,:\,  A,\, a\,]
	      & \longrightarrow
	      &  C^{\infty}(\Sigma)[\,\theta^{\prime Aa}\,:\, A,\, a\,]       \\[1.2ex]
	      &&	y^{\mu}       & \longmapsto     &   y^{\mu}(\sigma) \\[1.2ex]
	      &&    \theta^{Aa}& \longmapsto     &   \theta^{Aa}(\sigma)                                         
	   \end{array}
      $$
	 such that
      $$
        \widehat{f}^{\sharp}(\theta^{Aa})\; =\; \theta^{Aa}(\sigma)\;
	      \in \; \Span_{C^{\infty}(\Sigma)}\{\,\theta^{\prime Aa}\,|\, A, a\,\}\,.
      $$	
    }\end{itemize}

\bigskip

\begin{flushleft}
{\bf Fermionic D-branes as dynamical objects \`{a} la RNS or GS fermionic strings}
\end{flushleft}
\begin{terminology}
 {\it $[\,$Azumaya/matrix super-$C^k$-manifold associated to $({\cal S},{\cal E})$$\,]$.} {\rm
 Let
   $X$ be a $C^k$-manifold,
   ${\cal S}$ be a locally free ${\cal O}_X$-module of finite rank, and
   ${\cal S}$ be a locally free ${\cal O}_X^{\,\Bbb C}$-module of finite rank.
 For convenience, introduce the following terminologies:
   \begin{itemize}	
	  \item[{\Large $\cdot$}]
       $\widehat{X}:= (X,\widehat{\cal O}_X:=\bigwedge^{\bullet}{\cal S})$
	     be the {\it supermanifold generated by} ${\cal S}$ on $X$,
	  denote $\,\;\widehat{\cal E}:={\cal E}\otimes_{{\cal O}_X}\widehat{\cal O}_X$,
	
	  \item[{\Large $\cdot$}]
       $X^{\!A\!z} :=
	     (X,{\cal O}_X^{A\!z}:= \Endsheaf_{{\cal O}_X^{\,\Bbb C}}({\cal E}))$
       be the {\it Azumaya/matrix manifold associated to} ${\cal E}$ on $X$,
      $\;\;\widehat{X}^{\!A\!z} :=
         (\widehat{X},  \widehat{\cal O}_X^{A\!z}
                   :=\Endsheaf_{\widehat{\cal O}_X^{\,\Bbb C}}(\widehat{\cal E})
	                \simeq  {\cal O}_X^{A\!z}\otimes_{{\cal O}_X}\widehat{\cal O}_X)$
	     be the {\it Azumaya/matrix supermanifold specified by the pair} $({\cal S},{\cal E})$ on $X$,	
     (and $\widehat{\cal E}$ is the fundamental module of $\widehat{X}^{\!A\!z}$).  		
    \end{itemize}	
}\end{terminology}
	
\medskip

\begin{definition-prototype}
 {\bf [fermionic D-branes \`{a} la a RNS fermionic string].} {\rm
 Let
   \begin{itemize}
    \item[{\Large $\cdot$}]
     $Y$ be a $C^k$-manifold
    (e.g.\ a space-time with a Lorentzian metric,  or a Euclidean space-time from Wick rotation, or
      a Riemannian internal space in a compactification of superstring theory background).
   \end{itemize}
 Then,
   a {\it fermionic D-brane in $Y$ in the style of  a Ramond-Neveu-Schwarz fermionic string}
   consists of the following data:
   $$
     (X,\, {\cal S},\, {\cal E},\,
	     \widehat{\varphi}:
          (\widehat{X}^{\!A\!z},\widehat{\cal E})\rightarrow Y)\,,
   $$
  where
  \begin{itemize}
   \item[{\Large $\cdot$}]
    $X$ is a $C^k$-manifold
	(with a Riemannian or Lorentzian structure, depending on the context),
	
   \item[{\Large $\cdot$}]	
    ${\cal S}$ is a (finite) direct sum of sheaves of spinors on $X$,
	
   \item[{\Large $\cdot$}]	
   ${\cal E}$ is a locally free ${\cal O}_X^{\,\Bbb C}$-module of some finite rank $r$,
   
  \item[{\Large $\cdot$}]
   $(\widehat{X}^{\!A\!z},\widehat{\cal E})$
     is the Azumaya/matrix super-$C^k$-manifold with a fundamental module
	 specified by $({\cal S},{\cal E})$ on $X$ and
   $\,\widehat{\varphi}:
           (\widehat{X}^{\!A\!z},\widehat{\cal E})\rightarrow Y\,$					
     is a $C^k$-map from $(\widehat{X}^{\!A\!z},\widehat{\cal E})$ to $Y$,
	 as defined in Definition~4.2.1.3. 					
  \end{itemize}						
 }\end{definition-prototype}
 
\medskip

\begin{definition-prototype}
 {\bf [fermionic D-branes \`{a} la a GS fermionic string].} {\rm
 Let
  \begin{itemize}
   \item[{\Large $\cdot$}]
    $Y$ be a $C^k$-manifold
   (e.g.\ a space-time with a Lorentzian metric,  or a Euclidean space-time from Wick rotation, or
      a Riemannian internal space in a compactification of superstring theory background),
  
   \item[{\Large $\cdot$}]
    ${\cal S}_Y$ be a (finite) direct sum  of sheaves of spinors on $Y$,

   \item[{\Large $\cdot$}]
    $\widehat{Y}:= (Y, \widehat{\cal O}_Y:=\bigwedge^{\bullet}{\cal S}_Y )$
     be the super-$C^k$-manifold generated by ${\cal S}_Y$ on $Y$. 		
  \end{itemize}	
 Then,
   a {\it fermionic D-brane in $Y$ in the style of  a Green-Schwarz fermionic string}
   consists of the following data:
   $$
     (X,\, {\cal S},\, {\cal E},\,
	     \widehat{\varphi}:
          (\widehat{X}^{\!A\!z},\widehat{\cal E})\rightarrow \widehat{Y})\,,
   $$
  where
  \begin{itemize}
   \item[{\Large $\cdot$}]
    $X$ is a $C^k$-manifold
	(with a Riemannian or Lorentzian structure, depending on the context),
	
   \item[{\Large $\cdot$}]	
    ${\cal S}$ is a locally free ${\cal O}_X$-module
	    of the same rank as that of the ${\cal O}_Y$-module ${\cal S}_Y$,
	
   \item[{\Large $\cdot$}]	
   ${\cal E}$ is a locally free ${\cal O}_X^{\,\Bbb C}$-module of some finite rank $r$,
   
  \item[{\Large $\cdot$}]
   $(\widehat{X}^{\!A\!z},\widehat{\cal E})$
     is the Azumaya/matrix super-$C^k$-manifold with a fundamental module
	 specified by $({\cal S},{\cal E})$ on $X$ and
   $\,\widehat{\varphi}:
           (\widehat{X}^{\!A\!z},\widehat{\cal E})\rightarrow \widehat{Y}\,$					   
     is a $C^k$-map from $(\widehat{X}^{\!A\!z},\widehat{\cal E})$ to $\widehat{Y}$,
	    as defined in Definition~4.3.2,  					
    such that
    $$
      \varphi^{\sharp}({\cal S}_Y)\;
	     \subset \; {\cal O}_X^{A\!z}\otimes_{{\cal O}_X}{\cal S}\,.
    $$															
  \end{itemize}						
}\end{definition-prototype}

\medskip

\begin{remark} $[\,$action functional for fermionic D-brane$\,]$. {\rm
 The supersymmetric action functional of either type of fermionic D-branes in the above prototypical definitions
  remains to be worked out and understood.
 Cf.\ Sec.~6.
}\end{remark}

\bigskip

\subsection{The Higgs mechanism on D-branes vs.\ deformations of maps from a matrix brane}

Throughout the series of works in this project,
we have in a few occasions brought out the term `Higgsing/un-Higgsing' of D-branes
 in the study of deformations of maps/morphisms from an Azumaya/matrix scheme or manifold
 with a fundamental module;
 cf.\ [L-Y1: Sec.~2.2] and [L-Y2: Example~2.3.2.11].
In particlular, recall {\sc Figure} 5-2-0-1.
 %
 %
 \begin{figure} [htbp]
  \bigskip
  \centering
  \includegraphics[width=0.80\textwidth]{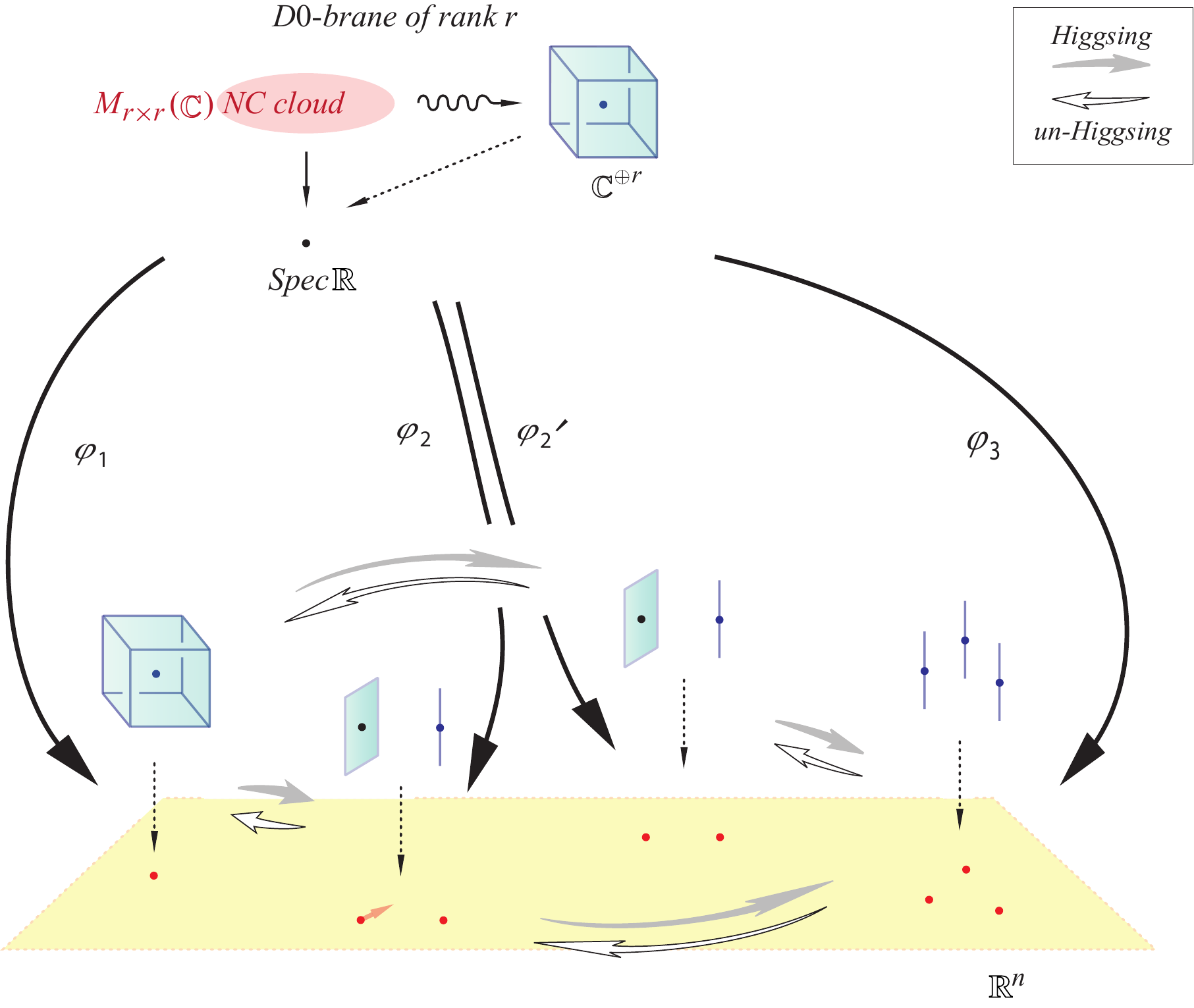}
 
  \bigskip
  \bigskip
  \centerline{\parbox{13cm}{\small\baselineskip 12pt
   {\sc Figure}~5-2-0-1. (Cf.\ [L-Y2: {\sc Figure}~2-1-1] (D(6)).)   
    Readers are referred to [L-Y3: {\sc Figure} 3-4-1, caption] (D(11.1)) for explanations.
       }}
  \bigskip
 \end{figure}	
A closer look at the link of the two can now be made.

\bigskip

\subsubsection{The Higgs mechanism in the Glashow-Weinberg-Salam model}

{To} manifest the parallel setting in our situation,
we highlight in this subsubsection the relevant classical part of
 the Higgs mechanism in the Glashow-Weinberg-Salam model for leptons
  (here, electron and muon and their corresponding neutrinos)
  that breaks the gauge symmetry from $\SU(2)\times U(1)_Y $
  (the gauge symmetry for the electroweak interaction)
   to $U(1)_{em}$ (the gauge symmetry for the electromagnetic interaction).
Here, both $U(1)_Y$ and $U(1)_{em}$ are isomorphic to the $U(1)$ group;
 the different labels $Y$ and {\it em}  indicates that $U(1)_{em}$ is a subgroup
  of $\SU(2)\times U(1)_Y$ that is different from the factor $U(1)_Y$ in the product.
 {\it em} here stands for `{\it e}lectro{\it m}agnetic'.
Mathematicians are referred to
  [P-S: Sec.~20.2] and also
  [I-Z: Sec.~12.6], [Mo: Chap.\ 3], [Ry: Sec.~8.5]
  for a complete discussion that takes care also of quantum-field-theoretical issues
  such as renormalizability of and anomalies in the theory.

\clearpage

\begin{flushleft}
{\bf The Lagrangian density of the Glashow-Weinberg-Salam model}
\end{flushleft}
This is a $4$-dimensional quantum field theory on the Minkowski space-time ${\Bbb M}^{3+1}$
 (with coordinates $x:=(x^{\mu})_{\mu=0,1,2,3}$),
whose Lagrangian density is given by
 $$
  {\cal L}\;  =\;
     {\cal L}_{\fermionscriptsize}\,+\,{\cal L}_{\gaugescriptsize}\,
	   +\, {\cal L}_{\Higgsscriptsize}\, +\, {\cal L}_{\Yukawascriptsize}\,,
 $$
 where
 \begin{itemize}
  \item[{\Large $\cdot$}]
   ${\cal L}_{\fermionscriptsize}\;=\;
      \bar{E}_L(i\,/\!\!\!\!D)E_L+\bar{e}_R(i\,/\!\!\!\!D)e_R
	  + \bar{Q}_L(i\,/\!\!\!\!D)Q_L+\bar{u}_R(i\,/\!\!\!\!D)u_R
	  + \bar{d}_R(i\,/\!\!\!\!D)d_R$,
 
  \item[{\Large $\cdot$}]
   ${\cal L}_{\gaugescriptsize}\;
     =\;  -\,\frac{1}{4} F_{\mu\nu}F^{\mu\nu}$,
   
  \item[{\Large $\cdot$}]
   ${\cal L}_{\Higgsscriptsize}\;
      =\;  (D_{\mu}\phi)^{\dagger}(D^{\mu}\phi)
	            +  \mu^2 \phi^{\dagger}\phi  -  \lambda (\phi^{\dagger}\phi)^2$,

  \item[{\Large $\cdot$}]
  ${\cal L}_{\Yukawascriptsize}\;
    =\; -\lambda_e\bar{E}_L\cdot\phi e_R
	     - \lambda_d\bar{Q}_L\cdot\phi d_R
		 - \lambda_u\epsilon^{ab}\bar{Q}_{La}\phi^{\dagger}_bu_R
		 + \, \mbox{(hermition conjugates)}$\\[.6ex]
   consists of the cubic interaction terms
   that are linear in Higgs-field components and quadratic in other matter-field components.
 \end{itemize}
The part in this expression that is most relevant to us is explained/reviewed below:
(Assuming up-down-repeated-dummy-index summation convention.)
 \begin{itemize}
  \item[(1)]
   The Lie algebra $\su(2)$ takes the basis $\{\tau_1, \tau_2, \tau_3\}$ from Pauli matrices
     $\sigma_1,\, \sigma_2,\, \sigma_3$   with
    $$
	  \tau_1\:=\:\frac{1}{2}\,\sigma_1\:
	   =\: \frac{1}{2}\left[ \begin{array}{cc}0  & 1  \\[.6ex] 1 &  0  \end{array} \right],	
       \hspace{1em}	
	  \tau_2\: =\: \frac{1}{2}\,\sigma_2\:
	   =\: \frac{1}{2}\left[ \begin{array}{cr} 0  & -i   \\[.6ex]  i &  0  \end{array} \right],	  
	   \hspace{1em}
	  \tau_3\: =\: \frac{1}{2}\,\sigma_3\:
	   =\: \frac{1}{2}\left[ \begin{array}{cr} 1  & 0  \\[.6ex]  0 & -1  \end{array} \right];
	$$
    $A = A_{\mu}^a(x)\tau_adx^{\mu}$ is the $\su(2)$ gauge field
     in the adjoint representation of $\su(2)$,
   $B = B_{\mu}(x)dx^{\mu}$ is the $U(1)_Y$ gauge field,
   $F=F_{\mu\nu}(x)dx^{\mu}\wedge dx^{\nu}$ the field strength (i.e.\ curvature)
     of the $\su(2)\oplus u(1)_Y$-valued $1$-form field $A \oplus B$.
	
  \item[(2)]	
   The $\su(2)\oplus u(1)_Y$ representation-theoretical contents of all the fields in the model
   are given by the following table:
 
  \vspace{1.2ex}
  \item[]
   \centerline{
   \begin{tabular}{lcccl}
    field in $\su(2)$-rep
	      && \hspace{-1.2ex}$\su(2)\oplus u(1)_Y$    \\[-.6ex]
    components 		&&  representation
          && \raisebox{-1.2ex}{\rule{0ex}{1.2em}}\hspace{1em}name \\  \hline\hline
	 \hspace{1em} gauge fields    \raisebox{-1.2ex}{\rule{0ex}{1.8em}}\\ \hline
      $A_{\mu}\; =\; \left(\!\!\mbox{\scriptsize
	        $\begin{array}{c} A_{\mu}^1 \\[1.2ex]
			                                           A_{\mu}^2 \\[1.2ex]
													   A_{\mu}^3    \end{array}$}\!\!\right)$ 	  
	   && $(3, 0)$
	   && {\footnotesize \hspace{-2.4ex}
	             $\begin{array}{l}\mbox{gauge boson}  \\[.2ex]
  				                                         \mbox{gauge boson}  \\[.2ex]
														 \mbox{gauge boson} \end{array}$}
				 \raisebox{-2.4ex}{\rule{0ex}{3.6em}} \\[.6ex] 		
     $B_{\mu}$	 && (1,0)  && {\footnotesize gauge boson}			   \\[.6ex]  \hline\hline				 
     \hspace{1em} matter fields    \raisebox{-1.2ex}{\rule{0ex}{1.8em}}\\ \hline
     $E_L\; =\; \left(\!\!\mbox{\scriptsize
	                            $\begin{array}{c} \nu_{eL} \\ e_L \end{array}$}\!\!\right)$ 	  
	   && $(2, -\frac{1}{2})$
	   && {\footnotesize \hspace{-2.4ex}
	             $\begin{array}{l}\mbox{neutrino (left-handed)}  \\[.2ex]
  				                                         \mbox{electron (left-handed)}\end{array}$}
				 \raisebox{-2.4ex}{\rule{0ex}{2.6em}} \\[.6ex]
     $Q_L\; =\; \left(\!\!\mbox{\scriptsize
	                            $\begin{array}{c} u_L \\ d_L  \end{array}$}\!\!\right)$
       && $(2, \frac{1}{6})$
	   &&{\footnotesize \hspace{-2.4ex}
	             $\begin{array}{l}\mbox{quark (left-handed)}\\[.2ex]
				                                         \mbox{quark (left-handed)}\end{array}$}
				 \raisebox{-2.4ex}{\rule{0ex}{2.6em}}						\\[.6ex]
     $e_R$       && $(1, -1)$      && {\footnotesize electron (right-handed)} \\[.6ex]
     $u_R$       && $(1, \frac{2}{3})$      && {\footnotesize quark (right-handed)}  \\[.6ex]
     $d_R$       && $(1, -\frac{1}{3})$
                          && {\footnotesize quark (right-handed)}
 						          \raisebox{-1.6ex}{\rule{0ex}{1.2em}} \\ \hline\hline
     \hspace{1em} Higgs fields    &&  \raisebox{-1.2ex}{\rule{0ex}{1.8em}}   \\ \hline
     $\phi\; =\; \left(\!\!\mbox{\scriptsize
	                       $\begin{array}{c}\phi^+ \\  \phi^0\end{array}$}\!\!\right)$		
       && $(2,  \frac{1}{2})$
	   && {\footnotesize  Higgs field}   \raisebox{-2.4ex}{\rule{0ex}{2.6em}}\\  \hline
   \end{tabular}} 
  
  \vspace{1.2ex}
  \item[]
   In the above table,  matter fields
      $\nu_{eL},\,  e^-_L,\,  u_L,\,  d_L$ (resp.\  $e_R,\, u_R,\, d_R$)
    are described by left-handed (resp.\ right-handed)  Weyl spinor fields on ${\Bbb M}^{3+1}$.
 
  \item[{\Large $\cdot$}]
   For a field $\psi$ belonging to a representation $T$ of $\su(2)$ with $u(1)_Y$ charge $Y$,
   the covariant derivative $D_{\mu}\psi$ of $\psi$ is given by
   $$
     D_{\mu}\psi \; =\;  (\partial _{\mu} -igA_{\mu}^aT_a -ig^{\prime}YB_{\mu})\,\psi\,.
   $$
   For example, for the complex $2$-component  Higgs field $\phi$,
   $$
     D_{\mu}\phi\;
     =\;  \left(  \partial_{\mu}-igA^a_{\mu}(x)\tau_a
	                   - \frac{i}{2}g^{\prime}B_{\mu}(x)\right) \phi\,.
   $$
  Similarly, for the covariant derivative of $E_L$, $Q_L$, $e_R$, $u_R$, $d_R$.
    
  \item[{\Large $\cdot$}]
   $g$, $g^{\prime}$, $\mu^2$, $\lambda$, and $\lambda_e$, $\lambda_d$, $\lambda_u$
   in ${\cal L}$ are real-valued coupling constants of the model.
	
  \item[{\Large $\cdot$}]	
   ${\cal L}_{\fermionscriptsize}$ as given is  the standard gauge-invariant Lagrangian density
   without a potential for massless fermions.
 \end{itemize}
Formally and classically, one has
  massless fermions ($\nu_{eL}$, $e_L$, $u_L$, $d_L$, $e_R$, $u_R$, $d_R$),
  massless gauge bosons ($A_{\mu}^a$, $B_{\mu}$),
  and tachyonic Higgs bosons ($\phi^+$, $\phi^0$) to begin with.
 
\bigskip

\begin{ssremark}$[\,$quark and strong interaction$\,]$. {\rm
 The above model contains quarks $u_L$, $u_R$, $d_L$, $d_R$.
 However, only their involvement with the electroweak interaction is considered.
 The Standard Model, which takes into account also the strong interaction,
  enlarges the gauge group from the current
  $\SU(2)\times U(1)_Y$ to $\SU(2)_L\times U(1)_Y\times \SU(3)_c$.
 On the other hand,
   as long as the purpose of comparison to Higgs mechanism on D-branes in our setting is concerned,
 one can remove all the quarks in the model and consider only the truncated theory
  $$
   \begin{array}{ccl}
   {\cal L}^{\prime}   & =
      &
     {\cal L}_{\fermionscriptsize}^{\prime}\,+\,{\cal L}_{\gaugescriptsize}\,
	   +\, {\cal L}_{\Higgsscriptsize}\, +\, {\cal L}_{\Yukawascriptsize}^{\prime} \\[1.2ex]
    &= &  \bar{E}_L(i\,/\!\!\!\!D)E_L+\bar{e}_R(i\,/\!\!\!\!D)e_R\;
		         -\,\frac{1}{4} F_{\mu\nu}F^{\mu\nu}\\[1.2ex]
    &&		       +  (D_{\mu}\phi)^{\dagger}(D^{\mu}\phi)\,
                              +  \mu^2 \phi^{\dagger}\phi   - \lambda (\phi^{\dagger}\phi)^2\;
			    -\lambda_e\bar{E}_L\cdot\phi e_R\;
	          + \, \mbox{(hermition conjugates)}\,.			
   \end{array}
  $$
}\end{ssremark}

\bigskip

\begin{flushleft}
{\bf  Spontaneous gauge-symmetry breaking from a spontaneous settling-down to\\ a vacuum of the Higgs field}
\end{flushleft}
The vacuum manifold for the Higgs field $\phi$ is given by the locus in the $\phi$-space $\simeq {\Bbb C}^2$
 on which the potential $V(\phi)$ takes its minimum:
 $$
   \{(z^1,z^2)\,|\,    |z^1|^2+|z^2|^2\,=\,  \mu^2/(2\lambda)   \}\;\subset \; {\Bbb C}^2\,.
 $$
Let
 $$
   v\;=\;  \sqrt{\mu^2/\lambda}\; >\; 0\,.
 $$
Then, up to a gauge transformation, we may assume that the Higgs field ``condenses" to
(i.e.\ takes its value only at) the {\it vacuum expectation value} ({\it VEV}) $(0,v/\sqrt{2})$.
This reduces the gauge group from the original $\SU(2)\times U(1)_Y$
 to the subgroup $\Stab((0,v/\sqrt{2}))$ (the {\it stabilizer}
      or {\it isotropy group} of $(0,v/\sqrt{2})$)
  consisting of all elements of $\SU(2)\times U(1)_Y$ that leave $(0,v/\sqrt{2})$ fixed.
Or equivalently in terms of Lie algebras, it reduces $\su(2)\oplus u(1)$
 to the sub-Lie algebra $\Ann((0,v/\sqrt{2}))$ (the {\it annihilator} of $(0,v/\sqrt{2})$)
  consisting of all elements in $\su(2)\oplus u(1)_Y$ that annihilate $(0,v/\sqrt{2})$:
  $$
     u(1)_{em}\;
	   :=\;  \{\, (c^1\tau_1 + c^2\tau_2 +c^3\tau_3, c^4)\,|\,   c^1=c^2=0,\, c^3=c^4\,\}\;
	   \subset\;  \su(2)\oplus u(1)_Y\,.
  $$
  
The original model ${\cal L}$ descends to a new quantum field theory
 on fields that fluctuate around the Higgs vacuum $(0,v/\sqrt{2})$
  (a procedure that influences only the Higgs field in the current model)
 and in terms of the induced representations of the residual $u(1)_{em}$
   (a procedure that reorganizes and may regroup and influence all the fields in the model):
 \begin{itemize}
  \item[(1)]
   A general Higgs field $\phi$ around the VEV
    $\frac{1}{\sqrt{2}}${\scriptsize $\left(\!\!\!\begin{array}{c} 0\\ v \end{array}\!\!\!\right)$}
    can be expressed as 	
	$$
	  \phi(x)\; =\; U(x)\frac{1}{\sqrt{2}}\left(\!\!\!\begin{array}{c} 0\\ v + h(x)\end{array}\!\!\!\right)
	$$
	under a gauge transformation $U(x)$ with value in $\SU(2)\times U(1)_Y$,
  where $h(x)$ is a real-valued scalar field 	on ${\Bbb M}^{3+1}$ with $\langle h(x)\rangle =0$.
  Thus, the only physical degree of freedom from $\phi$ after gauge-symmetry breaking  is $h$.
  
  \item[(2)]	
   The $u(1)_{em}$ representation-theoretical contents of all the fields around vacuum
    in the model after the symmetry breaking are given by the following table.
   Each is specified by its $u(1)_{em}$-charge.
 
  \vspace{1.2ex}
  \item[]
   \centerline{
   \begin{tabular}{lcrcl}
    field in $u(1)_{em}$-rep
	      && \hspace{-1.2ex}$u(1)_{em}$-charge
		  && \raisebox{-1.2ex}{\rule{0ex}{1.2em}}\hspace{1em}name \\  \hline\hline
	 \hspace{1em} ($u(1)_{em}$) gauge field    \raisebox{-1.2ex}{\rule{0ex}{1.8em}}\\ \hline      
    $A^{em}_{\mu}\;
	  :=\; \frac{1}{\sqrt{g^2+g^{\prime 2}}}(g^{\prime}A_{\mu}^3+gB_{\mu})$	
	   &&    $0$ \hspace{2.2em}
	   &&   {\footnotesize gauge boson} \raisebox{-2.8ex}{\rule{0ex}{2.6em}}    \\ \hline\hline	
	 \hspace{1em} massive vector fields    \raisebox{-1.2ex}{\rule{0ex}{1.8em}}\\ \hline      	   
	$W_{\mu}^+\; :=\; \frac{1}{\sqrt{2}}(A_{\mu}^1 - iA_{\mu}^2)$
	   &&   $1$ \hspace{2.2em}  && {\footnotesize vector boson}         \rule{0ex}{1.6em}  \\[1.8ex]
	$W_{\mu}^-\; :=\; \frac{1}{\sqrt{2}} (A_{\mu}^1 + i A_{\mu}^2)$
       && $-1$ \hspace{2.2em}        &&  {\footnotesize vector boson}         \\[1.8ex]
    $Z_{\mu}^0\;
	   :=\; \frac{1}{\sqrt{g^2+g^{\prime 2}}}(gA_{\mu}^3 - g^{\prime}B_{\mu})$	
       &&  $0$ \hspace{2.2em}       && {\footnotesize vector boson}	
	                          \raisebox{-3ex}{\rule{0ex}{1.8em}}	   \\[.6ex]  \hline\hline				 
     \hspace{1em} matter fields    \raisebox{-1.2ex}{\rule{0ex}{1.8em}}\\ \hline
     $E_L\; =\; \left(\!\!\mbox{\scriptsize
	                            $\begin{array}{c} \nu_{eL} \\ e_L \end{array}$}\!\!\right)$ 	  
	   &&  {\footnotesize $\begin{array}{r}  0 \\  -1  \end{array}$} \hspace{1.8em}
	   && {\footnotesize \hspace{-2.4ex}
	             $\begin{array}{l}\mbox{neutrino (left-handed)}  \\[.2ex]
  				                                         \mbox{electron (left-handed)}\end{array}$}
				 \raisebox{-2.4ex}{\rule{0ex}{2.6em}} \\[.6ex]
     $Q_L\; =\; \left(\!\!\mbox{\scriptsize
	                            $\begin{array}{c} u_L \\ d_L  \end{array}$}\!\!\right)$
       &&  {\footnotesize $\begin{array}{r}  2/3 \\[.2ex]  - 1/3  \end{array}$} \hspace{1.8em}
	   &&{\footnotesize \hspace{-2.4ex}
	             $\begin{array}{l}\mbox{quark (left-handed)}\\[.2ex]
				                                         \mbox{quark (left-handed)}\end{array}$}
				 \raisebox{-2.4ex}{\rule{0ex}{2.6em}}						\\[.6ex]
     $e_R$       &&  $-1$  \hspace{2.2em}    && {\footnotesize electron (right-handed)} \\[.6ex]
     $u_R$       && $2/3$ \hspace{2.2em}
	                      && {\footnotesize quark (right-handed)}  \\[.6ex]
     $d_R$       && $-1/3$ \hspace{2.2em}
                          && {\footnotesize quark (right-handed)}
 						          \raisebox{-1.6ex}{\rule{0ex}{1.2em}} \\ \hline\hline
     \hspace{1em} Higgs field    &&  \raisebox{-1.2ex}{\rule{0ex}{1.8em}}   \\ \hline
     $h$             && 0  \hspace{2.2em}
	   && {\footnotesize  Higgs field}   \raisebox{-2.4ex}{\rule{0ex}{2.6em}}\\  \hline
   \end{tabular}} 
  
  \vspace{1.2ex}
  \item[]
   Note in particular that
   the $u(1)_{em}$-charge of both Weyl spinors $e_L$ and $e_R$ on ${\Bbb M}^{3+1}$ is $-1$,
   reaffirming that they correspond to electrons of different chiralities moving in the Minkowski space-time.
   Note also that all quarks have fractional $u(1)_{em}$-charge.
 \end{itemize}
This gives the full field contents of the model after the gauge-symmetry breaking
  in terms of representations of the $U(1)_{em}$ symmetry left.

\bigskip

\begin{flushleft}
{\bf Mass generation to fermions and gauge bosons after the symmetry breaking}
\end{flushleft}
Rewrite the Lagrangian density ${\cal L}$ now in terms of fields after the gauge-symmetry breaking.
Then, up to an overall constant,
$$
 \begin{array}{lcl}
 {\cal L}_{\Higgsscriptsize}
     & =   & \frac{1}{2}(\partial_{\mu}h)^2   - \mu^2h^2
	                 +\frac{g^2v^2}{4}W_{\mu}^-W_{\mu}^+
                     + \frac{(g^2+g^{\prime 2})v^2}{4}(Z_{\mu}^0)^2  \\[1.2ex]
  && \hspace{6em}
      	+\,  (\,\mbox{higher order terms in $h$, $W_{\mu}^-$, $W_{\mu}^+$, $Z_{\mu}^0$}\,)					 
	             \\[2.4ex]
 {\cal L}_{\Yukawascriptsize}
    & =    &  - \frac{\lambda_e v}{\sqrt{2}}\bar{E}_L\cdot e_R
	               - \frac{\lambda_d v}{\sqrt{2}}\bar{Q}_L\cdot d_R
		           - \frac{\lambda_u v}{\sqrt{2}}\bar{u}_L u_R                   \\[1.2ex]
   &&	\hspace{6em}       + \, (\,\mbox{higher order terms and hermition conjugates}\,)\,.	
 \end{array}				
$$
Note that the kinetic terms
  for $h$ in ${\cal L}_{\Higgsscriptsize}$,
        $W_{\mu}^-$, $W_{\mu}^+$, $Z_{\mu}^0$ in ${\cal L}_{\gaugescriptsize}$,
        $e_L$, $e_R$,  $d_L$, $d_R$, $u_L$, $u_R$ in ${\cal L}_{\fermionscriptsize}$
  are all in the standard form.
It follows that classically the mass of particles associated to
  $h$, $W_{\mu}^-$, $W_{\mu}^+$, $Z_{\mu}^0$
 (resp.\ $e_L$, $e_R$, $d_L$, $d_R$, $u_L$, $u_R$)
 are read off from ${\cal L}_{\Higgsscriptsize}$ (resp.\ ${\cal L}_{\Yukawascriptsize}$) as
 $$
   m_h\; =\;  \sqrt{2}\mu\,, \hspace{2em}
   m_{W_{\mu}^-}\; =\;  m_{W_{\mu}^+}\; =\;  \frac{gv}{2}\,, \hspace{2em}
   m_{Z_{\mu}^0}\;=\; \frac{\sqrt{g^2+g^{\prime 2}}v}{2}\,,
 $$
 $$
   m_{e_L}\;=\; m_{e_R}\;=\; \frac{\lambda_e v}{\sqrt{2}}\,,  \hspace{2em}
   m_{d_L}\; =\; m_{d_R}\; =\; \frac{\lambda_d v}{\sqrt{2}}\,,  \hspace{2em}
   m_{u_L}\;=\; m_{u_R}\; =\; \frac{\lambda_u v}{\sqrt{2}}  \,.
 $$

{To} recap in words,
 \begin{itemize}
  \item[(1)]
   Before gauge-symmetry breaking,
     ${\cal L}_{\Higgsscriptsize}$ is the only part in ${\cal L}$ that contains
      quartic terms that are quadratic in gauge fields and quadratic in Higgs field.
   After gauge-symmetry breaking, such terms create a mass term for the gauge fields that correspond to broken
    gauge symmetry, rendering them massive vector bosons in the symmetry-broken theory.   	

  \item[(2)]
   Before gauge-symmetry breaking,
     ${\cal L}_{\Yukawascriptsize}$ is the only part in ${\cal L}$ that contains
      cubic terms that are quadratic in fermionic fields and linear in Higgs fields.
   After gauge-symmetry breaking, such terms create a mass term for the fermionic fields,
    rendering them massive fermions in the symmetry-broken theory. 	
 \end{itemize}

\bigskip

\begin{ssremark}
 $[\,$mathematical reflection$\,:$ principal bundle, representation,  associated bundle, reduction of gauge group,
          induced representation$\,]$.
 {\rm
 In mathematical terms,
 let
  \begin{itemize}
    \item[{\Large $\cdot$}]
    $P_{\Lorentzscriptsize}$ be the principle Lorentz-frame bundle
	(with group $\SO(3,1)$   and
    	trivialized by the flat Levi-Civita connection associated to the space-time metric)
	over ${\Bbb M}^{3+1}$,
   
    \item[{\Large $\cdot$}]
	 $P_G$
     be a principle $G$-bundle (trivial in the above case) with group $G=\SU(2)\times U(1)_Y$.
 \end{itemize}
 The various tensor products of irreducible representations $V_{\rho^L}$ of $\SO(3,1)$
  and irreducible representations $V_{\rho^G}$ of $G$
  give rise to various associated vector bundles $E_{\rho^L\otimes\rho^G}$
  of  the $\SO(3,1)\times G$-principle bundle
  $P_{\Lorentzscriptsize}\times_{{\Bbb M}^{3+1}}P_G $
 whose sections corresponds to various fields $\psi_{\rho^L\otimes \rho^G}$
   on the space-time ${\Bbb M}^{3+1}$.
 $\rho_L$ determines the spin (e.g.\ bosons vs.\ fermions) of $\psi_{\rho^L\otimes \rho^G}$
   while $\rho^G$ distinguishes other particle features (e.g.\ electrons vs.\ quarks).
 A choice of a collection of such representations and
   a choice of a gauge invariant Lagrangian density ${\cal L}$ for the fields corresponding to these representations
  together give a model of particle physics.

 When some of the fields take their VEV,
  the gauge group is reduced to the stabilizer subgroup $H\subset G$   of the VEV and
  a principal subbundle $P_H\subset P_G$
   (and hence $P_{\Lorentzscriptsize}\times_{{\Bbb M}^{3+1}}P_H
                                 \subset P_{\Lorentzscriptsize}\times_{{\Bbb M}^{3+1}}P_G$)
  is selected.
 Original fields expanding around VEV assume naturally induced representations from $G$ to its subgroup $H$.
 Re-writing ${\cal L}$ in terms of fields corresponding to these induced representations
  gives the classical picture of the Higgs mechanism.
  
 Mathematicians should be aware that this is the easy part.
 It is the contents at the quantum level that take works.
 And for the Glashow-Weinberg-Salam model, such quantum-field-theoretical conclusions
  are experimentally justified, up to possibly higher order corrections.
}\end{ssremark}

\bigskip

\subsubsection{The Higgs mechanism on the matrix brane world-volume}

While the full detail of the Higgs mechanism that fits our setting can be produced
 only after the action functional for differentiable maps from a matrix brane with fields is introduced,
 the basic structure that initiates the mechanism and the associated gauge-symmetry-breaking pattern
  are readily there in the setting.
Indeed, the following two steps
 \begin{itemize}
  \item[(1)]
   Recall Remark~5.2.1.2
     and consider the associated Lie-algebra bundle
    $E_{\frak G}$ of $P_G$ from the Adjoint representation $G$ on its Lie algebra ${\frak G}$.
   Note that all $G$-modules are naturally ${\frak G}$-modules from the induced endomorphisms.	
	
 \item[(2)]	
  Promote the Lie algebra ${\frak G}$ to a (unital) associative algebra ${\cal A}$,
    the Lie-algebra bundle $E_{\frak G}$ to an associative-algebra bundle $E_{\cal A}$ and
    all ${\frak G}$-modules to ${\cal A}$-modules;
  and  do the same reduction procedure as in Sec.~5.2.1,
    with the sub-Lie algebra of ${\frak G}$ specified by a VEV
	replaced by an appropriate subalgebra of ${\cal A}$.   	
 \end{itemize}
 give the essential formal classical picture of the Higgs mechanism on D-branes in our setting.
(And this is why we call it Higgs mechanism in our setting after all.)
We now explain the details,
 assuming that
   an unspecified action functional for D-branes in our setting with fields thereupon is given
 (cf.\ Sec.~6).

\bigskip

\begin{flushleft}
{\bf Candidates for the Higgs fields}
\end{flushleft}
There are two classes of fields on a matrix brane that can serve as the Higgs fields:
 \begin{itemize}
   \item[(1)]
    {\it Differentiable maps from the matrix brane to the target space(-time)}.\\[1.2ex]
	This can be traced back to the origin of this project ([L-Y1] (D(1))).
	Coincident D-branes in the space-time exhibit, in addition to enhanced gauge symmetries,
	 an enhanced matrix-valued scalar field on the D-brane world-volume that describes the deformations
	 of the D-branes collectively. Cf. [Po2: vol.\ I, Sec.\ 8.7] of Joseph Polchinski.
	
   \item[(2)]
    {\it Differentiable sections on the fundamental bundle $E$ or its dual $E^{\vee}$}. \\[1.2ex]
	Such fields can occur when a dynamical D-brane is immersed
	  in a non-dynamical background D-brane in the space-time,
	  in the same spirit as [Do-G: Sec.~5] of Michael Douglas and Gregory Moore.
    Cf.\ {\sc Figure}~1-2 and its caption. 	  	
 \end{itemize}

\medskip

\begin{ssremark} $[\,$soft lower-dimensional D-branes in a hard higher-dimensional D-brane$\,]$. {\rm
 Recall that the tension $T_p$ of a D$p$-brane in a target-space-time of a superstring is given by
  $$
    T_p\;  =\;   2\pi\,\left(\frac{T}{2\pi}\right)^{\frac{p+1}{2}}\,,
  $$
 where $T$ is the tension of the superstring.
 Thus, in the regime of the superstring theory where $T\gg 2\pi$,
  higher-dimensional D-branes would have much larger tension than the lower-dimensional ones and
 it makes sense to consider a soft/dynamical D-branes immersed in a hard/background D-branes.
}\end{ssremark}
 
\medskip

\begin{ssremark} $[\,$comparison with Glashow-Weinberg-Salam model$\,]$. {\rm
 If the Glashow-Weinberg-Salam model is realized as a gauge theory
   on the world-volume of coincident D3-branes,
 then the Higgs field $\phi$ in the model would correspond to a section of the Chan-Paton bundle $E$.
}\end{ssremark}

\bigskip

\begin{flushleft}
{\bf Reduction induced by a VEV of Higgs field: Transverse fluctuations and the new field contents}
\end{flushleft}
The notion of
  `{\it reduction}' induced by a VEV of Higgs field:
  \begin{itemize}
   \item[{\Large $\cdot$}]
    a `{\it transverse fluctuation}' of a Higgs field around its vacuum expectation value (VEV)  and
    the resulting `{\it new field contents}' under the ``{\it induced symmetry-breaking}'
       from the VEV of the Higgs field
  \end{itemize}	
 are naturally built into the setting  as follows.

\bigskip

\noindent
{\it Case $(a)$}: {\it Differentiable map from the brane world-volume  as Higgs field}

\medskip

\noindent
Let
 $\varphi_0:(X^{\!A\!z},{\cal E})\rightarrow Y$ be a $C^k$-map
 defined by a $C^k$-admissible $\varphi_0^{\sharp}:{\cal O}_Y\rightarrow {\cal O}_X^{A\!z}$,
 and
 $$
    \xymatrix{
    {\cal O}_X^{A\!z}= \Endsheaf_{{\cal O}_X^{\,\Bbb C}}({\cal E})\\
       {\cal A}_{\varphi_0}\;:=\, \rule{0ex}{3ex}
         		{\cal O}_X\langle\Image\varphi_0^{\sharp}\rangle \ar@{^{(}->}[u]
                    				&&& {\cal O}_Y\ar[lll]_-{\varphi_0^{\sharp}}\\
    {\cal O}_X\rule{0ex}{3ex}  \ar@{^{(}->}[u]							
    }
 $$
 be the underlying diagram.
Let ${\cal C}({\cal A}_{\varphi_0})$ be the {\it commutant sheaf} of
  ${\cal A}_{\varphi_0}$ in ${\cal O}_X^{A\!z}$, defined by
 $$
  {\cal C}({\cal A}_{\varphi_0})(U)\;
      :=\; \{a\in {\cal O}_X^{A\!z}(U)\,|\,
	                 \mbox{$[a, a^{\prime}]=0$ for all $a^{\prime}\in {\cal A}_{\varphi_0}(U)$}\}
 $$
 for open sets $U$ of $X$.
Here, $[a,a^{\prime}]:= aa^{\prime}-a^{\prime}a$.
Then,
 ${\cal C}({\cal A}_{\varphi_0})$
  is an ${\cal O}_X^{\,\Bbb C}$-subalgebra of ${\cal O}_X^{A\!z}$,
which contains ${\cal A}_{\varphi_0}$ since ${\cal A}_{\varphi_0}$ is commutative.
    
Consider
 \begin{itemize}
  \item[{\Large $\cdot$}]
    the class of  $C^k$-maps
      $\varphi:(X^{\!A\!z},{\cal E})\rightarrow Y$
        with the constraint that $\Image\varphi^{\sharp}\subset {\cal C}({\cal A}_{\varphi_0})$.
 \end{itemize}
A $C^k$-map  $\varphi$ in this class has the property
  that  ${\cal A}_{\varphi}\subset {\cal C}({\cal A}_{\varphi_0}) $ as well.
It follows that
 such $\varphi^{\sharp}:{\cal O}_Y\rightarrow {\cal O}_X^{A\!z}$
 induces an equivalence class of gluing systems of ring-homomorphisms
 $$
    \underline{\varphi}^{\sharp}\;:\; {\cal O}_Y\;
	   \longrightarrow\;   {\cal C}({\cal A}_{\varphi_0})
 $$
 and, equivalently,  $\varphi$ induces a map
 $$
   \underline{\varphi}\;:\; (X,{\cal C}({\cal A}_{\varphi_0}))\;\longrightarrow\; Y\,.
 $$
By construction,
 $\varphi$ factors through $\underline{\varphi}$ as indicated in the following commutative diagram:
 $$
  \xymatrix{
   &  (X, {\cal O}_X^{A\!z},{\cal E})\hspace{1em}
            \ar@<-3ex>@{->>}[d]      \ar[rr]^-{\varphi}                                                       && Y  \\
   &  (X, {\cal C}({\cal A}_{\varphi_0}), {\cal E} )
                                                \ar[rru]_-{\underline{\varphi}}   && &,
   }
 $$
 where
  \begin{itemize}
   \item[{\Large $\cdot$}]
    $(X,{\cal C}({\cal A}_{\varphi_0}))$
	  is a noncommutative space with the underlying topology $X$
	and the structure sheaf ${\cal C}({\cal A}_{\varphi_0})$;
	
   \item[{\Large $\cdot$}]	
    ${\cal E}$ is now regarded as the ${\cal C}({\cal A}_{\varphi_0})$-module
	via the built-in inclusion
    ${\cal C}({\cal A}_{\varphi_0})\subset {\cal O}_X^{A\!z}$;
	
   \item[{\Large $\cdot$}] 	
    the surjection
	  $(X,{\cal O}_X^{A\!z})\rightarrow (X,{\cal C}({\cal A}_{\varphi_0}))$
	 is defined by the built-in inclusion
	 ${\cal C}({\cal A}_{\varphi_0})\hookrightarrow {\cal O}_X^{A\!z}$.
  \end{itemize}
Note that, also by construction,
 \begin{itemize}
  \item[{\Large $\cdot$}] {\it
   $\varphi$ and $\underline{\varphi}$ have identical surrogates, i.e.,
   $X_{\varphi}=X_{\underline{\varphi}}$,    and
   identical  $C^k$-maps
     $f_{\varphi}:X_{\varphi}\rightarrow Y$ and
	 $f_{\underline{\varphi}}:X_{\underline{\varphi}}\rightarrow Y$
	   under the above identification of surrogates.}
 \end{itemize}
With Sec.~5.2.1 in mind,
it is natural to interpret such $\varphi$ as a transverse fluctuation of Higgs field around $\varphi_0$
 in the current case. Moreover,
 \begin{itemize}
  \item[{\Large $\cdot$}] {\it
   Note that
     ${\cal C}({\cal A}_{\varphi_0})$
  	   is an ${\cal O}_X^{\,\Bbb C}$-subalgebra of ${\cal O}_X^{A\!z}$
      of lower rank as ${\cal O}_X^{\,\Bbb C}$-modules.
   It is in this sense that
     $\varphi_0$ ``breaks the original symmetry" in the current context.}
 \end{itemize}
All ${\cal O}_X^{A\!z}$-modules  can be regarded
 also as ${\cal C}({\cal A}_{\varphi_0})$-modules naturally.
 
As for the generation of masses by $\varphi_0$,
 as long as there are terms 
   ${\cal L}_{\Yukawascriptsize}$ and ${\cal L}_{\Higgsscriptsize}$
 in the action functional for fields on $X^{\!A\!z}$
   that are parallel to like terms of the same notation in the Glashow-Weinberg-Salam model in Sec.~5.2.1,
 $\varphi_0$ would play such role.

\bigskip

\noindent
{\it Case $(b)$}:  {\it Section of fundamental module as Higgs field}

\medskip

\noindent
Let $\xi_0$ be a nowhere-zero $C^k$-section of $E$.
Since fiberwise the $\Aut(E)$-orbit of $\xi_0$ is open in $E$,
 there is no transverse fluctuation of the Higgs field in the current case.
(However, see Remark 5.2.2.3.)
Nevertheless,
 $\xi_0$,  regarded now as the VEV of the Higgs field in the current case,
  remains to have effect both on symmetry-breaking and on generation of masses, as we now explain.

Define the {\it null-sheaf} ${\cal N}(\xi_0)$ of $\xi_0$ in ${\cal O}_X^{A\!z}$
 to be the sheaf on $X$ defined by
 $$
  {\cal N}(\xi_0)(U)\;
        :=\; \left\{ a \in {\cal O}_X^{A\!z}(U)\,|\, a\cdot \xi_0|_U\;=\;0\,  \right\}
 $$
 for open sets $U$ of $X$.
Then, ${\cal N}(\xi_0)$ is an ${\cal O}_X^{\,\Bbb C}$-module that is multiplicatively closed.
But it is not an ${\cal O}_X^{\,\Bbb C}$-algebra since it has no unit element for the multiplication.
{To} remedy this, consider
  $$
    {\cal N}^+(\xi_0) \; :=\; {\cal N}(\xi_0) + {\cal O}_X^{\,\Bbb C}\;
	  \subset\;  {\cal O}_X^{A\!z}
  $$
This is now an ${\cal O}_X^{\,\Bbb C}$-subalgebra of ${\cal O}_X^{A\!z}$.
Note that ${\cal N}(\xi_0)\cap {\cal O}_X^{\,\Bbb C}=0$	in ${\cal O}_X^{A\!z}$;
  thus, ${\cal N}^+(\xi_0)\simeq {\cal N}(\xi_0)\oplus {\cal O}_X^{\,\Bbb C}$.

Consider
 \begin{itemize}
  \item[{\Large $\cdot$}]
    the class of  $C^k$-maps
      $\varphi:(X^{\!A\!z},{\cal E})\rightarrow Y$
        with the constraint that $\Image\varphi^{\sharp}\subset {\cal N}^+(\xi_0)$.
 \end{itemize}
A $C^k$-map $\varphi$ in this class has the property
  that  ${\cal A}_{\varphi}\subset {\cal N}^+(\xi_0)$ as well.
It follows that
 such $\varphi^{\sharp}:{\cal O}_Y\rightarrow {\cal O}_X^{A\!z}$
 induces an equivalence class of gluing systems of ring-homomorphisms
 $$
    \underline{\varphi}^{\sharp}\;:\; {\cal O}_Y\;\longrightarrow\;   {\cal N}^+(\xi_0)
 $$
 and, equivalently,  $\varphi$ induces a map
 $$
   \underline{\varphi}\;:\; (X,{\cal N}^+(\xi_0))\;\longrightarrow\; Y\,.
 $$
By construction,
 $\varphi$ factors through $\underline{\varphi}$ as indicated in the following commutative diagram:
 $$
  \xymatrix{
   &  (X, {\cal O}_X^{A\!z},{\cal E})\hspace{1em}
            \ar@<-3ex>@{->>}[d]      \ar[rr]^-{\varphi}                                                       && Y  \\
   &  (X, {\cal N}^+(\xi_0), {\cal E} )\ar[rru]_-{\underline{\varphi}}   && &,
   }
 $$
 where
  \begin{itemize}
   \item[{\Large $\cdot$}]
    $(X,{\cal N}^+(\xi_0))$ is a noncommutative space with the underlying topology $X$
	and the structure sheaf ${\cal N}^+(\xi_0)$;
	
   \item[{\Large $\cdot$}]	
    ${\cal E}$ is now regarded as the ${\cal N}^+(\xi_0)$-module via the built-in inclusion
    ${\cal N}^+(\xi_0)\subset {\cal O}_X^{A\!z}$;
	
   \item[{\Large $\cdot$}] 	
    the surjection $(X,{\cal O}_X^{A\!z})\rightarrow (X,{\cal N}^+(\xi_0))$
	 is defined by the built-in inclusion ${\cal N}^+(\xi_0)\hookrightarrow {\cal O}_X^{A\!z}$.
  \end{itemize}
Note that, also by construction,
 \begin{itemize}
  \item[{\Large $\cdot$}] {\it
   $\varphi$ and $\underline{\varphi}$ have identical surrogates, i.e.,
   $X_{\varphi}=X_{\underline{\varphi}}$,    and
   identical  $C^k$-maps
     $f_{\varphi}:X_{\varphi}\rightarrow Y$ and
	 $f_{\underline{\varphi}}:X_{\underline{\varphi}}\rightarrow Y$
	   under the above identification of surrogates.}
 \end{itemize}
 
On  the other hand,
 \begin{itemize}
  \item[{\Large $\cdot$}] {\it
    Since $\xi_0$ is nowhere-zero,
     ${\cal N}^+(\xi_0)$ is an ${\cal O}_X^{\,\Bbb C}$-subalgebra of ${\cal O}_X^{A\!z}$
     of lower rank as ${\cal O}_X^{\,\Bbb C}$-modules.
   It is in this sense that
     $\xi_0$ ``breaks the original symmetry" in the current context.}
 \end{itemize}
 
As for the generation of masses by $\xi_0$, the situation is the same as in the previous case: 
 as long as there are terms 
   ${\cal L}_{\Yukawascriptsize}$ and ${\cal L}_{\Higgsscriptsize}$
 in the action functional for fields on $X^{\!A\!z}$
   that are parallel to like terms of the same notation in the Glashow-Weinberg-Salam model in Sec.~5.2.1,
 such role of $\xi_0$ remain intact.
  
\bigskip
 
\begin{ssremark} $[\,$reduction of structure group from $\GL(r,{\Bbb C})$ to $U(r)$$\,]$. {\rm
 The notion of `transverse fluctuation' of the Higgs field in the current Case (b) will be postponed to the future.
 It is better addressed after a Hermitian metric on the Chan-Paton bundle $E$ is introduced
   and the details of the reduction
      $M_{r\times r}({\Bbb C})\Rightarrow \GL(n,{\Bbb C})\Rightarrow U(r)$
      in line with our setting are understood.
 Cf.~Remark~1.1.
 }\end{ssremark}

\bigskip

The following toy model  of a dynamical Azumaya/matrix brane with fermions illustrates
 the further issue of the generation of masses from the VEV of a Higgs field.
It is a modification (by a potential) of a truncated and simplified version of an action functional
 motivated by Matrix Theory in the sense of
  Thomas Banks, Willy Fischler, Stephen Shenker, and Leonard Susskind  [B-F-S-S].
See, for example, [T-vR1], [T-vR2], [T-vR3] for some related study in curved space-time.

\bigskip
  
\begin{ssexample} {\bf [Higgs mechanism on Azumaya/matrix brane with fermion].} {\rm
 Let
  \begin{itemize}
   \item[{\Large $\cdot$}]
    $X={\Bbb R}^1$ be the real line as a $C^{\infty}$-manifold with coordinate $t$, and\\
	${\cal O}_X := {\cal O}_{{\Bbb R}^1}$
	   be the structure sheaf of $C^{\infty}$-functions on ${\Bbb R}^1$;
	
   \item[{\Large $\cdot$}]	
    ${\cal E}\simeq {\cal O}_{{\Bbb R}^1}\otimes_{\Bbb R}{\Bbb C}^{\oplus r}$
	  be a free ${\cal O}_{{\Bbb R}^1}^{\,\Bbb C}$-module of rank $r$;
    for concreteness, we assume that ${\cal E}$ is trivialized;
	
   \item[{\Large $\cdot$}]	
    $(X^{\!A\!z}, {\cal E})
	    := ({\Bbb R}^{1, A\!z},{\cal E})
        := ({\Bbb R}^1,
		         {\cal O}_{{\Bbb R}^1}^{A\!z}
				           := \Endsheaf_{{\cal O}_{{\Bbb R}^1}^{\,\Bbb C}}({\cal E}),
				 {\cal E})$
    be an Azumaya/matrix real line with a fundamental module  and
   
   \item[{\Large $\cdot$}]
   $Y={\Bbb M}^{(d-1)+1}$ be the $d$-dimensional Minkowski space-time,
    as a $C^{\infty}$-manifold with coordinates $(y^a)_a :=(y^0, y^1,\,\cdots\,,y^{d-1})$
	and a flat Lorentzian metric\\
	$ds^2= -(dy^0)^2+(dy^1)^2+\,\cdots\,+ (dy^{d-1})^2$.
  \end{itemize}	
  
 Consider a $1$-dimensional quantum field theory on the matrix real line ${\Bbb R}^{1,A\!z}$
  with fields:
  \begin{itemize}
   \item[{\Large $\cdot$}]  {\it Bosonic}$\,$:
    \begin{itemize}
	 \item[-]
	  {\it $C^{\infty}$-maps} $\,\varphi: {\Bbb R}^{1,A\!z}\rightarrow {\Bbb M}^{(d-1)+1}\,$	  
	  from the matrix real line to the Minkowski space-time;
	 recall that $\varphi$ is defined by a $C^{\infty}$-admissible ring-homomorphism
	 $$
	  \begin{array}{ccccccc}
	   \hspace{3em}
	      & \varphi^{\sharp} & : & C^{\infty}({\Bbb M}^{(d-1)+1})
	      & \longrightarrow  & M_{r\times r}(C^{\infty}({\Bbb R}^1))\\[.6ex]
       &&& y^a  & \longmapsto  & Y^a(t)\,,   & a=0,\, 1,\,\cdots\,, d-1\,;
	  \end{array}
	 $$
	 	
	 \item[-]
	 {\it gauge fields}  $\,A(t)\, dt\,,  A(t)\in M_{r\times r}(C^{\infty}({\Bbb R}^1))\,$,
 	   on the fundamental module ${\cal E}$;
     recall that its induced connection on ${\cal O}_{{\Bbb R}^1}^{A\!z}$
    	is simply the inner derivation $[A(t),\,\cdot\,]\,dt$    and
	 note that  $A$ is nondynamical since its curvature vanishes;
	\end{itemize}
   
   \item[{\Large $\cdot$}] {\it Fermionic}$\,$:
     \begin{itemize}
	  \item[-]
       {\it ${\cal O}_{{\Bbb R}^1}^{A\!z}$-valued spinor fields} $\,\Theta^{\alpha}
	                     \in  ({\cal O}_{{\Bbb R}^1}^{A\!z}
		                                   \otimes_{{\cal O}_{{\Bbb R}^1}}  {\cal S}) ({\Bbb R}^1)
	                  =  M_{r\times r}(C^{\infty}(S))\,$
		  on ${\Bbb R}^1$,\\    $\alpha=1,\,\cdots\,, N$,
	     where
         \begin{itemize}	
	      \item[]
         {\Large $\cdot$}\hspace{.5ex}	
	         $S$ is the spinor bundle on ${\Bbb R}^1$ (with flat metric $(dt)^2$) of rank $1$, \\
	     {\Large $\cdot$}\hspace{.5ex}
		     ${\cal S}$ is the sheaf of ${\cal O}_{{\Bbb R}^1}$-modules associated to $S$,  and \\
		 {\Large $\cdot$}\hspace{.5ex}
	         $N$ is the dimension of a spinor representation of $\SO(d-1,1)$;
	     \end{itemize}
	   note that
  	    since ${\cal S}\simeq {\cal O}_{{\Bbb R}^1}$
	     as ${\cal O}_{{\Bbb R}^1}$-modules,
		 $M_{r\times r}(C^{\infty}(S))\simeq M_{r\times r}(C^{\infty}({\Bbb R}^1))$;	
	 \end{itemize}	
  \end{itemize}
 and Lagrangian: (Up-low-repeated-dummy-index summation rule is assumed.)
 $$
   \begin{array}{ccl}
    {\cal L}(\varphi, \Theta, A)
       &   :=     &   T_0\,  \Tr   \left\{\rule{0ex}{1em}\right.
                                   - D_tY_aD_tY^a\,  +\,   V(Y^0, Y^1,\,\cdots\,, Y^{d-1}) \\[1.2ex]
     &&	 \hspace{12em}
             	 +\,   \Theta_{\alpha}D_t\Theta^{\alpha}\,
                  -\,  c_0 \Theta^{\alpha}\gamma_{a,\alpha\beta}
						                          [\, Y^a, \Theta^{\beta}\,]                 \left.\rule{0ex}{1em}\right\}\,,
   \end{array}												
 $$
 where
  \begin{itemize}
   \item[{\Large $\cdot$}]
    $T_0$ and $c_0$ are constants (depending on the string tension and the string coupling constant
	 when the setting is fitted into string theory);
	
   \item[{\Large $\cdot$}]	
    $D_tY^a = \partial_t Y^a +  [A, Y^a]$  and
    $D_t\Theta^{\alpha}= \partial_t\Theta^{\alpha} + [A,\Theta^{\alpha}]$
	are covariant derivatives of $Y^a$ and $\Theta^{\alpha}$ respectively;
	
   \item[{\Large $\cdot$}] 	
    $\gamma_a$, $a=0,\,1,\,\cdots\,, d-1$, are the $\gamma$-matrices for ${\Bbb M}^{(d-1)+1}$.
  \end{itemize}
 As given, this quantum field theory on ${\Bbb R}^{1, A\!z}$ has massless fermions.
 
 In comparison to  the Lagrangian density for the Glashow-Weinberg-Salam model in Sec.~5.2.1,
  one has immediately that
  $$
    \begin{array}{lclcclcl  }
	{\cal L}_{\fermionscriptsize}
	     & = &    \Theta_{\alpha}D_t\Theta^{\alpha}\,,
       &&&  {\cal L}_{\Higgsscriptsize}\;
	                & =   &   - D_tY_aD_tY^a\,  +\,   V(Y^0, Y^1,\,\cdots\,, Y^{d-1})\,, \\[1.2ex]
  {\cal L}_{\gaugescriptsize}  & =  & 0\,,
	   &&& {\cal L}_{\Yukawascriptsize}
	                &  =  &   -\,  c_0 \Theta^{\alpha}\gamma_{a,\alpha\beta}
						                                                                   [\, Y^a, \Theta^{\beta}\,]\,.
   \end{array}		
  $$
 With the $C^{\infty}$-maps
   $\varphi: {\Bbb R}^{1,A\!z}\rightarrow {\Bbb M}^{(d-1)+1}$ serving as Higgs fields
    of the current toy model  and
   the Higgs mechanism in the Glashow-Weinberg-Salam model and
    the discussion of open strings and D-branes in [Po: vol.\ I: Sec.\ 8.6 and Sec.\ 8.7] of Polchinski in mind,
  one expects thus:
  
 \bigskip
  
 \noindent   
 {\bf Claim 5.2.2.4.1. [mass generation of fermion].} {\it
  A vacuum expectation value VEV $\varphi_0$  of $\varphi$ may generate mass for some of the fermions
    in the model.
  Furthermore,
    the mass the VEV $\varphi_0$ generates for fermions in the model may depend on the distance
	between the connected or irreducible components of the image brane $\varphi_0({\Bbb R}^{1,A\!z})$
	in the space-time ${\Bbb M}^{(d-1)+1}$.
 } 
 
 \bigskip
   
 {To} justify this, observe that

 \bigskip
 
 \noindent 
  {\bf Lemma 5.2.2.4.2. [commutator with diagonal matrix].} {\it
   Let
    $r=r_1+\,\cdots\,+r_l$ be a positive-integer decomposition of $r$,
    $\Id_{r_i\times r_i}$ be the $r_i\times r_i$ identity matrix,  and
   	$$
      M^{(r_1,\,\cdots\,, r_l)}_{(\lambda_1,\,\cdots\,\lambda_l)}\;  :=\;
	   \left[
	     \begin{array}{cccc}
		  \lambda_1\,\Id_{r_1\times r_1} \\
		  & \ddots   \\
		  && \ddots \\
		  &&& \lambda_l\,\Id_{r_l\times r_l}		
		 \end{array}
 	   \right]\,
    $$	
    with the $i$-th block in the diagonal being $\lambda_i\cdot\Id_{r_i\times r_i}$	
	and all other entries being $0$.
  Let
   $$
      B=\left[ B_{ij}\right]_{ij,\, l\times l}\in M_{r\times r}({\Bbb C})
   $$
    be an $r\times r$-matrix in the $l\times l$ block-matrix form,
    where $B_{ij}\in M_{r_i\times r_j}({\Bbb C})$ is an $r_i\times r_j$-matrix.
  Then, the commutator
    $$
      [M^{(r_1,\,\cdots\,, r_l)}_{(\lambda_1,\,\cdots\,,\lambda_l)}\,,\, B]\;
	   =\;  \left[(\lambda_i-\lambda_j)\,B_{ij} \right]_{ij,\, l\times l}\,.
    $$	
   } 
 
 \bigskip
   
 Suppose now that the potential $V(\varphi):= V(Y^0, Y^1,\,\cdots\,, Y^{d-1})$ is chosen so that
  it takes a VEV at a $C^{\infty}$-map
   $\varphi_0: {\Bbb R}^{1,A\!z}\rightarrow {\Bbb M}^{(d-1)+1}$
   defined by
   $$
	 \begin{array}{ccccccc}
	   \hspace{3em}
	      & \varphi_0^{\sharp} & : & C^{\infty}({\Bbb M}^{(d-1)+1})
	      & \longrightarrow  & M_{r\times r}(C^{\infty}({\Bbb R}^1))\\[.6ex]
       &&& y^a  & \longmapsto  & Y_0^a(t)\,,   & a=0,\, 1,\,\cdots\,, d-1
	 \end{array}
   $$
 with
  $$
    Y_0^0(t)\;=\;  t\cdot\Id_{r\times r}     \hspace{2em}\mbox{and}\hspace{2em}
    Y_0^a(t)\;=\;  M^{(r_1,\,\cdots\,, r_l )}_{(\lambda^a_1,\,\cdots\,,\lambda^a_l)}
	\hspace{1em}\mbox{for  $a=1,\,\cdots\,, d-1$}\,,
  $$
 for some positive-integer decomposition $r=r_1+\,\cdots\,+ r_l$ of $r$ independent of $a$.
 Recall from [L-Y3: Sec.\ 3] (D(11.1)) that
  all the $\lambda^a_i$ must be real and
  they correspond to the $y^a$-coordinate of the $i$-th components of $\varphi_0({\Bbb R}^{1,A\!z})$
  in ${\Bbb M}^{(d-1)+1}$.
 As elements in $M_{r\times r}(C^{\infty}(S))$,
 express
   $$
     \Theta^{\beta}\;=\;  [\Theta^{\beta}_{ij}]_{ij,\, l\times l}
   $$
  in the $\l \times l$ block-matrix form.
 Then,
  after
    formally expanding  $\varphi$ around the VEV $\varphi_0$,
    rewriting the Lagrangian ${\cal L}$ as done in the Glashow-Weinberg-Salam model,  and
	applying Lemma~5.2.2.4.2,
 one concludes that
 $$
  {\cal L}_{\Yukawascriptsize}\;
    =\;   -\, c_0\, \ \Theta^{\alpha}\gamma_{a,\alpha\beta}
	                                   [(\lambda^a_i-\lambda^a_j) \Theta^{\beta}_{ij}]_{ij,\, l\times l}\,
               +\, \mbox{(higher order terms)}\,.									   									   
 $$
 Assume that there is at least one $a$ such that $\lambda^a_i$ are distinct for all $i$.
 Then one has now a nontrivial mass-matrix for fermions that depends only on the set
  $$
    \{\lambda^a_i-\lambda_j^a\,|\, a=1,\,\cdots\,, d-1\,; 1\le i,j\le l\}
  $$
  that describes  the relative position/distance of connected/irreducible components of
  $\varphi_0({\Bbb R}^{1,A\!z})$ in ${\Bbb M}^{(d-1)+1}$.
 This justifies the claim and concludes the example.
 Cf.~{\sc Figure}~5-2-2-4-1.
 %
 %
 \begin{figure}[htbp]
  \bigskip
   \centering
   \includegraphics[width=0.80\textwidth]{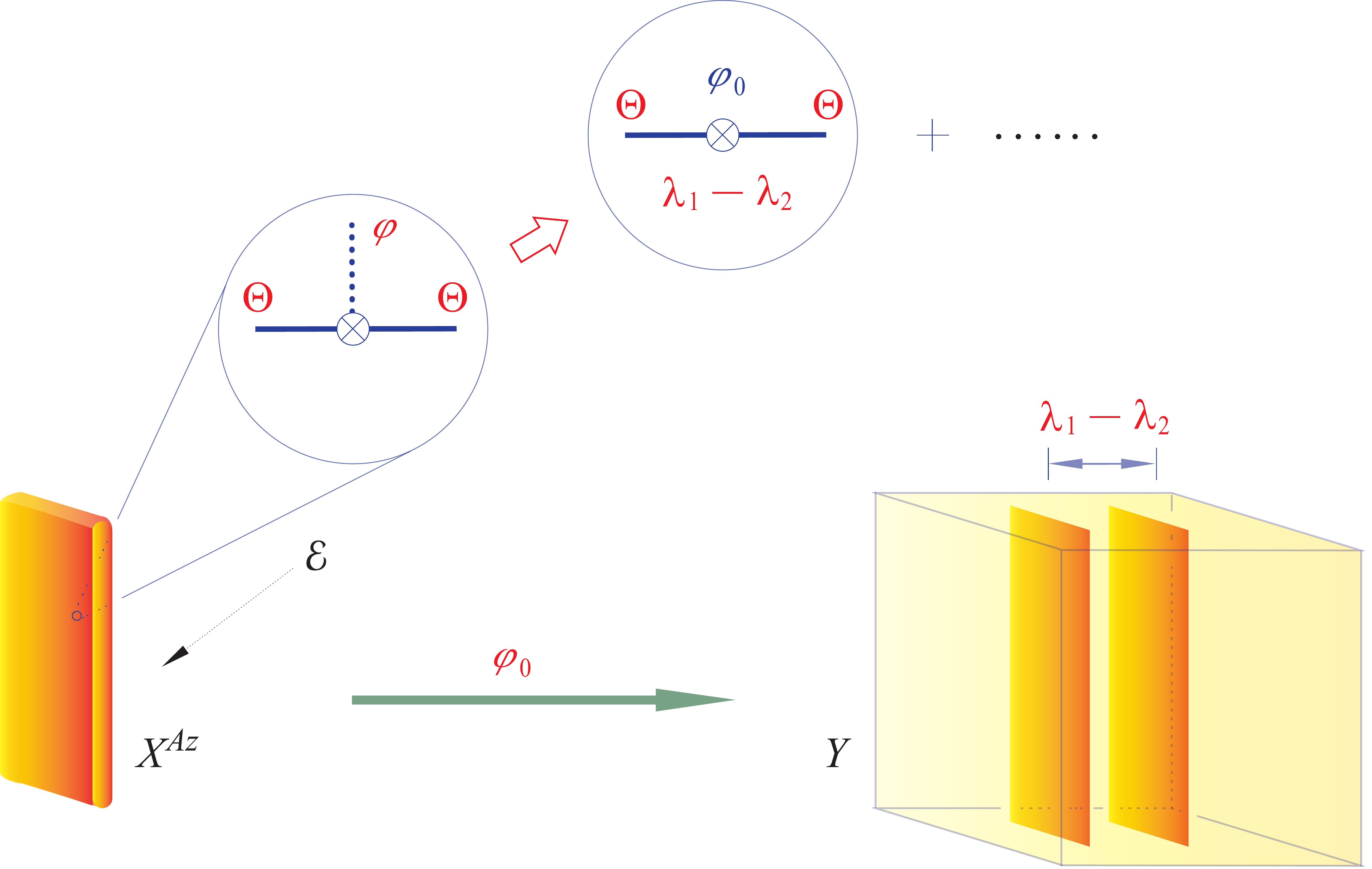}
 
   \bigskip
   \bigskip
  \centerline{\parbox{13cm}{\small\baselineskip 12pt
   {\sc Figure}~5-2-2-4-1.
   Generation of mass for fermions on the D-brane world-volume $(X^{\!A\!z},{\cal E})$, 
       a matrix manifold with a fundamental module,  
     through the Higgs mechanism  that takes $C^{\infty}$-maps $\varphi:X^{\!A\!z}\rightarrow Y$
	  as Higgs field. 
   In the {\sc Figure}, this process is indicated by 
       the arrow $\Rightarrow$ that transforms a $\Theta \varphi \Theta$ Feynman diagram 	
       to a series of diagrams with the lowest-order term the propagator diagram for fermions.
   The geometry of the image configuraion $\Image(\varphi_0)\subset Y$ of $X^{\!A\!z}$
        under a VEV $\varphi_0$
	  determines the mass of fermions $\Theta$ through the Yukawa coupling terms in the actional functional,  	  
	   as indicated by the relevant Feynman diagram, 
       with fermions in solid line $\,\mbox{-----}\,$ and 
	           Higgs fields in dashed line $\,\raisebox{.5ex}{......}\,$.
   Such Feynman diagrams may be thought of as reflecting the scattering in $X^{\!A\!z}$ 
     of particles associated to these fields on $X^{\!A\!z}$.			   
   The factor $(\lambda_1-\lambda_2)$ for fermion mass terms after Higgsing 
	    is only meant to be schematic, 
	   indicating its dependence on the distance ``$\lambda_1-\lambda_2$" of components 
	   of the image $\Image(\varphi_0)$ of $X^{\!A\!z}$ in $Y$. 			   
   }}
 \end{figure}
}\end{ssexample}

\bigskip

Further details of the Higgs mechanism on D-branes in our setting should be re-picked up
after the details of the action functional is understood; cf.\ Sec.~6.

\bigskip

\section{Where we are, and some more new directions}

Recall the following guiding question (cf.\ [L-Y1: Sec.~2.2] (D(1))):
 \begin{itemize}
  \item[{\bf Q.}]
  {\it What is a D-brane intrinsically?}
 \end{itemize}
 that initiated our D-brane project.
Following the line of Grothendieck's theory of schemes for modern algebraic geometry,
[L-Y1] (D(1)) provided a proto-typical setting for dynamical D-branes
  in the common realm of string theory and algebraic geometry,
  as {\it maps/morphisms from an Azumaya/matrix scheme with a fundamental module}.
Its equivalent settings were realized in [L-L-S-Y] (D(2), with Si Li and Ruifang Song).
Seven years later,
 [L-Y3] (D(11.1)) and the current note D(11.2) together
 brought into play the notion of differentiable rings
    from synthetic differential geometry and algebraic geometry over differentiable rings,
 and extends such settings to the common realm of string theory and differential/symplectic geometry,
  as {\it differentiable maps from an Azumaya/matrix manifold with a fundamental module}
  with similar equivalent settings.
Before moving on,
 it is instructive to pause here, as a conclusion of this note,
   with a reflection on where we are now --- in comparison with how string theory began ---  and
   a sample list of new themes/directions naturally brought out from the study.

\bigskip
   
First, another guiding question:
 \begin{itemize}
  \item[{\bf Q.}]
   {\it How does string theory begin? }
 \end{itemize}
Physically and historically, it began with the attempt to understand hadrons
 (particles that interacts through the strong interaction).
However,  as you open any textbook on string theory,
 another answer from another aspect may immediately come to you:
 \begin{itemize}
  \item[{\bf A.}]
   {\it Mathematically},
   string theory begins with the notion of a {\it differentiable map from a string or the world-sheet of a string}
   (open or closed, with or without world-sheet fermions) {\it to a space-time}.
 \end{itemize}
For example, 
  [B-B-Sc: Sec.\ 2.2], [G-S-W: vol.\ 1: Sec.\ 1.3.2], [Joh: Sec.\ 2.2],  [Po2: vol.\ I: Sec.\ 1.2], 
  [Zw: Chap.\ 6].
Indeed,
 replacing `string' (resp.\ `world-sheet') with any physical object (resp.\ `world-volume'),
  the same answer should work for any dynamical object moving in a space-time;
 in particular, D-branes.
And that's what we just completed and that's exactly where we are:
 \begin{itemize}
  \item[{\Large $\cdot$}]
  {\it After [L-Y3] (D(11.1)) and the current note D(11.2),   we are now
   at the beginning/entrance of a theory of D-branes 
    --- purely bosonic or with fermionic fields and supersymmetry --- 
    as dynamical objects moving in a space-time.}
   {\sc Figure}~6-1.
 \end{itemize}
 %
%
\begin{figure}[htbp]
 \bigskip
  \centering
  \includegraphics[width=0.80\textwidth]{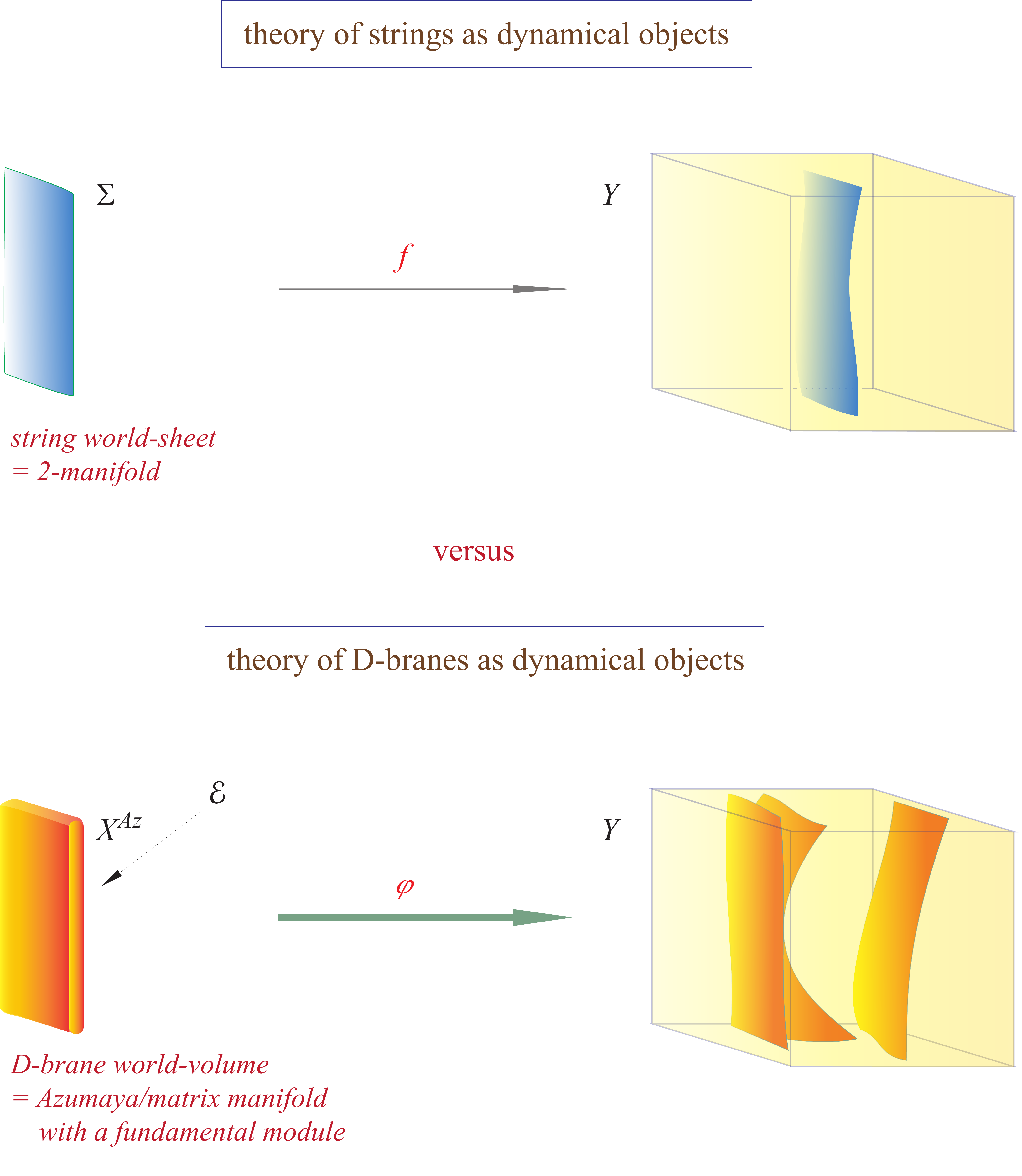}
  
  \bigskip
  \bigskip
 \centerline{\parbox{13cm}{\small\baselineskip 12pt
  {\sc Figure}~6-1.
  The mathematical starting point of string theory with string as a dynamical object moving in a space-time
        (cf.\ [G-S-W: Sec.\ 1.3, {\sc Figure} 1.3] of Green, Schwarz, and Witten)
   vs.\ the mathematical starting point of D-brane theory with D-brane as a dynamical object
  	      moving in a space-time: 
  The former begins with the notion of  `differentiable maps $f:\Sigma\rightarrow Y$ 
      from a string world-sheet to the space-time'
    while the latter begins with the notion of  
	 `differentiable maps $\varphi:(X^{\!A\!z},{\cal E})\rightarrow Y$
	   from a matrix manifold (i.e.\ the D-brane world-volume) with a fundamental module
      to the space-time'.  
   {\it Unlike } the string world-sheet	$\Sigma$, $X^{\!A\!z}$ carries a matrix-type ``noncommutative cloud"
       over its underlying topology. 
	 Under a differentiable map $\varphi$ 
	     as defined in [L-Y3: Definition~5.3.1.5] (D(11.1), 
               		 (cf.\ Definition~4.2.1.3 of the current note)
       the image $\varphi(X^{\!A\!z})$ can behavior in a complicated way.	
     In particular, it could be disconnected or carry some nilpotent fuzzy structure.	  
	 See also {\sc Figure}~5-2-0-1.
  }}
\end{figure}

\bigskip
 
With the above comparison to the history of string theory in mind,
in [L-Y3] (D(11.1)) we bring out
  a sample of five new directions all related to or motivated by D-branes in string theory
  that the notion of differentiable maps from matrix manifolds may play a role.
Some more immediate new directions include:
 \begin{itemize}
  \item[(1)]
   The {\it Dirac-Born-Infeld term}, the {\it Chern-Simons term},
    as well as any other term,  and their supersymmetric generalization
	in the full action functional for coincident D-branes
	from the aspect of functionals for maps from a matrix manifold
	         with various bosonic and fermionic fields thereupon.
        
  \item[(2)]
   {\it Synthetic/$\,C^k$-algebraic symplectic geometry}.

  \item[(3)]
   {\it Synthetic/$\,C^k$-algebraic calibrated geometry}.
       
  \item[(4)]
   A {\it new matrix theory} based on complex matrices of real eigenvalues.
 \end{itemize}
{From} the {\it string-theory point of view},
 \begin{itemize}
  \item[{\Large $\cdot$}]
  Theme (1) is {\it the next} guiding theme:
  Only when one is able to give a string-theory-compatible action functional
   on differentiable maps from an Azumaya/matric manifold with a fundamental module to a space-time
   can one begin to address physics of D-branes as fundamental objects in their own right in string theory .
 \end{itemize}
{From} the {\it mathematical point of view},
 \begin{itemize}
  \item[{\Large $\cdot$}]
  Theme (2) and Theme (3) have been missing in symplectic/calibrated geometry
    when Lagrangian or calibrated submanifolds and their deformations/collidings were studied.
  Whatever the reason they are overlooked,
   we now provide a motivation to study them
     from the new aspect of dynamical D-branes, cf.\ [L-Y3: Sec.\ 7.2] (D(11.1)).
 This should bring  the study of Lagrangian/calibrated submanifolds (possibly supporting a decorated sheaf)
  to a footing closer to that of Hilbert- or Quot-schemes or moduli of coherent or (semi-)stable sheaves
  in algebraic geometry.
 \end{itemize}
{From} {\it both mathematical and physical aspects},
 \begin{itemize}
  \item[{\Large $\cdot$}]
   Theme (4), the new matrix theory ---
       as a theory for differentiable maps from an Azumaya/matrix point  with a fundamental module to a space ---
	 could provide one with a starting point before attacking questions for general Azumaya/matrix manifolds.
 \end{itemize}

\bigskip
 
With the notion/framework of differentiable maps
 from a matrix manifold-or-supermanifold to a space-time (or superspace-time) in place,
the stage has just been set.

\newpage

\vspace{12em}
\baselineskip 13pt
{\footnotesize

\vspace{1em}

\noindent
chienhao.liu@gmail.com, chienliu@math.harvard.edu; \\
yau@math.harvard.edu

}

\end{document}